\documentclass[aps,prd,twocolumn,superscriptaddress,preprintnumbers,floatfix,nofootinbib,notitlepage,showkeys,longbibliography]{revtex4-1}
\usepackage{multirow}
\usepackage{siunitx}

\usepackage{graphicx,times}
\usepackage{latexsym}
\usepackage{mathtools}
\usepackage{booktabs}

\usepackage{amsmath,amssymb,amsbsy,amsfonts,bbm}
\usepackage{array}
\usepackage{bm}
\usepackage{graphics}
\usepackage{mathrsfs}
\usepackage{xcolor}
\usepackage{cancel}
\usepackage[normalem]{ulem}
\usepackage{makecell}
\usepackage{hyperref}
\usepackage[capitalise]{cleveref}

\newcommand{\<}{\langle}
\renewcommand{\>}{\rangle}
\newcommand{\tr}{\operatorname{tr}}
\newcommand{\Tr}{\operatorname{Tr}}
\renewcommand{\d}{\mathrm{d}}
\newcommand{\e}{\operatorname{e}}
\newcommand{\ddx}{\frac{\d}{\d x}}

\newcommand{\sch}[1]{}

\newcommand\TopRule{\Xhline{0.08em}}

\newcommand\MidRule{\Xhline{0.03em}}
\newcommand\BotRule{\Xhline{0.08em}}

\usepackage{amsthm}
\theoremstyle{plain}
\newtheorem*{thm*}{\sf THEOREM}

\begin{document}

\title{Hamiltonian models of lattice fermions solvable by the meron-cluster algorithm}

\author{Hanqing Liu}
\author{Shailesh Chandrasekharan}
\affiliation{Department of Physics, Box 90305, Duke University, Durham, NC 27708, USA}
\author{Ribhu K. Kaul}
\affiliation{Department of Physics \& Astronomy, University of Kentucky, Lexington, KY 40506, USA}

\begin{abstract}
We introduce a half-filled Hamiltonian of spin-half lattice fermions that can be studied with the efficient meron-cluster algorithm in any dimension. As with the usual bipartite half-filled Hubbard models, the na\"ive $U(2)$ symmetry is enhanced to $SO(4)$. On the other hand our model has a novel spin-charge flip ${\mathbb Z}^C_2$ symmetry which is an important ingredient of free massless fermions. In this work we focus on one spatial dimension, and show that our model can be viewed as a lattice-regularized two-flavor chiral-mass Gross-Neveu model. Our model remains solvable in the presence of the Hubbard coupling $U$, which maps to a combination of Gross-Neveu and Thirring couplings in one dimension. Using the meron-cluster algorithm we find that the ground state of our model is a valence bond solid when $U=0$. From our field theory analysis, we argue that the valence bond solid forms inevitably because of an interesting frustration between spin and charge sectors in the renormalization group flow enforced by the ${\mathbb Z}^C_2$ symmetry. This state spontaneously breaks translation symmetry by one lattice unit, which can be identified with a $\mathbb{Z}_2^\chi$ chiral symmetry in the continuum. 
We show that increasing $U$ induces a quantum phase transition to a critical phase described by the $SU(2)_1$ Wess-Zumino-Witten theory. The quantum critical point between these two phases is known to exhibit a novel symmetry enhancement between spin and dimer. Here we verify the scaling relations of these correlation functions near the critical point numerically. Our study opens up the exciting possibility of numerical access to similar novel phase transitions in higher dimensions in fermionic lattice models using the meron-cluster algorithm.
\end{abstract}

\maketitle

\section{Introduction}


The connections between lattice models of quantum many body physics and continuum field theories remains a forefront topic of research in both high energy and condensed matter physics. In high energy physics one often starts with a continuum field theory, and then introduces a lattice to enable non-perturbative numerical computations. Even though the lattice is viewed as a calculational artifice, it has lead to many profound physical insights. In condensed matter physics on the other hand, one starts with the lattice, which is derived from the physical crystal structure and the connection to continuum field theories emerges in the long distance physics close to a critical point. Despite these contrasting motivations, there are recurring common themes across the two approaches leading to a wonderful synergistic exchange of ideas.    

Since our work here may be of interest to both communities, we explain our motivations from both viewpoints in turn. In particle physics, there is a great interest in mechanisms through which fermions acquire masses. While fermion masses are usually thought to arise from fermion bilinear operators, it is well known that strong four-fermion couplings can also generate masses through spontaneous symmetry breaking \cite{Rosenstein:1990nm,ZinnJustin:1991yn}. The role of four-fermion couplings is still poorly understood since such couplings are irrelevant perturbatively at the free fermion fixed point in two or more spatial dimensions and have to be of finite strength when they generate a mass, requiring a non-perturbative analysis. From a Lorentz-symmetry point of view the interactions constructed with scalar or pseudo-scalar fermion bilinears are referred to as Gross-Neveu (GN) couplings \cite{Gross:1974jv}, while those constructed with vector or pseudo-vector fermion bilinears are called Thirring couplings \cite{Thirring:1958in}. These relativistic four-fermion field theories can be studied through the following Lagrangian
\begin{equation}
\label{eq:contL}
{\cal L}= \overline{\chi}_\alpha\gamma^\mu\partial_\mu\chi_\alpha + {\cal L}_{\rm GN} + {\cal L}_{\rm Thirring},
\end{equation}
where we use the Euclidean signature and we will keep this convention through the paper. The case with two flavors of Dirac fermions that is relevant to our paper is discussed in \cref{app:lagrangian}.
The fixed point structure of renormalization group (RG) flows in the space of these couplings is rich and reveals that the GN and Thirring couplings can have very different physics at strong couplings depending on the number of fermion flavors especially in $2+1$ dimensions \cite{Gies:2010st,Janssen:2012pq,PhysRevD.92.085046}.

In condensed matter physics one usually works with Hamiltonians instead of Lagrangians, which are motivated more naturally by the physics of the underlying material. For many important systems, such as one dimensional metals and two and three dimensional semi-metals, the long distance physics is described by relativistic four-fermion field theories. These Hamiltonians usually have two parts: an electron hopping term that in certain cases can result in massless Dirac fermions and an electron-electron interaction term that is modeled by a four-fermion interaction. These kind of models are typified by the iconic Hubbard model, 
\begin{equation}
\label{eq:hubb}
    H =- \sum_{\langle ij \rangle}\left ( t_{ij} c_{i\alpha}^\dagger c_{j\alpha} +{\rm h.c.} \right ) + U \sum_i n_{i\uparrow}n_{i\downarrow},
\end{equation}
where $c_{i\alpha}$ destroys an electron on lattice site $i$ with spin $\alpha$ and $n_{i\alpha}= c_{i\alpha}^\dagger c_{i\alpha}$. While the choice of the matrix elements $t_{ij}$ depends on the material to be modeled, by now many physically motivated cases give Dirac fermions whose masslessness is protected by lattice symmetries or topology~\cite{vafek2014:arcmp}. The topic of interest is how symmetries are spontaneously broken at strong interactions that covert the semi-metal to an insulator and the nature of the quantum critical point~\cite{herbut2006:grprl}.
These models, although motivated from the physics of electrons in crystals, are natural Hamiltonian discretizations of the four-fermion mass generation problem in \cref{eq:contL}. The complexities of the RG flows of the continuum theory manifest themselves here in rich phase diagrams that are currently poorly understood, making this an intensely studied area of research.

Generally speaking the only available unbiased  method to study such strongly coupled four-fermion lattice models is the Monte Carlo (MC) method \cite{Hands:1992be,Hands:1992ck}. Given the great interest in this topic, there has been a large body of work from both communities that has led to important progress and also a number of unresolved issues. Most work found in the high energy literature has been focused on understanding phase diagrams and fermion mass generation at strong four-fermion couplings \cite{Karkkainen94,DelDebbio:1997dv,Kogut:1998rg,DelDebbio:1999he,Hands:1999id,Christofi:2006zt}. The value of the critical number of fermion flavors below which fermion bilinear condensates in the Thirring model has been an interesting quantity to compute  \cite{Wellegehausen:2017goy,Hands:2016foa,Hands:2018vrd}. Lattice calculations with a local Lagrangian typically suffer from the fermion doubling problem. Due to this problem simple lattice four-fermion couplings get mapped into a linear combination of many continuum four-fermion couplings, which makes it difficult to understand which continuum model is being explored without careful fine tuning. For example recently it was recognized that the lattice Thirring model and the lattice GN models in $2+1$ dimensions seem to flow to the same continuum theory at the critical point \cite{Chandrasekharan:2013aya}.

The field has become more exciting recently due to new interest from condensed matter physics, which was originally inspired in part by the physics of graphene. Again the focus has been on mass terms generated by four-fermion interactions on the honeycomb lattice Hubbard model. Here the motivation is to study the semi-metal to antiferromagnetic transition and its universality class ~\cite{sorella2012:absence}. While the original work focused on $SU(2)$, generalizations to $SU(N)$ have  been carried out~\cite{lang2013:sun}. Interacting spinless lattice Hamiltonian models have also been demonstrated to be free of sign problems, which then allows one to study the simplest fermion mass generation mechanism at a quantum critical point, the chiral Ising transition~\cite{huffman2014:sign,li2015:mmc}. Perhaps the most striking recent result is that strong four-fermion couplings may also create fermion masses without symmetry breaking \cite{Ayyar:2016lxq,PhysRevD.97.094502,Catterall:2015zua,Butt:2018nkn}. In $2+1$ dimensions there is growing evidence that this mechanism of symmetric fermion mass generation may be a feature of continuum quantum field theory since it appears to be connected to an exotic quantum critical point \cite{Slagle:2014vma,PhysRevD.93.081701,PhysRevX.8.011026,PhysRevB.97.125112}. In $3+1$ dimensions it has been suggested that this mechanism may be helpful to construct lattice chiral gauge theories \cite{Kikukawa:2017ngf,Wang:2018cai,Catterall:2020fep}. Another interesting scenario that has been proposed is the existence of second order quantum phase transitions between different massive fermion phases \cite{Liu:2018sww,PhysRevLett.119.197203}. While such transitions cannot easily be understood within the Landau-Ginzburg paradigm, there is speculation that they may arise through a deconfined quantum critical point with emergent gauge fields \cite{Senthil:2004aza}. It is also believed that some of them could even be driven by topological terms in the low energy bosonic theory \cite{PhysRevB.74.064405,Slagle:2014vma}. A four-fermion realization of this transition has also been proposed \cite{Li2017,Li:2019acc}, where the existence of symmetry enhancement has also been studied at multi-critical points \cite{Torres:2019vcw}.  

Despite the power of MC methods, two bottlenecks quickly arise. The first one is the sign problem that limits the type of fermion models that can be studied. Furthermore, even when sign problems are solved, most traditional MC studies of quantum critical points in fermionic systems can only be performed on rather small lattice sizes, due to the poor scaling of the computational time with system size. Given these hurdles it is clearly of great interest to find new four-fermion models that are sign problem free and can be accessed with efficient MC algorithms. Recently, alternate types of fermion MC methods have become available, which are able to reach somewhat larger lattice sizes. For example, the recently proposed fermion-bag approach \cite{PhysRevD.82.025007,Chandrasekharan:2013rpa, Huffman:2017swn} can study system sizes involving up to $10,000$ sites \cite{PhysRevD.101.074501}. The meron-cluster algorithm, a precursor to the fermion-bag approach, is even more efficient while being applicable only to a more restricted class of models \cite{Chandrasekharan:1999cm, Chandrasekharan:2002vk,PhysRevB.66.045113}. The main motivation behind our current work is to design fermionic models that can be studied on large lattices using the meron-cluster algorithm.

Interestingly the new four-fermion model we introduce in this work is not only amenable to efficient simulations, it also has some novel physical features.  In brief we demonstrate below a novel mechanism by which the valence bond solid (VBS) state which breaks translations symmetry is realized in a one-dimensional fermionic system at arbitrarily weak coupling. It is well known that the Hubbard model \cref{eq:hubb} at half filling on bipartite lattices has an explicit $SU(2)_s$ spin symmetry and a hidden $SU(2)_c$ charge symmetry~\cite{zhang1990:su2}. The hopping term of the Hubbard model has an additional a spin-charge flip symmetry $\mathbb{Z}_2^C$ under which the two $SU(2)$ symmetries are interchanged. Indeed, as is well known when $U$ is repulsive (attractive) the spin (charge) sector is favored resulting in anti-ferromagnetic (superconducting) correlations. It is interesting to ask what is the fate of the Hubbard model when the interactions added preserve $\mathbb{Z}_2^C$? Here we show by field theoretic arguments and explicit MC simulations that in one spatial dimension, when the interactions preserve $\mathbb{Z}_2^C$, the system releases the frustration between spin and charge sectors by forming a VBS. This is a novel mechanism for the formation of a VBS in the one dimensional Hubbard model. From a field theory point of view, we argue that our model can be considered as a Hamiltonian lattice regularization of the two-flavor chiral-mass GN model with a spontaneously broken $\mathbb{Z}_2^\chi$ chiral symmetry. Using the meron-cluster representation we observe that our model is also related to the $\theta=\pi$ phase of the $\mathbb{CP}^3$ model where the charge conjugation symmetry is spontaneously broken \cite{Beard:2004jr}. When $U \neq 0$ a combination of GN and Thirring couplings is introduced. These new couplings induce a quantum phase transition to a phase where $\mathbb{Z}_2^\chi$ symmetry is restored, which turns out to be the $SU(2)_1$ Wess-Zumino-Witten (WZW) model perturbed by a marginally irrelevant coupling. The universality class of the transition in our model has been studied earlier by Affleck et al. \cite{Affleck:1988px}, who has argued that at the quantum critical point the marginal coupling vanishes, enhancing the symmetry of the theory. This transition has also been well studied numerically within the context of quantum spin-half chain\cite{OKAMOTO1992433,Eggert:1996er}. 

Our paper is organized as follows. In \cref{sec:Lattice} we introduce our lattice model and explain its symmetries. In \cref{sec:continuum} we argue that our lattice model is naturally mapped into a continuum model with a variety of four-fermion couplings. We discuss the continuum symmetries and relate them to the lattice symmetries. In \cref{sec:bosonization} we apply non-abelian bosonization to the continuum four-fermion theory and rewrite the low energy physics in terms of bosonic excitations. This helps us uncover the phase diagram of our lattice model. In \cref{sec:symenh}, we explain how the symmetries of the lattice model are enhanced at the critical point and derive expressions for the spin and dimer correlation functions near the critical point and in the conformal phase. In \cref{sec:numerical} we verify the theoretical analysis against MC results obtained using the meron-cluster algorithm and determine the critical point. We also confirm our estimate for the critical point using an exact diagonalization method. In \cref{conclusions} we summarize our results and provide an outlook for the future.

\section{Our Lattice Model}\label{sec:Lattice}
In this section we introduce our lattice model and discuss its symmetries. We focus on models with two flavors (or spins) of lattice fermions which can be annihilated and created at each lattice site $j$ by the usual fermionic Fock operators $c_{j\alpha}$ and $c^\dagger_{j\alpha}$, where $\alpha = \uparrow,\downarrow$. The choice of our model is constrained by the ability to use the meron-cluster algorithm as discussed in \cite{Chandrasekharan:2002vk}. While a variety of models fall in this class, here we focus on a particularly simple one whose Hamiltonian is
\begin{align}
  H_J = -J  \sum_{\langle i,j\rangle}H_{\langle i,j\rangle\uparrow}H_{\langle i,j\rangle\downarrow},
  \label{eq:HJ}
\end{align}
where
\thinmuskip=1mu
\medmuskip=2mu 
\thickmuskip=3mu 
\begin{align}
H_{\langle i,j\rangle\alpha} = - & ({c}_{i\alpha }^{\dagger }{c}_{j\alpha } + {c}_{j\alpha }^{\dagger }{c}_{i\alpha }) + 2\left({n}_{i\alpha }-\frac{1}{2}\right)\left({n}_{j\alpha}-\frac{1}{2}\right)-\frac{1}{2}.
\end{align}
\thinmuskip=3mu
\medmuskip=4mu 
\thickmuskip=5mu 
The symbol $\langle i,j\rangle$ refers to a bond between nearest neighbor sites $i$ and $j$. Our model above remains solvable by the meron-cluster algorithm in the presence of some other carefully chosen interactions, for example, the Hubbard interaction,
\begin{align}
  H_U = U\sum _{j}\Bigg\{\left({n}_{j\uparrow }-\frac{1}{2}\right)\left({n}_{j\downarrow }-\frac{1}{2}\right) + \frac{1}{4}\Bigg\}.
  \label{eq:HU}
\end{align}
In this work we will study the Hamiltonian
\begin{align}
  H = H_J + H_U.
\label{eq:ourmodel}
\end{align}
This model and its variants can be defined on a bipartite lattice in any dimension and remain solvable using the meron-cluster approach. The constraints of this solvability typically make them strongly interacting and end up within massive phases. However, as we will show in this work, they can still be useful to study phase transitions to massless phases. Here we will study this interesting quantum phase transition of \cref{eq:ourmodel} in one spatial dimension as a function of $U$.

Our model has a variety of interesting lattice symmetries that become manifest when written in terms of Majorana operators. We define two such operators $\gamma_j^1$ and $\gamma_j^2$ for spin-up fermions on each lattice site through the relations
\begin{align}
  c_{j\uparrow}&=\frac{1}{2}(\gamma_j^{1}-i\gamma_j^{2}),\quad  c^\dagger_{j\uparrow}=\frac{1}{2}(\gamma_j^{1}+i\gamma_j^{2})
\end{align}
for even $j$, and
\begin{align}
  c_{j\uparrow}&=\frac{1}{2}(\gamma_j^2+i\gamma_j^1),\quad c^\dagger_{j\uparrow}=\frac{1}{2}(\gamma_j^2-i\gamma_j^1)
\end{align}
for odd $j$. Note that $n_{j\uparrow}=\frac{1}{2} (-i\gamma_j^{1}\gamma_j^{2}+1)$ in both cases. Similarly we define two more Majorana operators $\gamma_j^3$ and $\gamma_j^4$ using the spin-down fermions. In terms of the four Majorana operators the two parts of the Hamiltonian take the form
\begin{align}
  H_J &= -\frac{J}{4}\sum_{\langle i,j\rangle}\prod_{\mu=1}^4(1+i\gamma_{i}^\mu\gamma_{j}^\mu), \\
  H_U &=  -\frac{U}{96}\sum_j \varepsilon_{\mu\nu\rho\sigma}\gamma_j^\mu \gamma_j^\nu \gamma_j^\rho \gamma_j^\sigma.
\end{align}
Let us now argue that $H_J$ is invariant under $O(4) \times \mathbb{Z}_2^\chi$ transformations, while $H_U$ is invariant under $SO(4) \times \mathbb{Z}_2^\chi$. The $SO(4)$ symmetry becomes obvious when we realize that the following six operators
\begin{align}
  \Gamma^{\mu\nu} = i\sum_{j}\gamma_j^\mu\gamma_j^\nu
\end{align}
satisfy the $\mathfrak{so}(4)$ algebra and commute with $H$. In fact the four operators $\gamma_j^\mu$ transform as an $SO(4)$ vector. We will argue in the next section that this symmetry is the vector subgroup of the full chiral symmetry group in the continuum. In addition, the $H_J$ has a discrete symmetry generated by $C_\uparrow := i\sum_{j}\gamma_j^{1}\gamma_j^{3}\gamma_j^{4}$, which we denote as $\mathbb{Z}_2^C$ for convenience. It is easy to verify that $C_\uparrow$ flips the sign of $\gamma_j^{2}$ but not the other three Majorana operators. Hence along with $C_\uparrow$ the four Majorana operators actually transform under the $O(4)$ group. Note that in the Dirac fermion language, $C_\uparrow c_{j\uparrow}C_\uparrow = (-1)^jc_{j\uparrow}^\dagger$ is the familiar particle-hole transformation on the spin-up fermions. Therefore it can also be viewed as the spin-charge flip transformation that is more familiar in condensed matter physics. On a regular lattice $H$ is also invariant under translations by one lattice site $T^\dagger_a c_{j\alpha} T_a = c_{j+1\alpha}$. As we will argue in the next section, there is a remnant discrete $\mathbb{Z}_2^\chi$ subgroup of the continuum chiral symmetry group buried within $T_a$. 
When $U=0$, the above lattice symmetries actually lead to an interesting degeneracy in the energy spectrum when the lattice size is a multiple of four, as will be shown in \cref{app:degeneracy}. In fact all energy levels are evenly degenerate, and in particular the ground state is doubly degenerate.

In order to compute quantities in our lattice model with the meron-cluster algorithm we use the continuous time formulation of the partition function,
\thinmuskip=1mu
\medmuskip=2mu 
\thickmuskip=3mu 
\begin{align}
Z = \sum_k \int \d t_k\cdots \d t_1 \ \sum_{[b]}\ J^k\  \Tr\Big(H_{b_k}(t_k) \cdots H_{b_1}(t_1)\Big),
\end{align}
\thinmuskip=3mu
\medmuskip=4mu 
\thickmuskip=5mu 
where $H_{b}(t)\ =\ \e^{t H_U} H_{\langle i,j\rangle\uparrow}H_{\langle i,j\rangle\downarrow} \e^{-t H_U}$ is the bond operator associated with the bond $b = \langle i j\rangle$ inserted at time $t$. Notice that here we use $\Tr$ to denote the trace over the full Hilbert space. Later we will use $\tr$ to denote the trace over local Hilbert spaces or local fields. The integrals over Euclidean time are always assumed to be time ordered such that $\beta > t_k > .... > t_2 > t_1 > 0$. The choice of $H_b$ leads to a simple formula for the trace in the fermionic Hilbert space. In particular it does not contain any determinants of large matrices, as is the case in the traditional auxiliary field methods. Instead, it can be shown that
\begin{align}
\Tr\Big(H_{b_k}(t_k) \cdots H_{b_1}(t_1)\Big) \ =\ \prod_{i} W(\ell_i),
\end{align}
where $\{\ell_1,\ell_2,...\}$ is a set of loops and $W(\ell_i)$ is the weight associated with the loop $\ell_i$ \cite{Chandrasekharan:2002vk}. The loops can be identified by introducing two parallel bonds for each $H_b$ at the appropriate imaginary time, as illustrated in \cref{fig:configurations}. As can be seen from the figure, these bonds naturally divide the lattice into disconnected loop clusters. When the trace over the fermionic Hilbert space is performed each cluster gets a weight $W(\ell) = 2(1\pm \e^{- U t_\ell/2})$, where $t_\ell$ is the linear temporal size of the loop. The sign associated with each loop is given by $(-1)^{n_t + n_b/2+1}$, where $n_t$ is the number of temporal winding of the loops and $n_b$ is the number of bonds in the loop. The fermionic nature of the problem is hidden in this sign. Note that when $U=0$ clusters with a negative sign (merons) are naturally forbidden. On the other hand when $U$ is very large all clusters are allowed and from the cluster representation our model becomes identical to the Heisenberg spin-half chain \cite{Evertz:1992rb,Beard:1996wj}.

\begin{figure}[h]
\centering \includegraphics[width=0.3\textwidth]{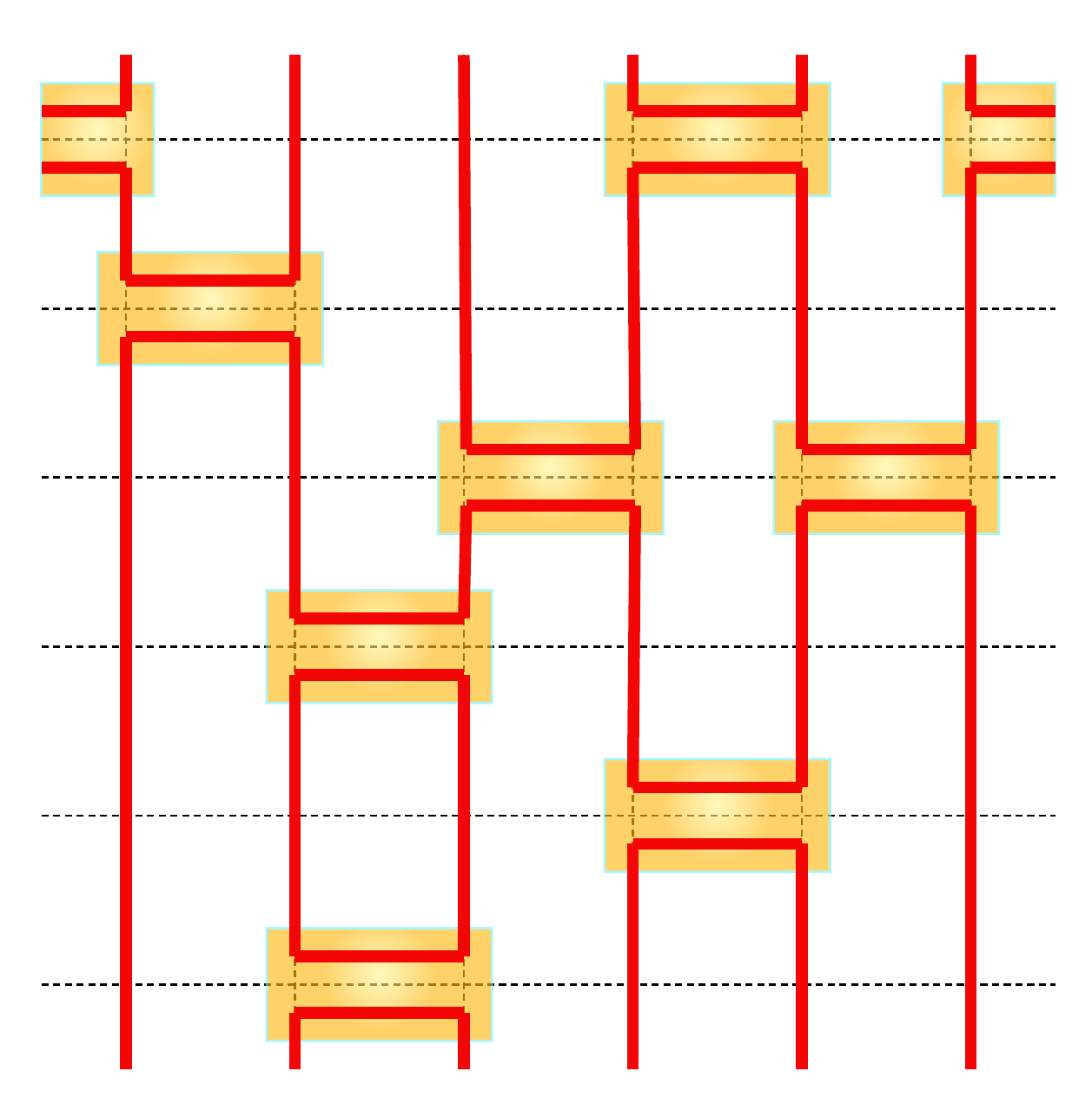}
  \caption{Illustration of a configuration of bonds that naturally divides space-time into loops. The fermionic trace is given as a product of weights associated with each loop as described in the text.}
  \label{fig:configurations}
\end{figure}

\section{Continuum Field Theory Analysis}
\label{sec:ft}

In this section, we review the results for the non-abelian bosonization of the Hubbard model pioneered by Affleck et al. \cite{Affleck:1988px,Affleck1990}, and the field theoretic picture for the phase transition from critical to VBS in one dimension. While this scenario is now standard, a new aspect of our model is the spin-charge flip symmetry $\mathbb{Z}^C_2$ which is present in $H_J$ at a lattice level even beyond the quadratic level. This extra symmetry appears in an interesting way in the continuum description, the RG flow and the resulting phase diagram.

\subsection{Connection to the Gross-Neveu-Thirring Model}
\label{sec:continuum}

In this section we will identify the $1+1$ dimensional continuum QFT which is described by our lattice model in one spatial dimension, given in \cref{eq:ourmodel}. Since the model is strongly interacting, in principle such an identification can be questionable. Still we can try to perform a tree level analysis by expanding the lattice Hamiltonian in terms of modes near the Fermi points of the free Hamiltonian and then include interactions perturbatively. Such an analysis will teach us how the lattice symmetries are embedded within the continuum ones \cite{Liu:2019dvk}. In particular we show below that the $H_J$ term can be identified with the strongly coupled lattice-regularized two-flavor chiral-mass GN model, while the $H_U$ term introduces a combination of the usual GN coupling and a Thirring coupling. More details of these couplings and the mapping are explained in \cref{app:lagrangian}.

In order to perform the tree level analysis described above, we introduce a small parameter $\varepsilon$ and deform $H_J \rightarrow H^\varepsilon_J$ as follows,
\begin{align}
  H^\varepsilon_J &= -\frac{J}{4\varepsilon}\sum_{\langle i,j\rangle}\prod_{\mu=1}^4(1+i\varepsilon\gamma_{i}^\mu\gamma_{j}^\mu).
\label{eq:modHJ}
\end{align}
Clearly all lattice symmetries discussed in the previous section are maintained order by order in $\varepsilon$. Expanding in powers of $\varepsilon$, we get
\begin{align}
  H^\varepsilon_J =: \sum_{n=1}^4\varepsilon^{n-1} H_J^{(2n)},
\end{align}
where the leading term
\begin{align}
  H_J^{(2)} =& -\frac{J}{2}  \sum_{\<i,j\>\alpha}c_{i\alpha}^\dagger c_{j\alpha} + c_{j\alpha}^\dagger c_{i\alpha}
 \end{align}
 describes free fermions, and
 \begin{align}
  H_J^{(4)} = J  \sum_{\<i,j\>}\Big\{ &- (c_{i\uparrow}^\dagger c_{j\uparrow} + c_{j\uparrow}^\dagger c_{i\uparrow})(c_{i\downarrow}^\dagger c_{j\downarrow} + c_{j\downarrow}^\dagger c_{i\downarrow})
  \nonumber\\
  & +\sum_\alpha\big({n}_{i\alpha }-\frac{1}{2}\big)\big({n}_{j \alpha}-\frac{1}{2}\big)\Big\}
\end{align}
is a four-fermion interaction. Higher order terms $H_J^{(6,8)}$ are irrelevant perturbatively near the free fermion fixed point and therefore their exact forms will not affect our discussion. By focusing on $H_J^{(2)}$ and $H_J^{(4)}$ we can identify the continuum model and also map the lattice symmetries to the continuum. Our lattice model with $\varepsilon=1$ will be regarded as merely another lattice discretization of the same continuum model.

In order to analyze the low energy physics near the free fermion fixed point, we expand our lattice fields $c_{j\alpha}$ in terms of smooth fields $\psi_{\alpha,L}(a j)$ and $\psi_{\alpha,R}(a j)$ near the two Fermi momenta $\mp k_F$, where $k_F = \pi/2a$ and $a$ is the lattice spacing. Without loss of generality we choose
\begin{align}\label{eq:linearization}
  \frac{1}{\sqrt{a}} c_{j\alpha} \approx \e^{ ik_Faj}\psi_{\alpha,L}(aj) + \e^{- ik_Faj}\psi_{\alpha,R}(aj)
\end{align}
for both $\alpha$. Inserting \cref{eq:linearization} into $H_J^{(2)}$,
we can expand in powers of lattice spacing $a$ to obtain the leading low energy effective Hamiltonian in the continuum as
\begin{align}
  H_J^{(2)\text{cont}} = - aJ  \int\d x \Big( \sum_{\alpha} &\psi_{\alpha,L}^\dagger i\ddx \psi_{\alpha,L} - \psi_{\alpha,R}^\dagger i\ddx \psi_{\alpha,R}\Big).
  \label{eq:freeHc}
\end{align}
This is the Hamiltonian for two-flavor free massless Dirac fermions. It is easy to verify that the above free Hamiltonian is invariant under $O(4)_L \times O(4)_R$ which we refer to as the full continuum chiral symmetry. The $SO(4)$ transformation in each chiral sector can be decomposed into spin $SU(2)_s$ and charge $SU(2)_c$ transformations, which becomes explicit if we introduce a $2\times 2$ matrix of operators,
\begin{align}
  \Psi_{L(R)} = \begin{pmatrix}
    \psi_{\uparrow,L(R)} & \psi_{\downarrow,L(R)}^{\dagger} \\
    \psi_{\downarrow,L(R)} & -\psi_{\uparrow,L(R)}^{\dagger}
  \end{pmatrix},
\end{align}
and define the spin and charge transformations as $\Psi_{L(R)} \mapsto S_{L(R)} \Psi_{L(R)} Q^\dagger_{L(R)}$, where $S_{L(R)}$ and $Q_{L(R)}$ are $SU(2)_s$ and $SU(2)_c$ matrices in each chiral sector.
When written in terms of $\Psi_{L(R)}$, the Hamiltonian
\begin{align}
H_J^{(2)\text{cont}} = - aJ \int \d x \frac{1}{2}& \tr\Big( \Psi_L^\dag i\ddx \Psi_L -\Psi_R^\dag i\ddx \Psi_R\Big)
\end{align}
is clearly invariant under the above transformation. However, since $(-S_{L(R)},-Q_{L(R)})$ should be identified with $(S_{L(R)},Q_{L(R)})$, the symmetry group of the continuum Hamiltonian $H_J^{(2)\text{cont}}$ is actually $(SU(2)_s\times SU(2)_c)\big/\mathbb{Z}_2 \cong SO(4)$ in each chiral sector. These transformations are generated on the Hilbert space by the spin current operators
\begin{align}
\mathcal{J}^i_{s L(R)}(x) &= \frac{1}{2} (
\psi_{\uparrow,L(R)}^{\dagger}, \psi_{\downarrow,L (R)}^{\dagger})
\sigma^i \begin{pmatrix} \psi_{\uparrow,L(R)} \\
\psi_{\downarrow,L(R)}\end{pmatrix},
\label{eq:sgen}
\end{align}
and the charge current operators 
\begin{align}
\mathcal{J}^{i}_{cL(R)}(x) &= \frac{1}{2} (
\psi_{\uparrow,L(R)}, \psi_{\downarrow,L(R)}^\dagger)
\sigma^{iT} \begin{pmatrix} \psi_{\uparrow,L(R)}^\dagger \\
\psi_{\downarrow,L(R)} \end{pmatrix}.
\label{eq:cgen}
\end{align}
In addition the Hamiltonian is also invariant under two independent spin-charge flip transformations $\mathbb{Z}_{2L(R)}^C$: $\Psi_{L(R)} \mapsto \Psi^\dagger_{L(R)}$, under which $\mathcal{J}^i_{s L(R)}$ and $\mathcal{J}^i_{c L( R)}$ exchange with each other. Including them the free Hamiltonian is indeed invariant under the $O(4)_L\times O(4)_R$ symmetry as stated above.

We know from \cref{eq:linearization} that under lattice translation $T_a$, $\psi_{L}^{\alpha} \mapsto i\psi_{L}^{\alpha}$ and $\psi_{R}^{\alpha} \mapsto - i\psi_{R}^{\alpha}$. This means $T_a$, corresponding to $Q_L = Q_R^\dagger = \exp(i\frac{\pi}{2}\sigma^3)$ in the continuum, generates a $\mathbb{Z}_4$ subgroup of the chiral $SU(2)_c$ group, under which $\mathcal{J}^i_{sL(R)} \mapsto \mathcal{J}^i_{sL(R)}$,
$\mathcal{J}^{1,2}_{cL(R)} \mapsto - \mathcal{J}^{1,2}_{cL(R)}$ and $\mathcal{J}^3_{cL(R)} \mapsto \mathcal{J}^3_{cL(R)}$. Note that here since $Q_L^2=Q_R^2=-\mathbbm{1}$, $T_a^2$ belongs to the vector subgroup of the chiral symmetry group, and therefore $T_a$ is effectively a $\mathbb{Z}_2$ subgroup of the chiral symmetry group, hence denoted by $\mathbb{Z}_2^\chi$.

Using a similar analysis as above we can identify the following continuum operators that corresponds to $H_J^{(4)}$
\begin{align}
  H_J^{(4)\text{cont}} &= 2aJ \int \d x \ M(x)^2,
\end{align}
where 
\begin{align}
M(x) = i\sum_\alpha \big(\psi_{\alpha,L}^\dagger \psi_{\alpha,R} - \psi_{\alpha,R}^\dagger \psi_{\alpha,L}\big).
\label{eq:Mx}
\end{align}
In order to understand the symmetries of this term it is useful to express it in terms of spin and charge currents using the relation $M(x)^2 = \left(\mathcal{J}_{sL}^i\mathcal{J}_{sR}^i + \mathcal{J}_{cL}^i\mathcal{J}_{cR}^i\right) - 1$,
which implies
\begin{align}
  H_J^{(4)\text{cont}} &= 2aJ \int \d x \left(\mathcal{J}_{sL}^i\mathcal{J}_{sR}^i + \mathcal{J}_{cL}^i\mathcal{J}_{cR}^i\right).
\end{align}
In this form it is easy to see that this term breaks the chiral $O(4)$ symmetry of the free theory down to the diagonal $O(4)_{L=R}$ subgroup. However, the $\mathbb{Z}_2^\chi$ group generated by $T_a$ survives. Therefore $H_J^{(4)\text{cont}}$ has $O(4)\times\mathbb{Z}_2^\chi$ symmetry.

Coming to the Hubbard interaction $H_U$ in \cref{eq:HU}, the corresponding continuum Hamiltonian was already derived by Affleck et al. \cite{Affleck:1988zj} and is given by
\begin{align}
 H_U^\text{cont}\approx \frac{1}{2}a U \int \d x &\Big(\mathcal{J}_{cL}^i\mathcal{J}_{cR}^i - \mathcal{J}_{sL}^i\mathcal{J}_{sR}^i 
 \nonumber \\
& - \frac{1}{3}(\mathcal{J}_{sL}^i\mathcal{J}_{sL}^i + \mathcal{J}_{sR}^i\mathcal{J}_{sR}^i)\Big)
\label{eq:contHu}
\end{align}
up to irrelevant pieces. Note that $H_U$ preserves all the symmetries of $H_J$ except the spin-charge flip symmetry which means the $O(4)\times\mathbb{Z}_2^\chi$ symmetry is reduced to $SO(4)\times\mathbb{Z}_2^\chi$. The last two terms in \cref{eq:contHu} only renormalize the speed of light and do not introduce any interactions, which is clearer in the language of conformal field theory (CFT). Throwing away these chiral terms, and normalizing the kinetic term to have a coefficient of one, we obtain the final continuum Hamiltonian of our lattice model as
\begin{align}
  H^\text{cont} = \int \d x& \Big( \sum_{\alpha} \big(-\psi^\dagger_{\alpha,L}i\ddx \psi_{\alpha,L}+\psi^\dagger_{\alpha,R}i\ddx \psi_{\alpha,R}\big) \nonumber\\
  & + \lambda_s\mathcal{J}_{sL}^i\mathcal{J}_{sR}^i + \lambda_c\mathcal{J}_{cL}^i\mathcal{J}_{cR}^i \Big),
  \label{eq:contmodel}
\end{align}
where $\lambda_s = 2(\varepsilon-\frac{U}{4J})$ and $\lambda_c = 2(\varepsilon+\frac{U}{4J})$. Assuming the $H^{(6,8)}$ terms do not modify the above analysis except perhaps to change the couplings $\lambda_s$ and $\lambda_c$, we argue that \cref{eq:ourmodel} is a lattice regularization of \cref{eq:contmodel}. In \cref{app:lagrangian} we construct the Lagrangian of \cref{eq:contmodel} and identify it with a linear combination of GN-Thirring couplings.

In summary the above discussion shows that the $O(4)$ symmetry of our lattice model in \cref{eq:ourmodel} can then be identified with the diagonal $O(4)$ subgroup of the continuum Hamiltonian \cref{eq:contmodel}. The translation by one site symmetry $T_a$ can be identified as the remnant chiral symmetry $\mathbb{Z}_2^\chi$ discussed above. While it is known that lattice translation is part of the continuum chiral symmetry, we have shown explicitly how it is embedded along with the flavor symmetries.

\subsection{Bosonization and Phase Diagram}
\label{sec:bosonization}

In order to understand the physics of our lattice model we begin by discussing the RG flows of the continuum model \cref{eq:contmodel} near the free fermion fixed point. The one-loop $\beta$-functions in the continuum are given by \cite{Affleck:1988zj}
\begin{align}
  \frac{\d\lambda_{s(c)}}{\d\log\mu} = -\frac{\lambda_{s(c)}^2}{2\pi},
\label{eq:betafunc}
\end{align}
which shows that the spin and charge currents are completely decoupled in the low energy theory, i.e., terms involving both spin currents and charge currents and respecting the $SO(4)$ symmetry are irrelevant. The RG flow diagram based on this $\beta$-function is shown in \cref{fig:flow-diagram}. The couplings are relevant when $\lambda_{s(c)} > 0$, and irrelevant when $\lambda_{s(c)} < 0$. 

\begin{figure}[t]
\includegraphics[width=0.35\textwidth]{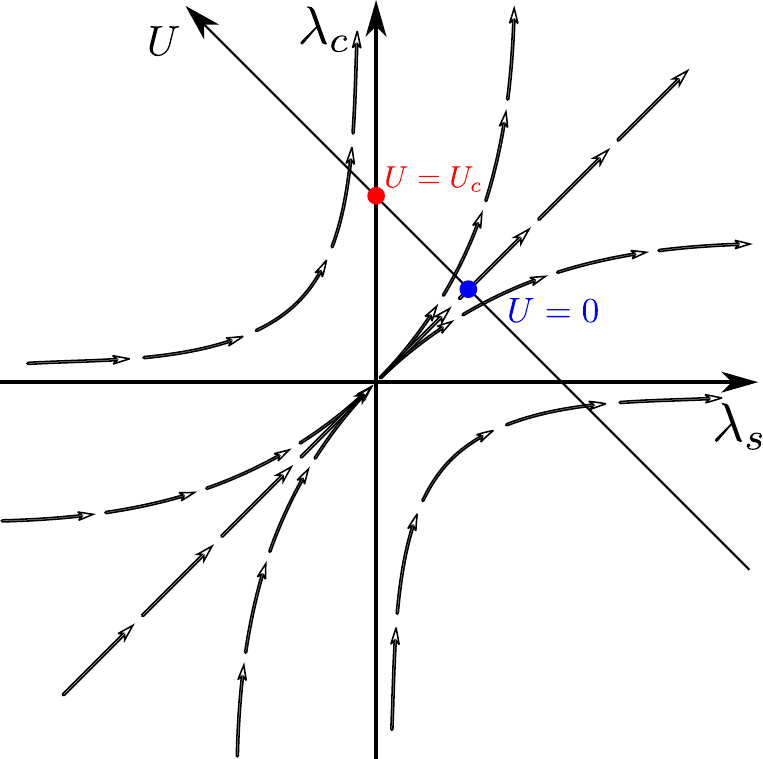}
  \caption{Phase diagram in the $\lambda_s \operatorname{-} \lambda_c$ plane. Both of the couplings are relevant when positive and irrelevant when negative. Our model lives on the $U$ axis. When $U = 0$, Our model is on the $\lambda_s = \lambda_c$ line due to the spin-charge flip symmetry, and is in a massive phase described by a lattice regularized chiral-mass GN model with a spontaneously broken $\mathbb{Z}_2^\chi$ symmetry. When $U$ increases to $U_c$, $\lambda_s = 0$ and our model is at the critical point of a second order phase transition, whose low energy physics is described by the $SU(2)_1$ WZW model. When $U>U_c$, $\lambda_s < 0$ and a marginally irrelevant coupling would modify the low energy WZW theory mildly.}
  \label{fig:flow-diagram}
\end{figure}

Our lattice model falls somewhere on the $(\lambda_s,\lambda_c)$ plane. For example when $U=0$ the model must be on the $\lambda_c=\lambda_s$ line due to the spin-charge flip symmetry, which is identified with the chiral-mass GN model with the well known physics of asymptotic freedom. Here we expect $\langle M(x)\rangle \neq 0$ due to the spontaneous breaking of the $\mathbb{Z}_2^\chi$ chiral symmetry \cite{Gross:1974jv}, which is related to the dimer order parameter and verified by our MC results presented later. When $U > 0$ our model moves away from the spin-charge symmetric axis towards the conformal phase in the second quadrant where $\lambda_s < 0$ and $\lambda_c > 0$. This implies the existence of a quantum phase transition at some critical value $U=U_c$, which we find to be second order. 

The most convenient way to understand the physics of the conformal phase is through the language of non-abelian bosonization \cite{Witten:1983ar}, where \cref{eq:contmodel} is mapped to the following non-linear sigma model,
\begin{align}
    S[g,h] &= S_\text{WZW}[g] + \lambda_s\int \d^2 x \mathcal{J}_{sL}^i(g)\mathcal{J}_{sR}^i(g) \nonumber\\
    &+ S_\text{WZW}[h] + \lambda_c\int \d^2 x \mathcal{J}_{cL}^i(h)\mathcal{J}_{cR}^i(h).
\label{eq:contbosqft}
\end{align}
$g$ and $h$ are $SU(2)$ group-valued fields, describing the spin and charged sectors respectively, and 
\begin{align}\label{WZW}
  &S_\text{WZW}[g] = \frac{1}{16\pi} \int \d^2x \tr (\partial_\mu g^{-1} \partial^\mu g) \nonumber\\
  &\quad + \frac{i}{24\pi}\int \d^3x \varepsilon^{\mu\nu\rho} \tr (g^{-1}\partial_\mu g g^{-1}\partial_\nu g g^{-1}\partial_\rho g )
\end{align}
is the action for the Wess-Zumino-Witten (WZW) theory in the Euclidean signature. While the first term in $S_\text{WZW}[g]$ involves an integration over 2-sphere space-time, the second term involves an integration over a 3-disk whose boundary is the 2-sphere. The currents $\mathcal{J}_{sL(R)}^i$ and $\mathcal{J}_{cL(R)}^i$ are given by
\begin{align}
    \mathcal{J}_{sL}^j &= \frac{i}{4\pi}\tr\left((\partial g)g^{-1}\sigma^j\right), \\
    \mathcal{J}_{sR}^j &= \frac{i}{4\pi}\tr\left(g^{-1}\bar\partial g\sigma^j\right), \\
    \mathcal{J}_{cL}^j &= \frac{i}{4\pi}\tr\left((\partial h)h^{-1}\sigma^j\right), \\
    \mathcal{J}_{cR}^j &= \frac{i}{4\pi}\tr\left(h^{-1}\bar\partial h\sigma^j\right),
\end{align}
where $\partial = \partial_x + i\partial_t$ and $\bar\partial$ is the complex conjugate of $\partial$. Note that although in the Euclidean space those currents are usually called holomorphic and anti-holomorphic and denoted by $\mathcal{J}$ and $\bar{\mathcal{J}}$, here we keep the convention from the Minkowski space and still call the currents as left and right handed for convenience. The bosonized action \cref{eq:contbosqft} enjoys almost the same symmetries as \cref{eq:contmodel}, except for the fact that $(g,h)$ and $(-g,-h)$ should be identified. This constraint comes from the fact that we have written the bosonized theory with an $SU(2)_s\times SU(2)_c$ symmetry rather than $SO(4)$. This means that observables in the fermionic theory will always be mapped into terms that do not violate this identification. For example $\tr g$ itself is not an observable. Besides, the chiral symmetries of $S_\text{WZW}$ in each sector are $SO(4)$ instead of $SU(2)_L\times SU(2)_R$. Also note that under spin-charge flip symmetry $g \leftrightarrow h$.

Armed with these bosonization results we see that when $U > 0$ the charge excitations have higher energies than the spin excitations. Thus, the long distance physics of the conformal phase in the second quadrant must be described by
\begin{align}
S = S_\text{WZW}[g] + \lambda_s\int \d^2 x \mathcal{J}_{sL}^i(g)\mathcal{J}_{sR}^i(g).
\label{eq:k1wzw}
\end{align}
This long distance physics can also be seen directly in the lattice model, where the charge sector becomes energetically less favorable when $U>0$ and the low energy physics is described by a Heisenberg spin-half chain. In a famous theorem, Lieb, Shultz and Mattis showed that spin-half chains can only be in one of two possible ground states \cite{Lieb:1961fr}: a dimerized phase where the translation symmetry is broken or a critical phase. In fermionic viewpoint, the dimerized phase is the chirally broken GN phase while the critical phase is the CFT described by \cref{eq:k1wzw} with $\lambda_s \leq 0$. This implies the universality class of the transition at $U=U_c$ in our model can also be understood from simpler spin-half chain models. One such example is
\begin{align}
H_{J_1\operatorname{-}J_2} = \sum_{i} \Big\{ J_1 \mathbf{S}_i \cdot \mathbf{S}_{i+1} 
 +\  J_2 \mathbf{S}_i \cdot \mathbf{S}_{i+2} \Big\}
\label{eq:mgmod}
\end{align}
with both $J_1, J_2 > 0$. When $J_2=0$, this model is the usual Heisenberg spin chain and is well known to be described the WZW conformal phase. Also, as Majumdar and Ghosh pointed out long ago, when $J_2/J_1 = 0.5$ the ground state is doubly degenerate due to dimerization, suggesting that the theory at those couplings is in a different phase \cite{Majumdar_1969,Majumdar_1970}. Previous studies have shown that the phase transition between the two phases occurs at $J_2/J_1 = 0.241167(5)$ \cite{PhysRevB.25.4925, Eggert:1996er}. Although the phase transition has also been studied without a sign problem in one-dimensional spin-half chain using the J-Q idea \cite{sanyal2011:1djq,Patil:2018wpt}, our model provides a route to the one-dimensional WZW critical phase to VBS quantum phase transition  starting from a lattice fermionic Hubbard model Hilbert space instead of a spin-half chain Hilbert space. A very narrow region of VBS was also reported in an extended Hubbard model~\cite{sandvik2004:exthubb}.

\subsection{Symmetry Enhancement and Correlation Functions}
\label{sec:symenh}

Since the charge sector decouples when $U > 0$, the long distance physics of the lattice model seems to only have the diagonal $SU(2)$ spin symmetry and the remnant $\mathbb{Z}^\chi_2$ chiral symmetry, under which $g\mapsto i\sigma^3 g i\sigma^3$. These symmetries can be observed in the continuum low energy theory described by \cref{eq:k1wzw}. However, if we set $\lambda_s=0$ the continuum theory is invariant under the enhanced chiral transformation $g \mapsto S_L g S_R^\dagger$ where $S_L$ and $S_R$ are two independent $SU(2)$ matrices implementing the left and right spin symmetries of the fermion model, except that $S_L=S_R=-\mathbbm{1}$ needs to be quotiented out. Based on the RG flow diagram shown in  \cref{fig:flow-diagram} we see that by tuning to the critical point $U=U_c$ we can indeed set $\lambda_s=0$, and thus we must see the enhanced chiral symmetry there. In our work we look for this symmetry enhancement, although observing it could be non-trivial since there will always be higher dimensional operators that break the symmetry on the lattice. Also, when $U \neq U_c$ but close to it, since $\lambda_s$ is a marginal coupling, the logarithmic corrections due to it would be visible in the lattice MC data at long distances. Here we derive the form of these corrections and use them in our analysis.

It was explained in \cite{Affleck1990,tsvelik_2003} that the symmetry enhancement is visible in correlation functions of spin and dimer operators on the lattice defined through the relations
\begin{align}
  S_j^z(t) &= (-1)^j\frac{1}{2}(n_{j\uparrow}-n_{j\downarrow})(t), \label{eq:spin}\\
  D_j(t) &= (-1)^j\frac{1}{2}(S^z_jS^z_{j+1}-S^z_{j-1}S^z_j)(t). \label{eq:dimer}
\end{align}
The leading continuum terms of these two operators are the primary fields $\varphi_S = \tr h \tr (gi\sigma^3)$ and $\varphi_D = \tr h \tr g$ of the WZW model. In the fermionic language these operators are given by $\varphi_S \propto \psi_{L}^{\dagger\alpha} \sigma_{\alpha\beta}^3 \psi_R^\beta + \text{h.c.}$ and $\varphi_D(x) \propto M(x)$ defined in \cref{eq:Mx}. We see that $\varphi_D$ can be used as an order parameter on the lattice for $\mathbb{Z}_2^\chi$. Since $\varphi_S$ and $\varphi_D$ transform into each other under the $SO(4)$ chiral transformations, their correlation functions will be related to each other. Using the well known techniques in CFT \cite{Knizhnik:1984nr, Zamolodchikov:1986bd, Gepner:1986wi}, their conformal dimensions have been computed to be $h_S =  h_D = \frac{1}{2}$, and the two point correlation functions of the spin and dimer are
\begin{align}
G_i(r) = \langle \varphi_i(r) \varphi_i(0) \rangle \propto r^{-2h_i},
\label{eq:critcorr}
\end{align}
for large values of $r$ at the critical point $U=U_c$. 

In order to be able to fit our MC data away from the critical point we have to include the logarithmic corrections to \cref{eq:critcorr}, due to the presence of the marginal coupling $\lambda_s$. We begin with the $\beta$-function of the marginal coupling $\lambda_s$
\begin{align}\label{eq:beta_function}
  \beta(\lambda_s) &= \frac{\d\lambda_s}{\d\ln r} = \frac{\lambda_s^2}{2\pi} + O(\lambda_s^3).
\end{align}
This is the same as \cref{eq:betafunc} except that we have replaced the energy scale $\mu$ with a length scale $r$, which reverses the sign of the beta function. Integrating  to first order, we obtain 
\begin{align}\label{eq:int_beta}
  \frac{1}{\lambda_s(r)} - \frac{1}{\lambda_0} = -\frac{1}{2\pi}\ln \frac{r}{r_0},
\end{align}
where $\lambda_0$ is the coupling at some short distance length scale $r_0$. Note that the solution $\lambda_s(r)$ above is meaningful for very large values of $r$ only when $\lambda_0 < 0$ (i.e., when $U > U_c$ and the theory is in the conformal phase), and therefore the coupling $\lambda_s(r)$ is marginally irrelevant. When $\lambda_0 > 0$ the coupling $\lambda_s(r)$ diverges at $r = r_0 \e^{2\pi/\lambda_0}$, which is the remnant of the Landau pole of perturbation theory and signals new physics at long distances due to the formation of a fermion mass. This means in the broken phase the above form of $\lambda_s(r)$ will only be valid close to the critical point and small values of $r$.

With the caveats discussed above, we can use $\lambda_s(r)$ in \cref{eq:int_beta} in both phases to derive the corrections for $G_i(r)$. Noting that it depends on $r$ both explicitly and through $\lambda_s(r)$, we can write $G_i(r) =: G_i(\lambda_s(r),r)$ to be more precise. We note that it satisfies the following RG equation:
\begin{align}\label{eq:RG1}
  \frac{\d \ln G_i(\lambda_s(r),r)}{\d \ln r} = -2(h_i+\gamma_i(\lambda_s(r))),
\end{align}
where $h_i$ is the conformal dimension of $\varphi_i$ at the critical point and $\gamma_i(\lambda_s(r))$ is the anomalous dimension of $\varphi_i$ induced by the marginal operator. Integrating this equation we get
\begin{align}
  \ln \frac{G_i(\lambda_s(r),r)}{G_i(\lambda_0, r_0)} &= \int_{r_0}^r -2(h_i + \gamma_i(\lambda_s(r')))\d \ln r',
  \label{eq:rgsol}
\end{align}
from which we can derive the corrections due the presence of the marginal operator once we know the $\gamma_i(\lambda_s(r))$. 

First let us focus on the conformal phase and near the critical point, i.e. $\lambda_0\lesssim 0$. In 2d CFT the conformal dimension $h_i$ of the operator $\varphi_i$ can be calculated using the finite size energy $E_i = 2\pi h_i/L $ of the state $|\varphi_i\>$ through the state-operator correspondence \cite{Cardy:1984rp}. Then the anomalous dimension $\gamma_i$ of the operator is related to the change in this energy due to the presence of the marginal operator. In our case, to leading order in $\lambda_s$, this leads to the expression \cite{Affleck:1988px},
\begin{align}\label{eq:anomalous_dimension}
  \gamma_i(\lambda_s) &= \frac{b_i}{2\pi}\lambda_s  + O(\lambda_s^2),
\end{align}
where 
\begin{align}
  b_i &= \int \d x~\< \varphi_i|\mathcal{J}_{sL}^j \mathcal{J}_{sR}^j|\varphi_i\> \nonumber\\
     &= \frac{1}{2} \int \d x~\< \varphi_i|(\mathcal{J}_{sL}^j + \mathcal{J}_{sR}^j)^2 - \mathcal{J}_{sL}^j\mathcal{J}_{sL}^j - \mathcal{J}_{sR}^j\mathcal{J}_{sR}^j|\varphi_i\> \nonumber\\
     &= \frac{1}{2}(s_\text{tot}(s_\text{tot}+1) - s_L(s_L+1) - s_R(s_R+1)).
\end{align}
In the above formula $s_L = s_R = 1/2$ for both spin and dimer states, while $s_\text{tot}$ depends on the state $|\varphi_i\rangle$. The values of $b_i$ can now be calculated for both spin and dimer fields and are tabulated in \cref{tab:b_i} below.

\begin{table}[htb]
  \centering
  \renewcommand{\arraystretch}{1.2}
  \setlength{\tabcolsep}{4pt}
  \begin{tabular}{{c} *{6}{c}}
  \TopRule
  Lattice Field & WZW field & $s_\text{tot}$ & $s_L$ & $s_R$ & $h_i$ & $b_i$ \\ \MidRule
  $S^i$  & $\varphi_S = \tr h \tr g\sigma^i$ & 1 & $1/2$ & $1/2$ & $1/2$ & $1/4$ \\
  \MidRule
  $D$  & $\varphi_D = \tr h\tr g$ & 0 & $1/2$ & $1/2$ & $1/2$ & $-3/4$ \\ 
  \BotRule
  \end{tabular}
  \caption{$h_i$ and $b_i$ for the lattice fields $S^i, D$}
  \label{tab:b_i}
\end{table}
Inserting these values of $b_i$ and $\lambda_s(r)$ from \cref{eq:int_beta} into \cref{eq:anomalous_dimension} we can compute $\gamma_i(\lambda_s(r))$. Substituting this into \cref{eq:rgsol} we get
\begin{align}
  \frac{G_S(\lambda(r), r)}{G_S(\lambda_0, r_0)} &= \frac{r_0}{r}\left(1-\frac{\lambda_0}{2\pi}\ln\frac{r}{r_0}\right)^{\frac{1}{2}}, \label{eq:spin_ir}\\
  \frac{G_D(\lambda(r), r)}{G_D(\lambda_0, r_0)} &=  \frac{r_0}{r}\left(1-\frac{\lambda_0}{2\pi}\ln\frac{r}{r_0}\right)^{-\frac{3}{2}}. \label{eq:dimer_ir}
\end{align}
Similar expressions have been derived earlier in the context of spin chains \cite{Eggert:1996er,Affleck_1998} and verified numerically \cite{Eggert:1996er}. Note again that for a fixed value of $\lambda_0$ but large values of $r$, the above expressions make sense only when $\lambda_0\leq 0$, i.e. when we are in the conformal phase as already mentioned above. On the other hand when $r$ is in a fixed range but $\lambda_0 \rightarrow 0$ the above expressions give us the leading corrections to conformal behavior on both sides of the critical point.

We could do a similar analysis in the massive phase. However, in the massive phase the anomalous dimension \cref{eq:anomalous_dimension} obtained in the conformal phase cannot be completely correct. Besides, like the $\beta$-function, it was derived in perturbation theory and so there could be some non-perturbative corrections to it due to the presence of the mass scale $M$. In particular if we assume the correlation function to the leading order is of the form $G_i(r) \sim \e^{- \alpha_i M r}$, this would contribute additional non-perturbative terms to the anomalous dimension $\gamma_i(\lambda_s)$. For example if we assume \begin{align}
\lambda_s(r) = -\frac{2\pi}{\ln Mr}
\label{eq:lambda_massive}
\end{align}
as a possible definition for the RG invariant mass scale $M$, then it is easy to see that
\begin{align}\label{eq:anomalous_dimension_massive}
  \gamma_i(\lambda_s) &= \frac{b_i}{2\pi}\lambda_s + \frac{\alpha_i}{2} \ \e^{-2\pi/\lambda_s} + O(\lambda_s^2).
\end{align}
Since we are only deriving corrections to $G(r)$ due to the marginal coupling perturbatively, in this work we ignore such non-perturbative effects. Then, \cref{eq:spin_ir,eq:dimer_ir} can still be used in the massive phase near the critical point except that we will find $\lambda_0 > 0$ and $r$ is constrained to be such that $\lambda_0/2\pi\ln (r/r_0) < 1$.

\section{Numerical Results}\label{sec:numerical}

\subsection{Monte Carlo Results}

\begin{figure*}[htb]
  \includegraphics[width=0.49\textwidth]{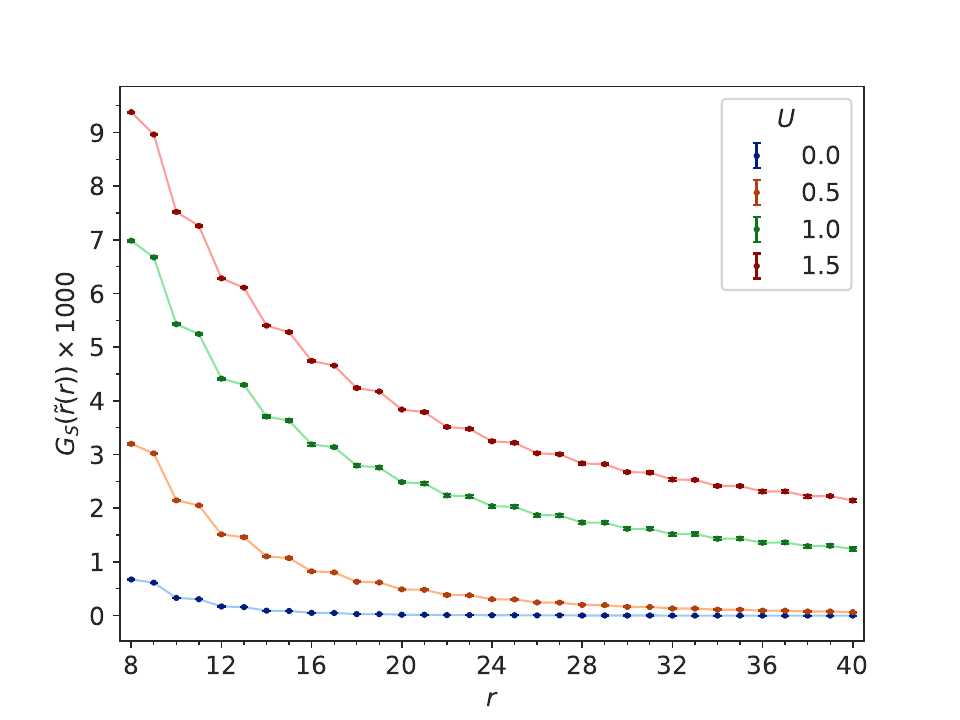}
  \includegraphics[width=0.49\textwidth]{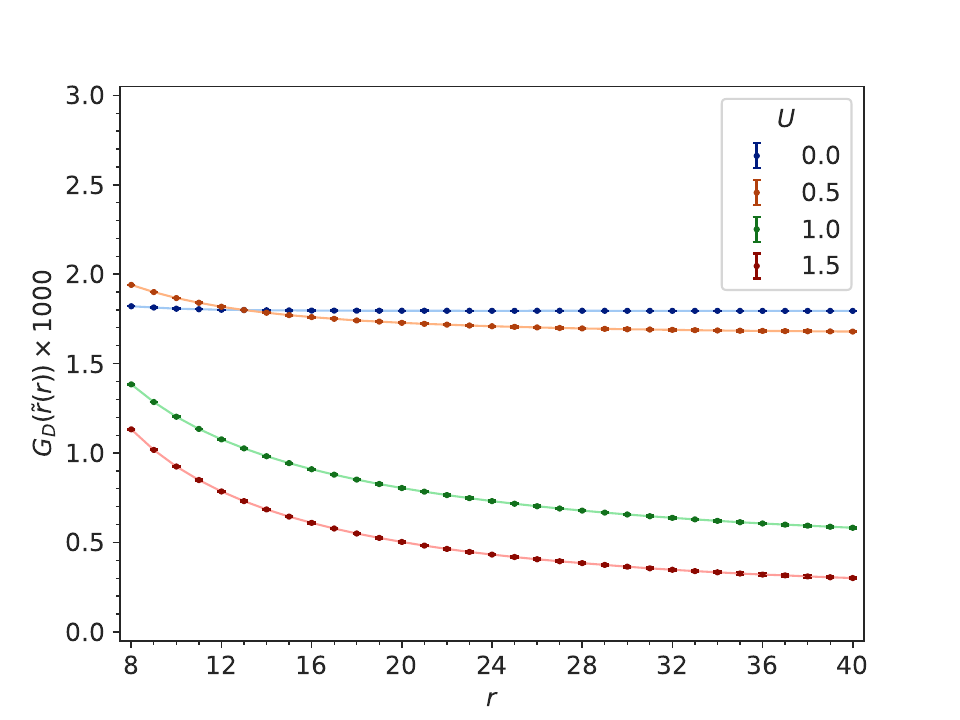}
  \caption{Spin correlation functions $G_S(\tilde{r}(r))$ and dimer correlation functions $G_D(\tilde{r}(r))$ as functions of $r$ at $L=128$ for $U=0$, $0.5$, $1.0$ and $1.5$. Both correlation functions decrease exponentially (with possible multiplicative power corrections) in $r$. In addition, the dimer correlation functions exhibit long range orders.
  }
  \label{fig:spin-dimer-massive}
\end{figure*}

We have used the meron-cluster algorithm to compute the equal time spin and dimer correlation functions defined through the expressions
\begin{align}
\tilde{G}_{i}(r) =
\frac{1}{Z}\Tr\Big(\e^{-\beta H}{\cal O}_i(r) {\cal O}_i(0)\Big),
\end{align}
where ${\cal O}_i(r),i=S,D$ are the spin and dimer operators defined in \cref{eq:spin,eq:dimer}. The symbol $r$ is used for the spatial lattice site $j$ at some fixed time slice. In our work we choose $\beta = L$ where $L$ is the size of our spatial lattice in lattice units. Based on our estimate for the speed of light this leads to a small effective temperature. In particular we have evidence that physical temperatures and physical length scales are related by $T_{\rm phys} \approx 0.25 (La)^{-1}$ in our lattice model.
We compute $\tilde{G}_i(r)$ at various lattice sizes and $U$. In \cref{fig:spin-dimer-massive} we illustrate some of our results at $U=0$, $0.5$, $1$ and $1.5$ on a lattice size of $L=128$. For clarity we focus on the region of $8 \leq r \leq 40$. The figure clearly shows a non-zero expectation value for the dimer order parameter $\langle D_j(t)\rangle$ at $U=0$ while there is no such expectation value for the spin operator. This is consistent with our expectations that for small values of $U$ the lattice model breaks the $\mathbb{Z}_2^\chi$ chiral symmetry as explained in \cref{sec:bosonization}. Another important point to note is that the spin and dimer correlation functions also behave very differently at least for small values of $U$ but slowly begin to become similar by $U \approx 1.5$.

Assuming we are in the vicinity of the critical point around $U \gtrsim 1.5$, we want to fit our MC data to \cref{eq:spin_ir,eq:dimer_ir}. However, since we work on a finite lattice with periodic boundary conditions, the correlation functions receives finite size corrections. Fortunately, in the conformal phase the finite size corrections can be obtained using the map from an infinite plane to a cylinder, which results in replacing $r$ by
\begin{align}
  r \rightarrow \tilde{r} = \frac{L}{\pi}\sin\frac{\pi r}{L},
\end{align}
where $L$ is the spatial size \cite{Cardy:1984rp,Affleck_1998}. Furthermore, the spin correlation functions clearly shows oscillating behavior due to the higher order operators in the observable with ferromagnetic behavior. Taking these corrections into account, we make the following ansatz for the lattice correlation functions at $U \gtrsim U_c$ on a finite lattice,
\begin{align}
  G_S(\tilde{r}) &= \frac{A_1}{\tilde{r}}
  \left(1-\frac{\tilde{\lambda}_0}{2\pi}\ln\tilde{r}\right)^{\frac{1}{2}} -\  \frac{A_2(-1)^r}{\tilde{r}^2}, \label{eq:critfitspin}  \\
  G_D(\tilde{r}) &= \frac{B}{\tilde{r}}\left(1-\frac{\tilde{\lambda}_0}{2\pi}\ln\tilde{r}\right)^{-\frac{3}{2}},
  \label{eq:critfitdimer}
\end{align}
where we have introduced a single fit parameter $\tilde{\lambda}_0 = \lambda_0/(1+\frac{\lambda_0}{2\pi}\ln \tilde{r}_0)$, which is an RG invariant quantity and numerically equals the coupling measured at $\tilde{r}_0 = 1$ in the lattice unit. The critical point is obtained when $\tilde{\lambda}_0 = 0$.

\begin{table*}[htb]
  \centering
  \renewcommand{\arraystretch}{1.4}
  \setlength{\tabcolsep}{0pt}
  \begin{tabular}{S[table-format=2.4]|S[table-format=2.6]S[table-format=2.6]S[table-format=2.7]S[table-format=3.6]S[table-format=2.3]|S[table-format=2.6]S[table-format=2.6]S[table-format=3.6]S[table-format=2.3]|S[table-format=2.8]S[table-format=3.6]S[table-format=2.3]}
  \TopRule
  & \multicolumn{5}{c|}{Combined fit} & \multicolumn{4}{c|}{Spin fit} & \multicolumn{3}{c}{Dimer fit} \\
  \MidRule
  {$U$}~~ & {$A_1$} & {$A_2$} & {$B$} & {$\tilde\lambda_0$} & {$\chi_\nu^2$} & 
  {$A_1$} & {$A_2$} & {$\tilde\lambda_0^s$} & { $\chi_\nu^2$} & {$B$} & {$\tilde\lambda_0^d$} & { $\chi_\nu^2$} \\
  \MidRule
  1.5 & 0.0816(2) & 0.0241(4) & 0.00747(4) & 0.333(7) & 0.55 & 
  0.0832(3) & 0.0245(4) & 0.409(11) & 0.20 & 
  0.00770(5) & 0.294(9) & 0.32 \\
  1.6 & 0.0836(2) & 0.0247(4) & 0.00766(4) & 0.230(7) & 0.69 & 
  0.0845(3) & 0.0250(4) & 0.273(12) & 0.73 & 
  0.00780(5) & 0.204(9) & 0.46 \\
  1.7 & 0.0838(1) & 0.0248(4) & 0.00818(3) & 0.066(6) & 1.10 & 
  0.0854(2) & 0.0252(4) & 0.145(12) & 0.31 & 
  0.00828(4) & 0.047(6) & 1.44 \\
  1.745 & 0.0839(1) & 0.0249(4) & 0.00847(3) & -0.019(6) & 1.23 & 
  0.0858(3) & 0.0253(4) & 0.082(14) & 0.97 & 
  0.00857(4) & -0.037(7) & 1.02 \\
  1.8 & 0.0852(1) & 0.0253(3) & 0.00846(3) & -0.048(6) & 1.10 & 
  0.0851(2) & 0.0252(4) & -0.054(13) & 0.75 & 
  0.00845(4) & -0.046(7) & 0.30 \\
  1.9 & 0.0854(2) & 0.0253(3) & 0.00890(5) & -0.195(9) & 1.32 & 
  0.0867(2) & 0.0255(3) & -0.123(13) & 1.21 & 
  0.00922(7) & -0.254(13) & 0.99 \\
  2.0 & 0.0871(2) & 0.0257(3) & 0.00894(5) & -0.258(10) & 0.78 & 
  0.0875(3) & 0.0258(3) & -0.234(15) & 0.32 & 
  0.00904(7) & -0.278(14) & 1.20 \\
  2.5 & 0.0911(2) & 0.0262(3) & 0.00982(6) & -0.670(11) & 1.85 &
  0.0929(2) & 0.0265(3) & -0.560(14) & 2.03 &
  0.01054(10) & -0.812(18) & 0.56 \\
  4.0 & 0.0997(2) & 0.0268(3) & 0.01129(6) & -1.345(13) & 1.90 &
  0.1014(2) & 0.0271(3) & -1.222(15) & 1.74 &
  0.01264(13) & -1.628(27) & 0.31 \\
  $\infty$~~~ & 0.1110(2) & 0.0269(2) & 0.01322(5) & -1.975(10) & 9.68 &
  0.1243(2) & 0.0276(2) & -1.683(11) & 3.94 &
  0.01601(11) & -2.548(21) & 0.31 \\
  \BotRule
\end{tabular}
\caption{Parameters in \cref{eq:critfitspin,eq:critfitdimer} obtained by fitting the MC data for $U \gtrsim U_c = 1.75(5)$.}
\label{tab:fitting}
\end{table*}

We have performed a combined fit for both $G_S(\tilde{r})$ and $G_D(\tilde{r})$ to the form \cref{eq:critfitspin,eq:critfitdimer} at each fixed value of $U \geq 1.5$, which is tabulated in the combined fit column of \cref{tab:fitting}. Each of these is a four parameter fit involving $A_1,A_2,B, \tilde{\lambda}_0$ and uses all data from $L = 64$, $80$, $96$, $128$ and $12 \leq r \leq 40$. As can be seen from these results, our data fit well to the form \cref{eq:critfitspin,eq:critfitdimer} for all couplings in the range $ 1.5 \leq U \leq 2.0$. Since at the critical point we expect $\tilde{\lambda}_0 \approx 0$, applying conservative systematic errors related to our fitting procedures we estimate $U_c = 1.75(5)$. In order to visualize that the correlation functions indeed obeys the power law near the critical point, we plot them in $\log$ scale in \cref{fig:spin-dimer-critical}. We only plot at even $r$ in order to avoid the distracting oscillation.

\begin{figure*}[htb]
  \includegraphics[width=0.49\textwidth]{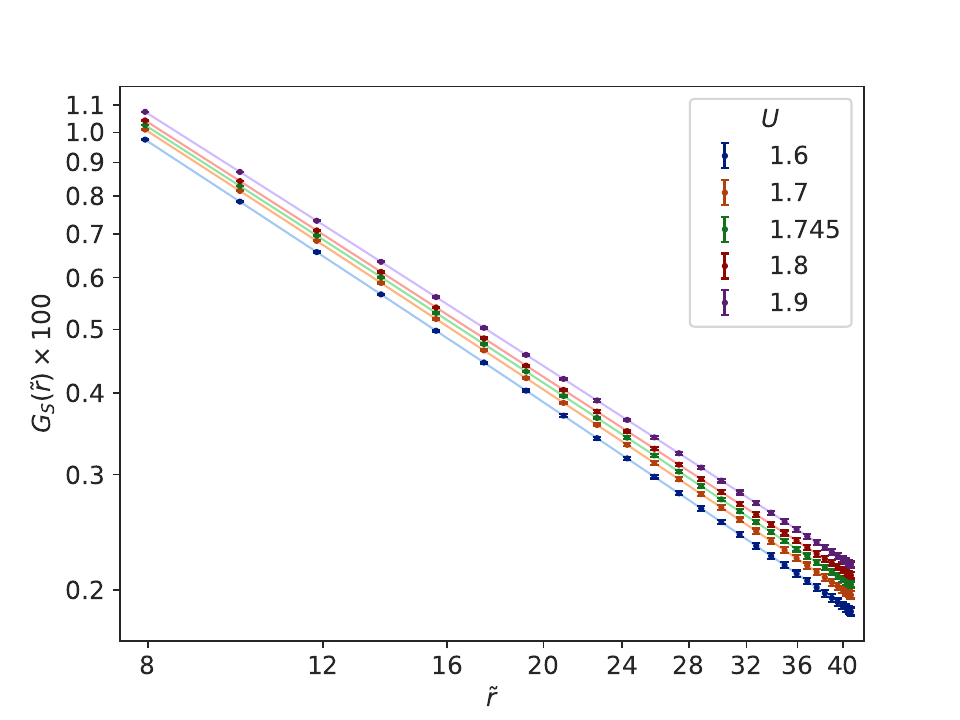}
  \includegraphics[width=0.49\textwidth]{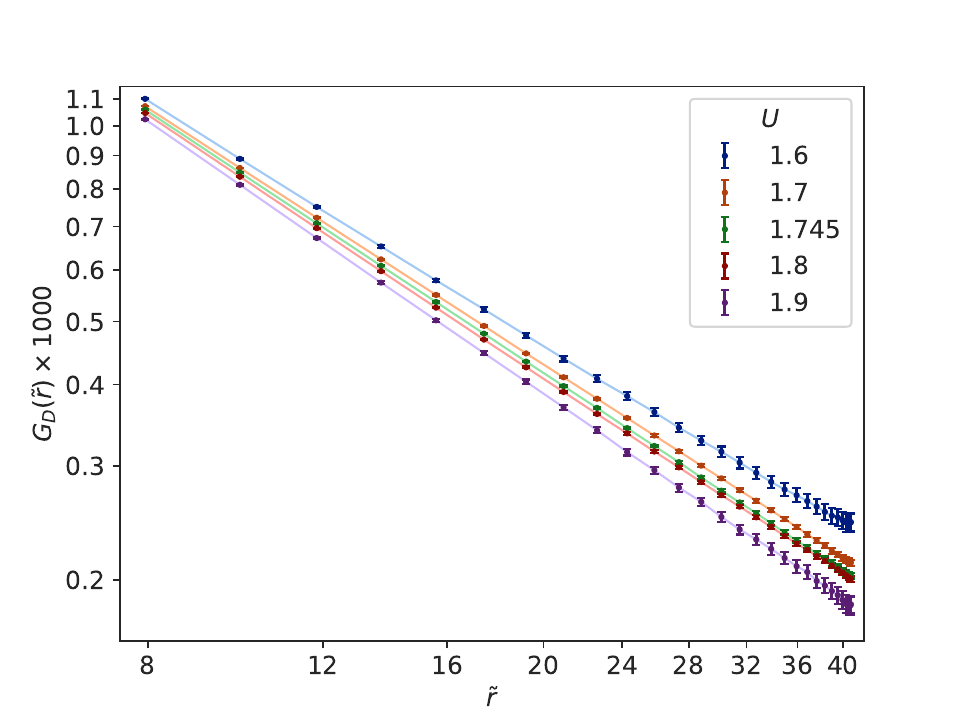}
  \caption{Spin correlation functions $G_S(\tilde{r})$ and dimer correlation functions $G_D(\tilde{r})$ as functions of $\tilde{r}$ at even $r$ and $L=128$ for $U=1.6$, $1.7$, $1.745$, $1.8$ and $1.9$. Both correlation functions decrease algebraically in $\tilde{r}$ with $\log$ corrections.
  }
  \label{fig:spin-dimer-critical}
\end{figure*}

The quality of the combined fit column in \cref{tab:fitting} becomes poor for $U > 2.0$, and changing the fitting range of $r$ does not seem to improve the fits much. In fact the fit seems to have a very bad reduced chi-squared $\chi^2_\nu = 9.68$ when $U\rightarrow\infty$. This seems a bit surprising since the forms described by \cref{eq:critfitspin,eq:critfitdimer} must work well in the entire conformal phase. One reason is that our data is very precise at $U=\infty$, where $H$ is the same as the Heisenberg spin chain, and therefore the results become sensitive to higher order terms in the lattice model which are ignored in the theoretical analysis. In this limit the meron clusters have the same weight as the non-merons, and therefore we do not need to check merons and the algorithm is much more efficient. However, we believe there is more to the story here. Fortunately, there are precise asymptotic results for the spin and dimer correlation functions in the Heisenberg spin chain \cite{Affleck_1998,PhysRevB.96.134429,PhysRevB.94.014417}:
\begin{align}
  G_S(\tilde{r}\rightarrow \infty) &= \frac{\left(\ln \tilde{r}\right)^{\frac{1}{2}}}{(2\pi)^{3/2}\tilde{r}} - \frac{(-1)^r}{(2\pi)^{2} \tilde{r}^2},\label{eq:Heisenberg_spin} \\
  G_D(\tilde{r}\rightarrow \infty) &= \frac{\left(\ln \tilde{r}\right)^{-\frac{3}{2}}}{(2\pi)^{3/2}\tilde{r}} + \frac{(-1)^r\left(\ln \tilde{r}\right)^2}{6\pi^4 \tilde{r}^4},\label{eq:Heisenberg_dimer}
\end{align}
and indeed they are consistent with \cref{eq:critfitspin,eq:critfitdimer} in the $\tilde{r}\rightarrow\infty$ limit. More precisely in the limit $U\rightarrow\infty$, we must observe the following constraints among the fit parameters $2\pi A_1 (-\tilde{\lambda}_0)^{1/2} \stackrel{!}{=} 4\pi^2 A_2 \stackrel{!}{=} 8\pi^3 B(-\tilde{\lambda}_0)^{-3/2} \stackrel{!}{=} 1$. However results from the combined fit give us $2\pi A_1 (-\tilde{\lambda}_0)^{1/2} \approx 1.06$, $4\pi^2 A_2 \approx 1.06$ and $8\pi^3 B(-\tilde{\lambda}_0)^{-3/2} \approx 1.18$. While the first two constraints are not far from expectations, the last one seems quite a bit off. We want to check how the fits change if we allow for $\tilde\lambda_0$ to be different in the spin and dimer correlation functions. Therefore we performed separate fits of our data, which are tabulated in the spin fit and dimer fit columns in \cref{tab:fitting}, and plotted in \cref{fig:spin-dimer-Heisenberg}. Note that in the vicinity of the critical point both values tend to become small as expected. Focusing on the $U\rightarrow\infty$ limit, the separate fits now suggest $2\pi A_1 (-\tilde{\lambda}_0^s)^{1/2} \approx 1.01$, $4\pi^2 A_2 \approx 1.09$ and $8\pi^3 B(-\tilde{\lambda}_0^d)^{-3/2} \approx 0.98$. Although the deviation in $A_2$ still seems to be large, this is understandable because $A_2$ is a higher order correction. Now they are in much better agreement with the above constraints, suggesting that for some reason that we do not yet understand the values of $\tilde{\lambda}_0$ for spin and dimer begin to drift apart as we go away from the critical point.
\begin{figure*}[htp]
  \includegraphics[width=0.49\textwidth]{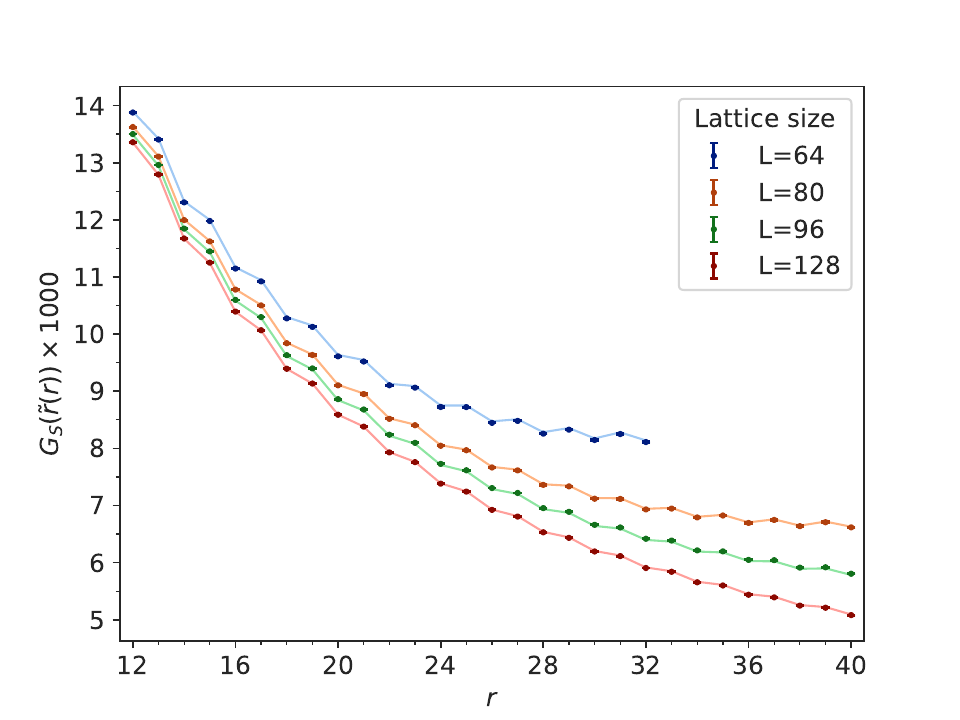}
  \includegraphics[width=0.49\textwidth]{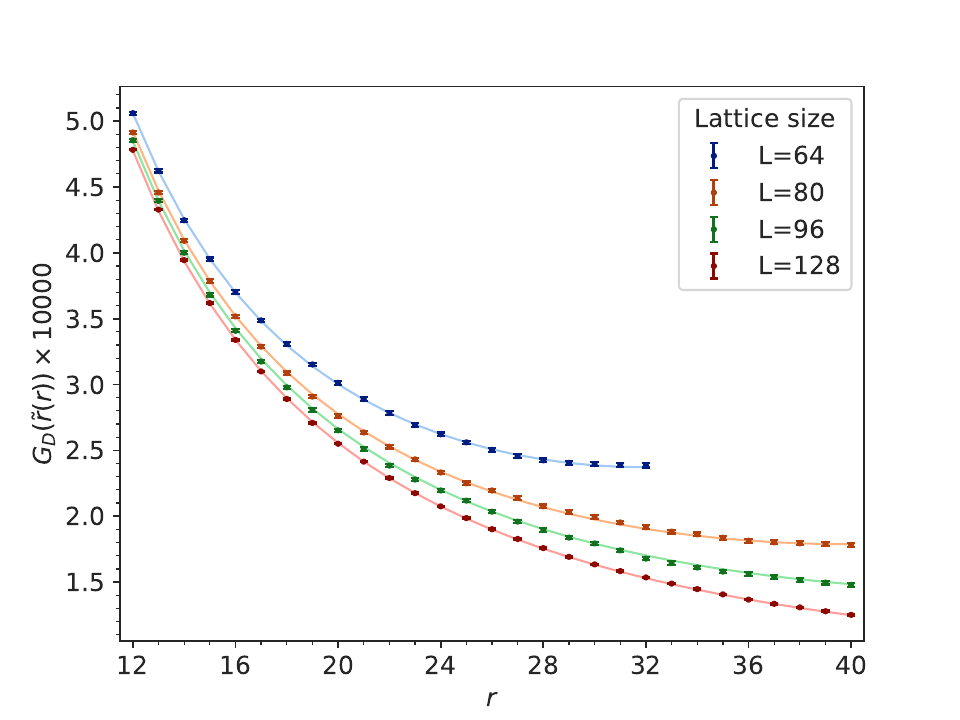}
  \caption{MC results for the spin and dimer correlation functions $G_S(\tilde r(r))$(left) and $G_D(\tilde r(r))$(right) obtained at $U=\infty$ are shown as function of $r$. The solid line that passes through the data are fits to \cref{eq:critfitspin,eq:critfitdimer}. 
  }
  \label{fig:spin-dimer-Heisenberg}
\end{figure*}

\subsection{Exact Diagonalization Results}\label{sec:exactdiag}

In order to confirm that our estimate for the critical point is reasonable, we also use an alternate idea outlined in \cite{OKAMOTO1992433}, based on the low energy physics of the WZW model \cref{eq:k1wzw} that emerges at the critical point. Since the WZW model is invariant under the enhanced chiral transformations with two independent $SU(2)$ (``left'' and ``right'') symmetries, the energy eigenstates can be labeled with spin quantum numbers $(s_L,s_R)$. It is known that the ground state has $(s_L, s_R) = (0,0)$ with momentum $0$, while the first excited state has $(s_L, s_R) = (\frac{1}{2},\frac{1}{2})$ with momentum $\pi$ \cite{Affleck:1988px}. However, since the lattice model is only invariant under the diagonal $SU(2)$ subgroup, the energy eigenstates on the lattice will only be labeled by the total spin $s_\text{tot}$. The four-fold degeneracy requires fine tuning to the critical point where the singlet ($s_\text{tot}=0$) and the triplet ($s_\text{tot} = 1$) state together form an $(s_L, s_R) = (\frac{1}{2},\frac{1}{2})$ state. This suggests an alternative method to determine the critical point: we can compute the lowest five energy eigenvalues as a function of $U$ and $L$ using an exact diagonalization method and locate the coupling where the first excited state becomes four-fold degenerate. When $L$ is finite, the critical coupling where the energies of these two total spin sectors cross can be referred to as a pseudo-critical point $U_c(L)$. This point turns into the true critical point as $L$ becomes large.

In order to implement the above idea we have computed the first five eigen-energies by diagonalizing the Hamiltonian on small lattices with the coordinate descent method \cite{Wang:2019a,Wang:2019b}. The behavior of the lowest five states as a function of $U$ at $L=14$ and $L=16$ is plotted in \cref{fig:exactd}. We observe that in the broken phase (small $U$) the ground state and the first excited state turn out to be spin singlets with $s_\text{tot} =0$. The next three degenerate excited states form a triplet with $s_\text{tot} =1$. In contrast when $U > U_c$, the triplet states have lower energy as compared to the singlet state. We thus can determine $U_c(L)$ as the coupling where the triplet and singlet states cross, which are tabulated in \cref{tab:U_cL} and plotted in \cref{fig:pUcL} as a function of $1/L^2$. Based on \cref{fig:pUcL} we can estimate $U_c(L) \approx 1.746(1)$, which is consistent with our estimate in the previous section and with the estimate using finite size scaling in the conference proceeding published earlier \cite{Liu:2019dvk}.

\begin{table}[htb]
  \centering
  \renewcommand{\arraystretch}{1.2}
  \setlength{\tabcolsep}{4pt}
  \begin{tabular}{S[table-format=3.1]S[table-format=2.9]|S[table-format=3.1]S[table-format=2.9]}
    \TopRule
    {$L$} & {$U_c(L)$} & {$L$} & {$U_c(L)$} \\ \MidRule
    4 & 2.53982388 & 6 & 1.26699259 \\
    8 & 1.86035876 & 10 & 1.68834326 \\
    12 & 1.75726171 & 14 & 1.73529271 \\
    16 & 1.74508577 & 18 & 1.74308979 \\
    20 & 1.74502594 & 22 & 1.745311  \\ 
    24 & 1.745996 & 26 & 1.746393 \\ \BotRule
  \end{tabular}
  \caption{Pseudo-critical couplings $U_c(L)$ as a function of $L$ obtained using the exact diagonalization method.}
  \label{tab:U_cL}
\end{table}

\begin{figure*}[ht]
\centering 
\includegraphics[width=0.49\textwidth]{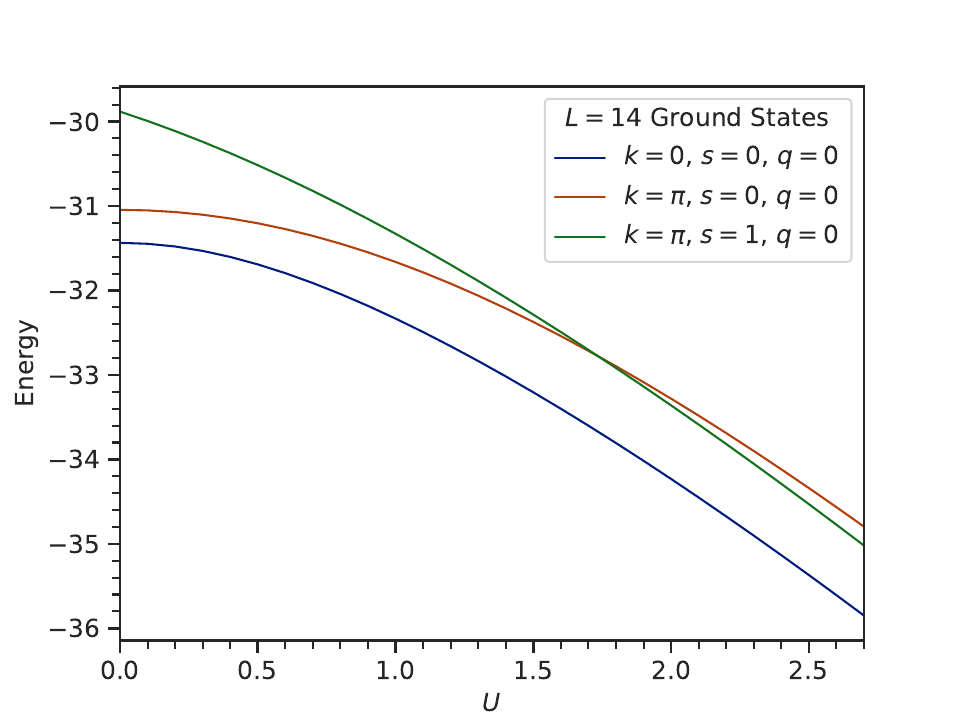}
\includegraphics[width=0.49\textwidth]{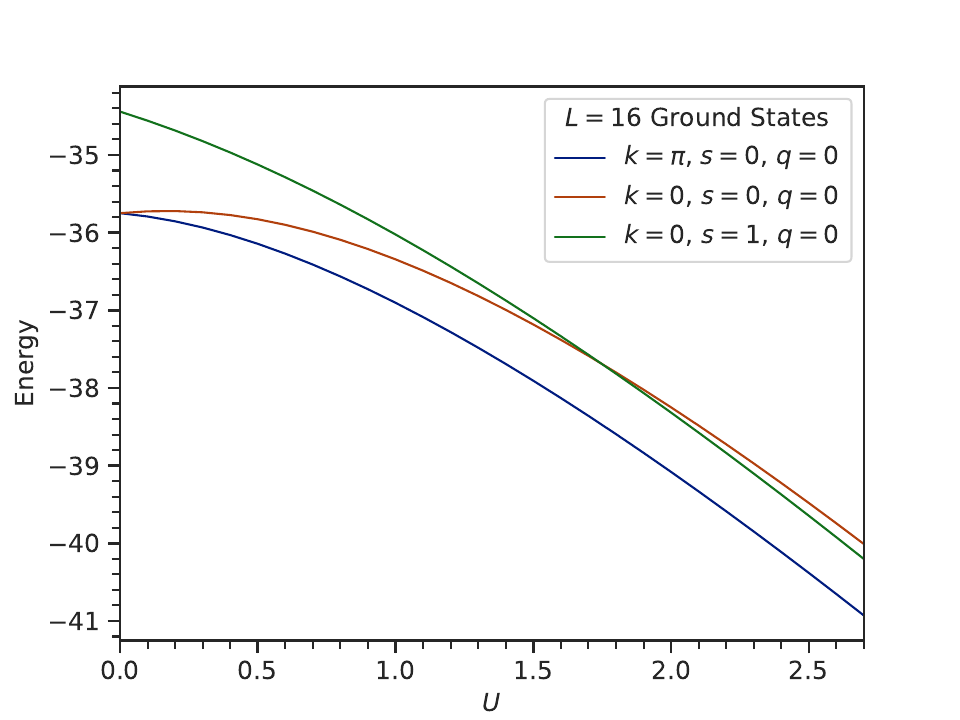}
\caption{The plot of the lowest five energy eigenvalues obtained using an exact diagonalization method as a function of $U$ at $L=14$ and $L=16$. Note that only three eigenvalues are shown because the $s=1$ states are three fold degenerate.}
\label{fig:exactd}
\end{figure*}

\begin{figure}[ht]
\centering 
\includegraphics[width=0.49\textwidth]{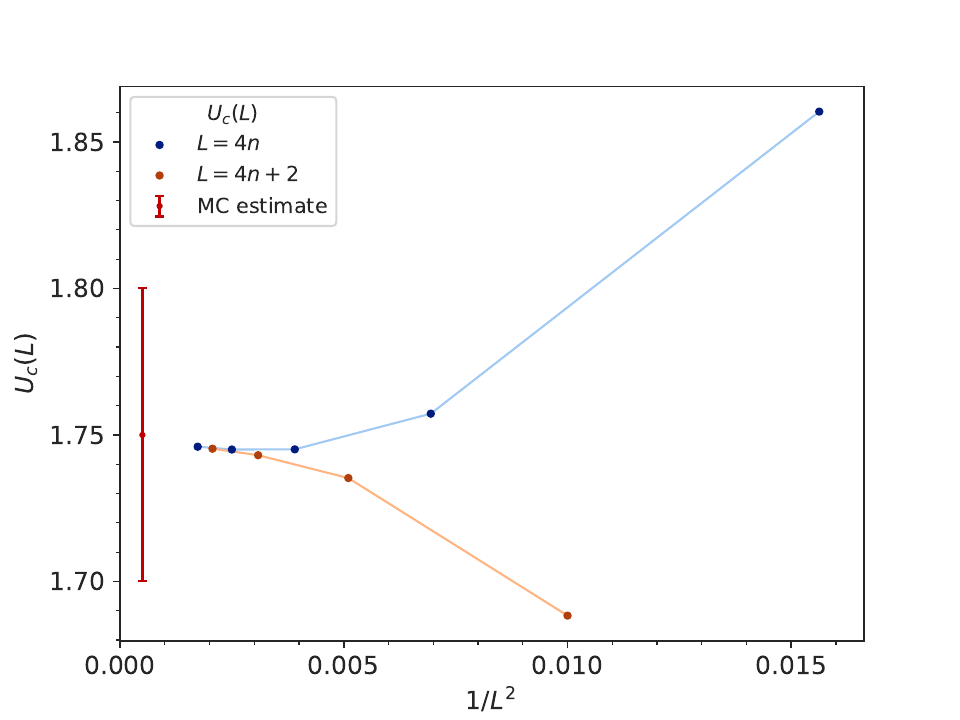}
\caption{Plot of the behavior of the pseudo-critical coupling $U_c(L)$ as a function of $1/L^2$. The values of the couplings for this plot are given in \cref{tab:U_cL}.}
\label{fig:pUcL}
\end{figure}

Looking more closely at \cref{fig:exactd} we observe some peculiarities. For example at $U=0$ we note that the ground state is degenerate at $L=16$ but not at $L=14$. We have explained the reason for this degeneracy in \cref{app:degeneracy}. In fact we show that there is an interesting difference in the energy spectrum when the lattice size is $L=4n$ versus when it is $L=4n-2$, where $n=1,2,\cdots$. Moreover, this difference permeates even away from $U=0$ and is observed in the dramatically different values of $U_c(L)$ when $L$ is small. The difference however decreases rapidly as $L$ increases, and both of them approach to the true critical point as expected. Another peculiarity comes from the momentum quantum number $k$ for the first five energy eigenstates. We note that $k$ flips between $\pi$ and $0$ when comparing $L=4n$ with $4n+2$. This is as expected because similar phenomena are also observed in the Heisenberg spin chain, with $k = 0$ for $L=4n$ and $k = \pi$ for $L=4n+2$ \cite{Karbach:1998}. The extra $\pi$ phase in our model compared to the Heisenberg spin chain is due to the fermionic nature of our model, since there is an intrinsic extra minus sign when the system is translated by one site when $L = 2n$.

While it is not easy to extend our meron-cluster algorithm for the more general Hamiltonian with the parameter $\varepsilon$ introduced in \cref{eq:modHJ}, we can extend the above exact diagonalization method to determine the critical point for arbitrary $\varepsilon$. For example, when $\varepsilon = 0.1$ and $J=1$, $U_c(L)$ is tabulated in \cref{tab:U_cJ0.1}. When $\varepsilon$ is small the perturbative analysis should be a good guide and we obtained at leading order $U_c= 4 J \varepsilon = 0.4$ in \cref{sec:continuum}, which is in good agreement with \cref{tab:U_cJ0.1}.

\begin{table}[htb]
  \centering
  \renewcommand{\arraystretch}{1.2}
  \setlength{\tabcolsep}{4pt}
  \begin{tabular}{S[table-format=3.1]S[table-format=2.9]|S[table-format=3.1]S[table-format=2.9]}
    \TopRule
    {$L$} & {$U_c(L)$} & {$L$} & {$U_c(L)$} \\ \MidRule
    4 & 0.399755 & 6 & 0.195927 \\ 
    8 & 0.397655 & 10 & 0.310813 \\
    12 & 0.396644 & 14 & 0.347239 \\
    16 & 0.396076 & 18 & 0.3636  \\ \BotRule
  \end{tabular}
  \caption{Pseudo-critical couplings $U_c(L)$ as a function of $L$ obtained using the exact diagonalization method for the model defined by \cref{eq:modHJ} at $\varepsilon=0.1$ and $J = 1$.}
  \label{tab:U_cJ0.1}
\end{table}

Using the exact diagonalization method we have also confirmed that our model at $U = 0$ is indeed in a massive phase with spontaneously broken $\mathbb{Z}_2^\chi$ symmetry. In such phase the energy gap to the first excited state is expected to become exponentially small as $L$ becomes large, while the gap to the second excited state remains non-zero at $L\rightarrow\infty$. We plot these gaps in \cref{fig:gaps-L}, where the solid lines are exponential fits. These features are qualitatively visible, along with the peculiar degeneracy in our model. 

\begin{figure}[htb]
\includegraphics[width=0.48\textwidth]{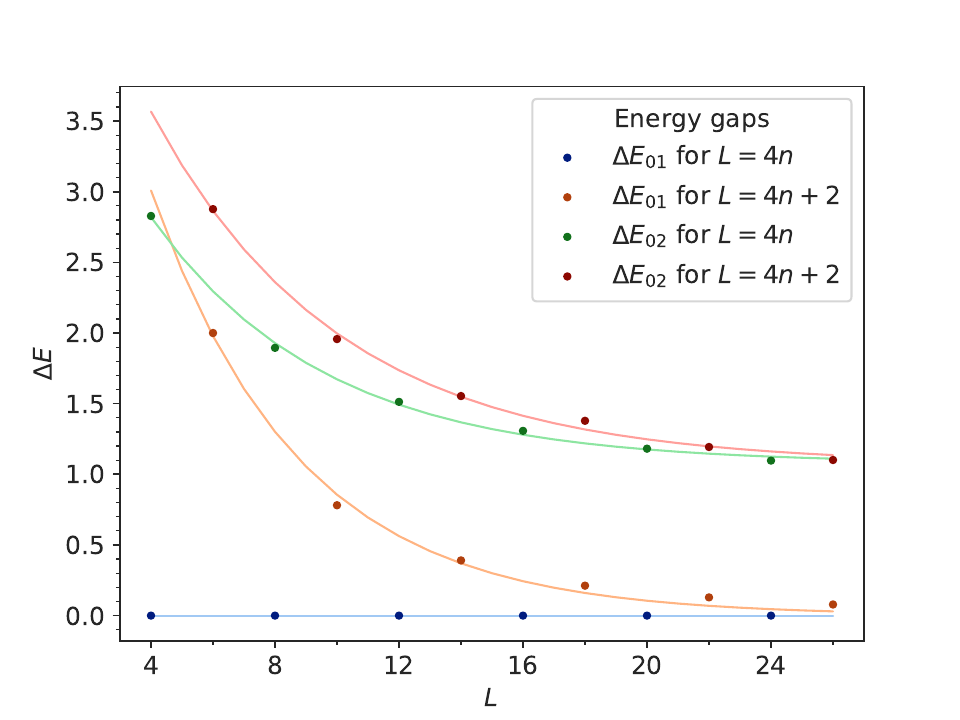}
\caption{Scaling of energy gaps of the first and second excited states with the ground state at $U = 0$. The gaps of the first excited states close at $L\rightarrow\infty$ for both $L = 4n$ and $4n+2$. At finite $L$, the former is identically zero, while the latter closes exponentially. The gaps of the second excited states stay open for both $L = 4n$ and $4n+2$. The curves show exponential fits.}
\label{fig:gaps-L}
\end{figure}

\section{Summary and Conclusions}
\label{conclusions}

In this work we constructed a strongly correlated lattice fermion Hamiltonian that was solvable by the meron-cluster algorithm. We were able to add the Hubbard coupling $U$ to our model without losing this property. This combined lattice model had $SO(4)\times \mathbb{Z}_2^\chi$ symmetry for all values of $U$, and an extra spin-charge flip symmetry $\mathbb{Z}_2^\chi$ at $U = 0$. While the lattice model can be formulated in any dimension, here we studied it in one spatial dimension. In order to study the phase diagram of the model as a function of $U$ we used non-abelian bosonization to relate our model to two decoupled $SU(2)_1$ WZW models with marginal couplings that depend on $U$. In particular we discovered that the model undergoes a quantum phase transition between a massive phase where the $\mathbb{Z}_2^\chi$ symmetry is spontaneously broken and a conformal phase without such symmetry breaking. The massive phase is the well known asymptotically free chiral-mass GN model with massive fermions and spontaneous symmetry breaking, while the conformal phase is the single $SU(2)_1$ WZW model. Using the meron-cluster method and exact diagonalization we provided further evidence of this scenario.

We view the current study as just a first step in demonstrating that meron-cluster algorithms can be useful to study phase diagrams of lattice fermion systems where fermions become massive, while the low energy bosonic physics is still interesting either due to frustration or the presence of topological terms. Extending our model to $2+1$ dimensions would be very interesting. For example in \cite{Li:2019acc} it was argued that in an extended Hubbard model very similar to ours, there is a direct second order phase transition between a VBS phase and an antiferromagnetic phase, and an enhanced $SO(5)$ symmetry is expected to emerge at the transition. The critical exponents at this exotic transition were computed in the earlier studies, but unfortunately they do not match those obtained from a loop gas formulation of the same phase transition \cite{PhysRevLett.115.267203}. On the other hand the critical exponents of the fermionic realization seem consistent with the bounds obtained from conformal bootstrap \cite{Poland:2018epd}. As far as we know, this controversy remains unresolved, partially because the fermionic model was studied on rather small lattices with less than $500$ lattice sites, due to difficulties associated with auxiliary field fermion MC algorithm. We are currently exploring if our model and its extensions allow us to study this transition more efficiently using the meron-cluster approach.

\begin{acknowledgments}
We would like to thank Emilie Huffman and Hersh Singh for helpful discussions and pointing us to the literature on interesting fermionic quantum critical behavior. SC would like to thank Uwe-Jens Wiese for discussions. HL would like to thank Zhe Wang for sharing his computer codes on the coordinate descend method, and Xin Zhang for helpful discussions. The material presented here is supported by the U.S. Department of Energy, Office of Science, Nuclear Physics program under Award Numbers DE-FG02-05ER41368. RKK was supported in part by NSF award DMR-1611161. This work used computational resources provided by the Extreme Science and Engineering Discovery Environment (XSEDE) \cite{xsede}, which is supported by National Science Foundation grant number ACI-1548562.
\end{acknowledgments}

\addcontentsline{toc}{section}{References}
\bibliography{Refs,GN,MC}

\begin{thebibliography}{82}%
\makeatletter
\providecommand \@ifxundefined [1]{%
 \@ifx{#1\undefined}
}%
\providecommand \@ifnum [1]{%
 \ifnum #1\expandafter \@firstoftwo
 \else \expandafter \@secondoftwo
 \fi
}%
\providecommand \@ifx [1]{%
 \ifx #1\expandafter \@firstoftwo
 \else \expandafter \@secondoftwo
 \fi
}%
\providecommand \natexlab [1]{#1}%
\providecommand \enquote  [1]{``#1''}%
\providecommand \bibnamefont  [1]{#1}%
\providecommand \bibfnamefont [1]{#1}%
\providecommand \citenamefont [1]{#1}%
\providecommand \href@noop [0]{\@secondoftwo}%
\providecommand \href [0]{\begingroup \@sanitize@url \@href}%
\providecommand \@href[1]{\@@startlink{#1}\@@href}%
\providecommand \@@href[1]{\endgroup#1\@@endlink}%
\providecommand \@sanitize@url [0]{\catcode `\\12\catcode `\$12\catcode
  `\&12\catcode `\#12\catcode `\^12\catcode `\_12\catcode `\%12\relax}%
\providecommand \@@startlink[1]{}%
\providecommand \@@endlink[0]{}%
\providecommand \url  [0]{\begingroup\@sanitize@url \@url }%
\providecommand \@url [1]{\endgroup\@href {#1}{\urlprefix }}%
\providecommand \urlprefix  [0]{URL }%
\providecommand \Eprint [0]{\href }%
\providecommand \doibase [0]{http://dx.doi.org/}%
\providecommand \selectlanguage [0]{\@gobble}%
\providecommand \bibinfo  [0]{\@secondoftwo}%
\providecommand \bibfield  [0]{\@secondoftwo}%
\providecommand \translation [1]{[#1]}%
\providecommand \BibitemOpen [0]{}%
\providecommand \bibitemStop [0]{}%
\providecommand \bibitemNoStop [0]{.\EOS\space}%
\providecommand \EOS [0]{\spacefactor3000\relax}%
\providecommand \BibitemShut  [1]{\csname bibitem#1\endcsname}%
\let\auto@bib@innerbib\@empty
\bibitem [{\citenamefont {Rosenstein}\ \emph {et~al.}(1991)\citenamefont
  {Rosenstein}, \citenamefont {Warr},\ and\ \citenamefont
  {Park}}]{Rosenstein:1990nm}%
  \BibitemOpen
  \bibfield  {author} {\bibinfo {author} {\bibfnamefont {B.}~\bibnamefont
  {Rosenstein}}, \bibinfo {author} {\bibfnamefont {Brian}\ \bibnamefont
  {Warr}}, \ and\ \bibinfo {author} {\bibfnamefont {S.~H.}\ \bibnamefont
  {Park}},\ }\bibfield  {title} {\enquote {\bibinfo {title} {{Dynamical
  symmetry breaking in four Fermi interaction models}},}\ }\href {\doibase
  10.1016/0370-1573(91)90129-A} {\bibfield  {journal} {\bibinfo  {journal}
  {Phys. Rept.}\ }\textbf {\bibinfo {volume} {205}},\ \bibinfo {pages}
  {59--108} (\bibinfo {year} {1991})}\BibitemShut {NoStop}%
\bibitem [{\citenamefont {Zinn-Justin}(1991)}]{ZinnJustin:1991yn}%
  \BibitemOpen
  \bibfield  {author} {\bibinfo {author} {\bibfnamefont {Jean}\ \bibnamefont
  {Zinn-Justin}},\ }\bibfield  {title} {\enquote {\bibinfo {title} {{Four
  fermion interaction near four-dimensions}},}\ }\href {\doibase
  10.1016/0550-3213(91)90043-W} {\bibfield  {journal} {\bibinfo  {journal}
  {Nucl. Phys. B}\ }\textbf {\bibinfo {volume} {367}},\ \bibinfo {pages}
  {105--122} (\bibinfo {year} {1991})}\BibitemShut {NoStop}%
\bibitem [{\citenamefont {Gross}\ and\ \citenamefont
  {Neveu}(1974)}]{Gross:1974jv}%
  \BibitemOpen
  \bibfield  {author} {\bibinfo {author} {\bibfnamefont {David~J.}\
  \bibnamefont {Gross}}\ and\ \bibinfo {author} {\bibfnamefont {Andre}\
  \bibnamefont {Neveu}},\ }\bibfield  {title} {\enquote {\bibinfo {title}
  {{Dynamical Symmetry Breaking in Asymptotically Free Field Theories}},}\
  }\href {\doibase 10.1103/PhysRevD.10.3235} {\bibfield  {journal} {\bibinfo
  {journal} {Phys. Rev. D}\ }\textbf {\bibinfo {volume} {10}},\ \bibinfo
  {pages} {3235} (\bibinfo {year} {1974})}\BibitemShut {NoStop}%
\bibitem [{\citenamefont {Thirring}(1958)}]{Thirring:1958in}%
  \BibitemOpen
  \bibfield  {author} {\bibinfo {author} {\bibfnamefont {Walter~E.}\
  \bibnamefont {Thirring}},\ }\bibfield  {title} {\enquote {\bibinfo {title}
  {{A Soluble relativistic field theory?}}}\ }\href {\doibase
  10.1016/0003-4916(58)90015-0} {\bibfield  {journal} {\bibinfo  {journal}
  {Annals Phys.}\ }\textbf {\bibinfo {volume} {3}},\ \bibinfo {pages} {91--112}
  (\bibinfo {year} {1958})},\ \bibinfo {note} {[,509(1958)]}\BibitemShut
  {NoStop}%
\bibitem [{\citenamefont {Gies}\ and\ \citenamefont
  {Janssen}(2010)}]{Gies:2010st}%
  \BibitemOpen
  \bibfield  {author} {\bibinfo {author} {\bibfnamefont {Holger}\ \bibnamefont
  {Gies}}\ and\ \bibinfo {author} {\bibfnamefont {Lukas}\ \bibnamefont
  {Janssen}},\ }\bibfield  {title} {\enquote {\bibinfo {title} {{UV fixed-point
  structure of the three-dimensional Thirring model}},}\ }\href {\doibase
  10.1103/PhysRevD.82.085018} {\bibfield  {journal} {\bibinfo  {journal} {Phys.
  Rev. D}\ }\textbf {\bibinfo {volume} {82}},\ \bibinfo {pages} {085018}
  (\bibinfo {year} {2010})},\ \Eprint {http://arxiv.org/abs/1006.3747}
  {arXiv:1006.3747 [hep-th]} \BibitemShut {NoStop}%
\bibitem [{\citenamefont {Janssen}\ and\ \citenamefont
  {Gies}(2012)}]{Janssen:2012pq}%
  \BibitemOpen
  \bibfield  {author} {\bibinfo {author} {\bibfnamefont {Lukas}\ \bibnamefont
  {Janssen}}\ and\ \bibinfo {author} {\bibfnamefont {Holger}\ \bibnamefont
  {Gies}},\ }\bibfield  {title} {\enquote {\bibinfo {title} {{Critical behavior
  of the (2+1)-dimensional Thirring model}},}\ }\href {\doibase
  10.1103/PhysRevD.86.105007} {\bibfield  {journal} {\bibinfo  {journal} {Phys.
  Rev. D}\ }\textbf {\bibinfo {volume} {86}},\ \bibinfo {pages} {105007}
  (\bibinfo {year} {2012})},\ \Eprint {http://arxiv.org/abs/1208.3327}
  {arXiv:1208.3327 [hep-th]} \BibitemShut {NoStop}%
\bibitem [{\citenamefont {Gehring}\ \emph {et~al.}(2015)\citenamefont
  {Gehring}, \citenamefont {Gies},\ and\ \citenamefont
  {Janssen}}]{PhysRevD.92.085046}%
  \BibitemOpen
  \bibfield  {author} {\bibinfo {author} {\bibfnamefont {Friedrich}\
  \bibnamefont {Gehring}}, \bibinfo {author} {\bibfnamefont {Holger}\
  \bibnamefont {Gies}}, \ and\ \bibinfo {author} {\bibfnamefont {Lukas}\
  \bibnamefont {Janssen}},\ }\bibfield  {title} {\enquote {\bibinfo {title}
  {Fixed-point structure of low-dimensional relativistic fermion field
  theories: Universality classes and emergent symmetry},}\ }\href {\doibase
  10.1103/PhysRevD.92.085046} {\bibfield  {journal} {\bibinfo  {journal} {Phys.
  Rev. D}\ }\textbf {\bibinfo {volume} {92}},\ \bibinfo {pages} {085046}
  (\bibinfo {year} {2015})}\BibitemShut {NoStop}%
\bibitem [{\citenamefont {Vafek}\ and\ \citenamefont
  {Vishwanath}(2014)}]{vafek2014:arcmp}%
  \BibitemOpen
  \bibfield  {author} {\bibinfo {author} {\bibfnamefont {Oskar}\ \bibnamefont
  {Vafek}}\ and\ \bibinfo {author} {\bibfnamefont {Ashvin}\ \bibnamefont
  {Vishwanath}},\ }\bibfield  {title} {\enquote {\bibinfo {title} {Dirac
  fermions in solids: From high-tc cuprates and graphene to topological
  insulators and weyl semimetals},}\ }\href {\doibase
  10.1146/annurev-conmatphys-031113-133841} {\bibfield  {journal} {\bibinfo
  {journal} {Annual Review of Condensed Matter Physics}\ }\textbf {\bibinfo
  {volume} {5}},\ \bibinfo {pages} {83--112} (\bibinfo {year} {2014})},\
  \Eprint
  {http://arxiv.org/abs/https://doi.org/10.1146/annurev-conmatphys-031113-133841}
  {https://doi.org/10.1146/annurev-conmatphys-031113-133841} \BibitemShut
  {NoStop}%
\bibitem [{\citenamefont {Herbut}(2006)}]{herbut2006:grprl}%
  \BibitemOpen
  \bibfield  {author} {\bibinfo {author} {\bibfnamefont {Igor~F.}\ \bibnamefont
  {Herbut}},\ }\bibfield  {title} {\enquote {\bibinfo {title} {Interactions and
  phase transitions on graphene’s honeycomb lattice},}\ }\href {\doibase
  10.1103/physrevlett.97.146401} {\bibfield  {journal} {\bibinfo  {journal}
  {Physical Review Letters}\ }\textbf {\bibinfo {volume} {97}} (\bibinfo {year}
  {2006}),\ 10.1103/physrevlett.97.146401}\BibitemShut {NoStop}%
\bibitem [{\citenamefont {Hands}\ \emph
  {et~al.}(1993{\natexlab{a}})\citenamefont {Hands}, \citenamefont {Kocic},\
  and\ \citenamefont {Kogut}}]{Hands:1992be}%
  \BibitemOpen
  \bibfield  {author} {\bibinfo {author} {\bibfnamefont {Simon}\ \bibnamefont
  {Hands}}, \bibinfo {author} {\bibfnamefont {Aleksandar}\ \bibnamefont
  {Kocic}}, \ and\ \bibinfo {author} {\bibfnamefont {John~B.}\ \bibnamefont
  {Kogut}},\ }\bibfield  {title} {\enquote {\bibinfo {title} {{Four Fermi
  theories in fewer than four-dimensions}},}\ }\href {\doibase
  10.1006/aphy.1993.1039} {\bibfield  {journal} {\bibinfo  {journal} {Annals
  Phys.}\ }\textbf {\bibinfo {volume} {224}},\ \bibinfo {pages} {29--89}
  (\bibinfo {year} {1993}{\natexlab{a}})},\ \Eprint
  {http://arxiv.org/abs/hep-lat/9208022} {arXiv:hep-lat/9208022} \BibitemShut
  {NoStop}%
\bibitem [{\citenamefont {Hands}\ \emph
  {et~al.}(1993{\natexlab{b}})\citenamefont {Hands}, \citenamefont {Kocic},\
  and\ \citenamefont {Kogut}}]{Hands:1992ck}%
  \BibitemOpen
  \bibfield  {author} {\bibinfo {author} {\bibfnamefont {Simon}\ \bibnamefont
  {Hands}}, \bibinfo {author} {\bibfnamefont {Aleksandar}\ \bibnamefont
  {Kocic}}, \ and\ \bibinfo {author} {\bibfnamefont {John~B.}\ \bibnamefont
  {Kogut}},\ }\bibfield  {title} {\enquote {\bibinfo {title} {{The Four Fermi
  model in three-dimensions at nonzero density and temperature}},}\ }\href
  {\doibase 10.1016/0550-3213(93)90460-7} {\bibfield  {journal} {\bibinfo
  {journal} {Nucl. Phys. B}\ }\textbf {\bibinfo {volume} {390}},\ \bibinfo
  {pages} {355--378} (\bibinfo {year} {1993}{\natexlab{b}})},\ \Eprint
  {http://arxiv.org/abs/hep-lat/9206024} {arXiv:hep-lat/9206024} \BibitemShut
  {NoStop}%
\bibitem [{\citenamefont {K\"{a}rkk\"{a}inen}\ \emph
  {et~al.}(1994)\citenamefont {K\"{a}rkk\"{a}inen}, \citenamefont {Lacaze},
  \citenamefont {Lacock},\ and\ \citenamefont {Petersson}}]{Karkkainen94}%
  \BibitemOpen
  \bibfield  {author} {\bibinfo {author} {\bibfnamefont {L.}~\bibnamefont
  {K\"{a}rkk\"{a}inen}}, \bibinfo {author} {\bibfnamefont {R.}~\bibnamefont
  {Lacaze}}, \bibinfo {author} {\bibfnamefont {P.}~\bibnamefont {Lacock}}, \
  and\ \bibinfo {author} {\bibfnamefont {B.}~\bibnamefont {Petersson}},\
  }\bibfield  {title} {\enquote {\bibinfo {title} {Critical behaviour of the
  three-dimensional gross-neveu and higgs-yukawa models},}\ }\href {\doibase
  10.1016/0550-3213(94)90309-3} {\bibfield  {journal} {\bibinfo  {journal}
  {Nucl. Phys. B}\ }\textbf {\bibinfo {volume} {415}},\ \bibinfo {pages}
  {781–796} (\bibinfo {year} {1994})}\BibitemShut {NoStop}%
\bibitem [{\citenamefont {Del~Debbio}\ \emph {et~al.}(1997)\citenamefont
  {Del~Debbio}, \citenamefont {Hands},\ and\ \citenamefont
  {Mehegan}}]{DelDebbio:1997dv}%
  \BibitemOpen
  \bibfield  {author} {\bibinfo {author} {\bibfnamefont {L.}~\bibnamefont
  {Del~Debbio}}, \bibinfo {author} {\bibfnamefont {S.J.}\ \bibnamefont
  {Hands}}, \ and\ \bibinfo {author} {\bibfnamefont {J.C.}\ \bibnamefont
  {Mehegan}} (\bibinfo {collaboration} {UKQCD}),\ }\bibfield  {title} {\enquote
  {\bibinfo {title} {{The Three-dimensional Thirring model for small N(f)}},}\
  }\href {\doibase 10.1016/S0550-3213(97)00435-5} {\bibfield  {journal}
  {\bibinfo  {journal} {Nucl. Phys. B}\ }\textbf {\bibinfo {volume} {502}},\
  \bibinfo {pages} {269--308} (\bibinfo {year} {1997})},\ \Eprint
  {http://arxiv.org/abs/hep-lat/9701016} {arXiv:hep-lat/9701016} \BibitemShut
  {NoStop}%
\bibitem [{\citenamefont {Kogut}\ \emph {et~al.}(1998)\citenamefont {Kogut},
  \citenamefont {Lagae},\ and\ \citenamefont {Sinclair}}]{Kogut:1998rg}%
  \BibitemOpen
  \bibfield  {author} {\bibinfo {author} {\bibfnamefont {J.B.}\ \bibnamefont
  {Kogut}}, \bibinfo {author} {\bibfnamefont {J.F.}\ \bibnamefont {Lagae}}, \
  and\ \bibinfo {author} {\bibfnamefont {D.K.}\ \bibnamefont {Sinclair}},\
  }\bibfield  {title} {\enquote {\bibinfo {title} {{Thermodynamics of lattice
  QCD with chiral four fermion interactions}},}\ }\href {\doibase
  10.1103/PhysRevD.58.034504} {\bibfield  {journal} {\bibinfo  {journal} {Phys.
  Rev. D}\ }\textbf {\bibinfo {volume} {58}},\ \bibinfo {pages} {034504}
  (\bibinfo {year} {1998})},\ \Eprint {http://arxiv.org/abs/hep-lat/9801019}
  {arXiv:hep-lat/9801019} \BibitemShut {NoStop}%
\bibitem [{\citenamefont {Del~Debbio}\ and\ \citenamefont
  {Hands}(1999)}]{DelDebbio:1999he}%
  \BibitemOpen
  \bibfield  {author} {\bibinfo {author} {\bibfnamefont {L.}~\bibnamefont
  {Del~Debbio}}\ and\ \bibinfo {author} {\bibfnamefont {S.J.}\ \bibnamefont
  {Hands}},\ }\bibfield  {title} {\enquote {\bibinfo {title} {{The
  Three-dimensional Thirring model for N(f) = 4 and N(f) = 6}},}\ }\href
  {\doibase 10.1016/S0550-3213(99)00258-8} {\bibfield  {journal} {\bibinfo
  {journal} {Nucl. Phys. B}\ }\textbf {\bibinfo {volume} {552}},\ \bibinfo
  {pages} {339--362} (\bibinfo {year} {1999})},\ \Eprint
  {http://arxiv.org/abs/hep-lat/9902014} {arXiv:hep-lat/9902014} \BibitemShut
  {NoStop}%
\bibitem [{\citenamefont {Hands}\ and\ \citenamefont
  {Lucini}(1999)}]{Hands:1999id}%
  \BibitemOpen
  \bibfield  {author} {\bibinfo {author} {\bibfnamefont {Simon}\ \bibnamefont
  {Hands}}\ and\ \bibinfo {author} {\bibfnamefont {Biagio}\ \bibnamefont
  {Lucini}},\ }\bibfield  {title} {\enquote {\bibinfo {title} {{The Phase
  diagram of the three dimensional Thirring model}},}\ }\href {\doibase
  10.1016/S0370-2693(99)00843-6} {\bibfield  {journal} {\bibinfo  {journal}
  {Phys. Lett. B}\ }\textbf {\bibinfo {volume} {461}},\ \bibinfo {pages}
  {263--269} (\bibinfo {year} {1999})},\ \Eprint
  {http://arxiv.org/abs/hep-lat/9906008} {arXiv:hep-lat/9906008} \BibitemShut
  {NoStop}%
\bibitem [{\citenamefont {Christofi}\ and\ \citenamefont
  {Strouthos}(2007)}]{Christofi:2006zt}%
  \BibitemOpen
  \bibfield  {author} {\bibinfo {author} {\bibfnamefont {Stavros}\ \bibnamefont
  {Christofi}}\ and\ \bibinfo {author} {\bibfnamefont {Costas}\ \bibnamefont
  {Strouthos}},\ }\bibfield  {title} {\enquote {\bibinfo {title} {{Three
  dimensional four-fermion models: A Monte Carlo study}},}\ }\href {\doibase
  10.1088/1126-6708/2007/05/088} {\bibfield  {journal} {\bibinfo  {journal}
  {JHEP}\ }\textbf {\bibinfo {volume} {05}},\ \bibinfo {pages} {088} (\bibinfo
  {year} {2007})},\ \Eprint {http://arxiv.org/abs/hep-lat/0612031}
  {arXiv:hep-lat/0612031} \BibitemShut {NoStop}%
\bibitem [{\citenamefont {Wellegehausen}\ \emph {et~al.}(2017)\citenamefont
  {Wellegehausen}, \citenamefont {Schmidt},\ and\ \citenamefont
  {Wipf}}]{Wellegehausen:2017goy}%
  \BibitemOpen
  \bibfield  {author} {\bibinfo {author} {\bibfnamefont {Bj\"orn~H.}\
  \bibnamefont {Wellegehausen}}, \bibinfo {author} {\bibfnamefont {Daniel}\
  \bibnamefont {Schmidt}}, \ and\ \bibinfo {author} {\bibfnamefont {Andreas}\
  \bibnamefont {Wipf}},\ }\bibfield  {title} {\enquote {\bibinfo {title}
  {{Critical flavor number of the Thirring model in three dimensions}},}\
  }\href {\doibase 10.1103/PhysRevD.96.094504} {\bibfield  {journal} {\bibinfo
  {journal} {Phys. Rev. D}\ }\textbf {\bibinfo {volume} {96}},\ \bibinfo
  {pages} {094504} (\bibinfo {year} {2017})},\ \Eprint
  {http://arxiv.org/abs/1708.01160} {arXiv:1708.01160 [hep-lat]} \BibitemShut
  {NoStop}%
\bibitem [{\citenamefont {Hands}(2016)}]{Hands:2016foa}%
  \BibitemOpen
  \bibfield  {author} {\bibinfo {author} {\bibfnamefont {Simon}\ \bibnamefont
  {Hands}},\ }\bibfield  {title} {\enquote {\bibinfo {title} {{Towards Critical
  Physics in 2+1d with U(2N)-Invariant Fermions}},}\ }\href {\doibase
  10.1007/JHEP11(2016)015} {\bibfield  {journal} {\bibinfo  {journal} {JHEP}\
  }\textbf {\bibinfo {volume} {11}},\ \bibinfo {pages} {015} (\bibinfo {year}
  {2016})},\ \Eprint {http://arxiv.org/abs/1610.04394} {arXiv:1610.04394
  [hep-lat]} \BibitemShut {NoStop}%
\bibitem [{\citenamefont {Hands}(2019)}]{Hands:2018vrd}%
  \BibitemOpen
  \bibfield  {author} {\bibinfo {author} {\bibfnamefont {Simon}\ \bibnamefont
  {Hands}},\ }\bibfield  {title} {\enquote {\bibinfo {title} {{Critical flavor
  number in the 2+1D Thirring model}},}\ }\href {\doibase
  10.1103/PhysRevD.99.034504} {\bibfield  {journal} {\bibinfo  {journal} {Phys.
  Rev. D}\ }\textbf {\bibinfo {volume} {99}},\ \bibinfo {pages} {034504}
  (\bibinfo {year} {2019})},\ \Eprint {http://arxiv.org/abs/1811.04818}
  {arXiv:1811.04818 [hep-lat]} \BibitemShut {NoStop}%
\bibitem [{\citenamefont {Chandrasekharan}\ and\ \citenamefont
  {Li}(2013)}]{Chandrasekharan:2013aya}%
  \BibitemOpen
  \bibfield  {author} {\bibinfo {author} {\bibfnamefont {Shailesh}\
  \bibnamefont {Chandrasekharan}}\ and\ \bibinfo {author} {\bibfnamefont
  {Anyi}\ \bibnamefont {Li}},\ }\bibfield  {title} {\enquote {\bibinfo {title}
  {{Quantum critical behavior in three dimensional lattice Gross-Neveu
  models}},}\ }\href {\doibase 10.1103/PhysRevD.88.021701} {\bibfield
  {journal} {\bibinfo  {journal} {Phys. Rev. D}\ }\textbf {\bibinfo {volume}
  {88}},\ \bibinfo {pages} {021701} (\bibinfo {year} {2013})},\ \Eprint
  {http://arxiv.org/abs/1304.7761} {arXiv:1304.7761 [hep-lat]} \BibitemShut
  {NoStop}%
\bibitem [{\citenamefont {Sorella}\ \emph {et~al.}(2012)\citenamefont
  {Sorella}, \citenamefont {Otsuka},\ and\ \citenamefont
  {Yunoki}}]{sorella2012:absence}%
  \BibitemOpen
  \bibfield  {author} {\bibinfo {author} {\bibfnamefont {Sandro}\ \bibnamefont
  {Sorella}}, \bibinfo {author} {\bibfnamefont {Yuichi}\ \bibnamefont
  {Otsuka}}, \ and\ \bibinfo {author} {\bibfnamefont {Seiji}\ \bibnamefont
  {Yunoki}},\ }\bibfield  {title} {\enquote {\bibinfo {title} {Absence of a
  spin liquid phase in the hubbard model on the honeycomb lattice},}\ }\href
  {\doibase 10.1038/srep00992} {\bibfield  {journal} {\bibinfo  {journal}
  {Scientific Reports}\ }\textbf {\bibinfo {volume} {2}} (\bibinfo {year}
  {2012}),\ 10.1038/srep00992}\BibitemShut {NoStop}%
\bibitem [{\citenamefont {Lang}\ \emph {et~al.}(2013)\citenamefont {Lang},
  \citenamefont {Meng}, \citenamefont {Muramatsu}, \citenamefont {Wessel},\
  and\ \citenamefont {Assaad}}]{lang2013:sun}%
  \BibitemOpen
  \bibfield  {author} {\bibinfo {author} {\bibfnamefont {Thomas~C.}\
  \bibnamefont {Lang}}, \bibinfo {author} {\bibfnamefont {Zi~Yang}\
  \bibnamefont {Meng}}, \bibinfo {author} {\bibfnamefont {Alejandro}\
  \bibnamefont {Muramatsu}}, \bibinfo {author} {\bibfnamefont {Stefan}\
  \bibnamefont {Wessel}}, \ and\ \bibinfo {author} {\bibfnamefont {Fakher~F.}\
  \bibnamefont {Assaad}},\ }\bibfield  {title} {\enquote {\bibinfo {title}
  {Dimerized solids and resonating plaquette order insu(n)-dirac fermions},}\
  }\href {\doibase 10.1103/physrevlett.111.066401} {\bibfield  {journal}
  {\bibinfo  {journal} {Phys. Rev. Lett.}\ }\textbf {\bibinfo {volume} {111}}
  (\bibinfo {year} {2013}),\ 10.1103/physrevlett.111.066401}\BibitemShut
  {NoStop}%
\bibitem [{\citenamefont {Huffman}\ and\ \citenamefont
  {Chandrasekharan}(2014)}]{huffman2014:sign}%
  \BibitemOpen
  \bibfield  {author} {\bibinfo {author} {\bibfnamefont {Emilie~Fulton}\
  \bibnamefont {Huffman}}\ and\ \bibinfo {author} {\bibfnamefont {Shailesh}\
  \bibnamefont {Chandrasekharan}},\ }\bibfield  {title} {\enquote {\bibinfo
  {title} {Solution to sign problems in half-filled spin-polarized electronic
  systems},}\ }\href {\doibase 10.1103/PhysRevB.89.111101} {\bibfield
  {journal} {\bibinfo  {journal} {Phys. Rev. B}\ }\textbf {\bibinfo {volume}
  {89}},\ \bibinfo {pages} {111101} (\bibinfo {year} {2014})}\BibitemShut
  {NoStop}%
\bibitem [{\citenamefont {Li}\ \emph {et~al.}(2015)\citenamefont {Li},
  \citenamefont {Jiang},\ and\ \citenamefont {Yao}}]{li2015:mmc}%
  \BibitemOpen
  \bibfield  {author} {\bibinfo {author} {\bibfnamefont {Zi-Xiang}\
  \bibnamefont {Li}}, \bibinfo {author} {\bibfnamefont {Yi-Fan}\ \bibnamefont
  {Jiang}}, \ and\ \bibinfo {author} {\bibfnamefont {Hong}\ \bibnamefont
  {Yao}},\ }\bibfield  {title} {\enquote {\bibinfo {title} {Solving the fermion
  sign problem in quantum monte carlo simulations by majorana
  representation},}\ }\href {\doibase 10.1103/physrevb.91.241117} {\bibfield
  {journal} {\bibinfo  {journal} {Physical Review B}\ }\textbf {\bibinfo
  {volume} {91}} (\bibinfo {year} {2015}),\
  10.1103/physrevb.91.241117}\BibitemShut {NoStop}%
\bibitem [{\citenamefont {Ayyar}\ and\ \citenamefont
  {Chandrasekharan}(2016{\natexlab{a}})}]{Ayyar:2016lxq}%
  \BibitemOpen
  \bibfield  {author} {\bibinfo {author} {\bibfnamefont {Venkitesh}\
  \bibnamefont {Ayyar}}\ and\ \bibinfo {author} {\bibfnamefont {Shailesh}\
  \bibnamefont {Chandrasekharan}},\ }\bibfield  {title} {\enquote {\bibinfo
  {title} {{Fermion masses through four-fermion condensates}},}\ }\href
  {\doibase 10.1007/JHEP10(2016)058} {\bibfield  {journal} {\bibinfo  {journal}
  {JHEP}\ }\textbf {\bibinfo {volume} {10}},\ \bibinfo {pages} {058} (\bibinfo
  {year} {2016}{\natexlab{a}})},\ \Eprint {http://arxiv.org/abs/1606.06312}
  {arXiv:1606.06312 [hep-lat]} \BibitemShut {NoStop}%
\bibitem [{\citenamefont {Catterall}\ and\ \citenamefont
  {Butt}(2018)}]{PhysRevD.97.094502}%
  \BibitemOpen
  \bibfield  {author} {\bibinfo {author} {\bibfnamefont {Simon}\ \bibnamefont
  {Catterall}}\ and\ \bibinfo {author} {\bibfnamefont {Nouman}\ \bibnamefont
  {Butt}},\ }\bibfield  {title} {\enquote {\bibinfo {title} {Topology and
  strong four fermion interactions in four dimensions},}\ }\href {\doibase
  10.1103/PhysRevD.97.094502} {\bibfield  {journal} {\bibinfo  {journal} {Phys.
  Rev. D}\ }\textbf {\bibinfo {volume} {97}},\ \bibinfo {pages} {094502}
  (\bibinfo {year} {2018})}\BibitemShut {NoStop}%
\bibitem [{\citenamefont {Catterall}(2016)}]{Catterall:2015zua}%
  \BibitemOpen
  \bibfield  {author} {\bibinfo {author} {\bibfnamefont {Simon}\ \bibnamefont
  {Catterall}},\ }\bibfield  {title} {\enquote {\bibinfo {title} {{Fermion mass
  without symmetry breaking}},}\ }\href {\doibase 10.1007/JHEP01(2016)121}
  {\bibfield  {journal} {\bibinfo  {journal} {JHEP}\ }\textbf {\bibinfo
  {volume} {01}},\ \bibinfo {pages} {121} (\bibinfo {year} {2016})},\ \Eprint
  {http://arxiv.org/abs/1510.04153} {arXiv:1510.04153 [hep-lat]} \BibitemShut
  {NoStop}%
\bibitem [{\citenamefont {Butt}\ \emph {et~al.}(2018)\citenamefont {Butt},
  \citenamefont {Catterall},\ and\ \citenamefont {Schaich}}]{Butt:2018nkn}%
  \BibitemOpen
  \bibfield  {author} {\bibinfo {author} {\bibfnamefont {Nouman}\ \bibnamefont
  {Butt}}, \bibinfo {author} {\bibfnamefont {Simon}\ \bibnamefont {Catterall}},
  \ and\ \bibinfo {author} {\bibfnamefont {David}\ \bibnamefont {Schaich}},\
  }\bibfield  {title} {\enquote {\bibinfo {title} {{$SO(4)$ invariant
  Higgs-Yukawa model with reduced staggered fermions}},}\ }\href {\doibase
  10.1103/PhysRevD.98.114514} {\bibfield  {journal} {\bibinfo  {journal} {Phys.
  Rev. D}\ }\textbf {\bibinfo {volume} {98}},\ \bibinfo {pages} {114514}
  (\bibinfo {year} {2018})},\ \Eprint {http://arxiv.org/abs/1810.06117}
  {arXiv:1810.06117 [hep-lat]} \BibitemShut {NoStop}%
\bibitem [{\citenamefont {Slagle}\ \emph {et~al.}(2015)\citenamefont {Slagle},
  \citenamefont {You},\ and\ \citenamefont {Xu}}]{Slagle:2014vma}%
  \BibitemOpen
  \bibfield  {author} {\bibinfo {author} {\bibfnamefont {Kevin}\ \bibnamefont
  {Slagle}}, \bibinfo {author} {\bibfnamefont {Yi-Zhuang}\ \bibnamefont {You}},
  \ and\ \bibinfo {author} {\bibfnamefont {Cenke}\ \bibnamefont {Xu}},\
  }\bibfield  {title} {\enquote {\bibinfo {title} {{Exotic quantum phase
  transitions of strongly interacting topological insulators}},}\ }\href
  {\doibase 10.1103/PhysRevB.91.115121} {\bibfield  {journal} {\bibinfo
  {journal} {Phys. Rev. B}\ }\textbf {\bibinfo {volume} {B91}},\ \bibinfo
  {pages} {115121} (\bibinfo {year} {2015})},\ \Eprint
  {http://arxiv.org/abs/1409.7401} {arXiv:1409.7401 [cond-mat.str-el]}
  \BibitemShut {NoStop}%
\bibitem [{\citenamefont {Ayyar}\ and\ \citenamefont
  {Chandrasekharan}(2016{\natexlab{b}})}]{PhysRevD.93.081701}%
  \BibitemOpen
  \bibfield  {author} {\bibinfo {author} {\bibfnamefont {Venkitesh}\
  \bibnamefont {Ayyar}}\ and\ \bibinfo {author} {\bibfnamefont {Shailesh}\
  \bibnamefont {Chandrasekharan}},\ }\bibfield  {title} {\enquote {\bibinfo
  {title} {Origin of fermion masses without spontaneous symmetry breaking},}\
  }\href {\doibase 10.1103/PhysRevD.93.081701} {\bibfield  {journal} {\bibinfo
  {journal} {Phys. Rev. D}\ }\textbf {\bibinfo {volume} {93}},\ \bibinfo
  {pages} {081701} (\bibinfo {year} {2016}{\natexlab{b}})}\BibitemShut
  {NoStop}%
\bibitem [{\citenamefont {You}\ \emph {et~al.}(2018{\natexlab{a}})\citenamefont
  {You}, \citenamefont {He}, \citenamefont {Xu},\ and\ \citenamefont
  {Vishwanath}}]{PhysRevX.8.011026}%
  \BibitemOpen
  \bibfield  {author} {\bibinfo {author} {\bibfnamefont {Yi-Zhuang}\
  \bibnamefont {You}}, \bibinfo {author} {\bibfnamefont {Yin-Chen}\
  \bibnamefont {He}}, \bibinfo {author} {\bibfnamefont {Cenke}\ \bibnamefont
  {Xu}}, \ and\ \bibinfo {author} {\bibfnamefont {Ashvin}\ \bibnamefont
  {Vishwanath}},\ }\bibfield  {title} {\enquote {\bibinfo {title} {Symmetric
  fermion mass generation as deconfined quantum criticality},}\ }\href
  {\doibase 10.1103/PhysRevX.8.011026} {\bibfield  {journal} {\bibinfo
  {journal} {Phys. Rev. X}\ }\textbf {\bibinfo {volume} {8}},\ \bibinfo {pages}
  {011026} (\bibinfo {year} {2018}{\natexlab{a}})}\BibitemShut {NoStop}%
\bibitem [{\citenamefont {You}\ \emph {et~al.}(2018{\natexlab{b}})\citenamefont
  {You}, \citenamefont {He}, \citenamefont {Vishwanath},\ and\ \citenamefont
  {Xu}}]{PhysRevB.97.125112}%
  \BibitemOpen
  \bibfield  {author} {\bibinfo {author} {\bibfnamefont {Yi-Zhuang}\
  \bibnamefont {You}}, \bibinfo {author} {\bibfnamefont {Yin-Chen}\
  \bibnamefont {He}}, \bibinfo {author} {\bibfnamefont {Ashvin}\ \bibnamefont
  {Vishwanath}}, \ and\ \bibinfo {author} {\bibfnamefont {Cenke}\ \bibnamefont
  {Xu}},\ }\bibfield  {title} {\enquote {\bibinfo {title} {From bosonic
  topological transition to symmetric fermion mass generation},}\ }\href
  {\doibase 10.1103/PhysRevB.97.125112} {\bibfield  {journal} {\bibinfo
  {journal} {Phys. Rev. B}\ }\textbf {\bibinfo {volume} {97}},\ \bibinfo
  {pages} {125112} (\bibinfo {year} {2018}{\natexlab{b}})}\BibitemShut
  {NoStop}%
\bibitem [{\citenamefont {Kikukawa}(2019)}]{Kikukawa:2017ngf}%
  \BibitemOpen
  \bibfield  {author} {\bibinfo {author} {\bibfnamefont {Yoshio}\ \bibnamefont
  {Kikukawa}},\ }\bibfield  {title} {\enquote {\bibinfo {title} {{On the gauge
  invariant path-integral measure for the overlap Weyl fermions in
  $\underline{16}$ of SO(10)}},}\ }\href {\doibase 10.1093/ptep/ptz115}
  {\bibfield  {journal} {\bibinfo  {journal} {PTEP}\ }\textbf {\bibinfo
  {volume} {2019}},\ \bibinfo {pages} {113B03} (\bibinfo {year} {2019})},\
  \Eprint {http://arxiv.org/abs/1710.11618} {arXiv:1710.11618 [hep-lat]}
  \BibitemShut {NoStop}%
\bibitem [{\citenamefont {Wang}\ and\ \citenamefont
  {Wen}(2020)}]{Wang:2018cai}%
  \BibitemOpen
  \bibfield  {author} {\bibinfo {author} {\bibfnamefont {Juven}\ \bibnamefont
  {Wang}}\ and\ \bibinfo {author} {\bibfnamefont {Xiao-Gang}\ \bibnamefont
  {Wen}},\ }\bibfield  {title} {\enquote {\bibinfo {title} {{Nonperturbative
  definition of the standard models}},}\ }\href {\doibase
  10.1103/PhysRevResearch.2.023356} {\bibfield  {journal} {\bibinfo  {journal}
  {Phys. Rev. Res.}\ }\textbf {\bibinfo {volume} {2}},\ \bibinfo {pages}
  {023356} (\bibinfo {year} {2020})},\ \Eprint
  {http://arxiv.org/abs/1809.11171} {arXiv:1809.11171 [hep-th]} \BibitemShut
  {NoStop}%
\bibitem [{\citenamefont {Catterall}(2020)}]{Catterall:2020fep}%
  \BibitemOpen
  \bibfield  {author} {\bibinfo {author} {\bibfnamefont {Simon}\ \bibnamefont
  {Catterall}},\ }\bibfield  {title} {\enquote {\bibinfo {title} {{Chiral
  Lattice Theories From Staggered Fermions}},}\ }\href@noop {} {\bibfield
  {journal} {\bibinfo  {journal} {arXiv:2010.02290}\ } (\bibinfo {year}
  {2020})}\BibitemShut {NoStop}%
\bibitem [{\citenamefont {Liu}\ \emph {et~al.}(2019)\citenamefont {Liu},
  \citenamefont {Wang}, \citenamefont {Sato}, \citenamefont {Hohenadler},
  \citenamefont {Wang}, \citenamefont {Guo},\ and\ \citenamefont
  {Assaad}}]{Liu:2018sww}%
  \BibitemOpen
  \bibfield  {author} {\bibinfo {author} {\bibfnamefont {Yuhai}\ \bibnamefont
  {Liu}}, \bibinfo {author} {\bibfnamefont {Zhenjiu}\ \bibnamefont {Wang}},
  \bibinfo {author} {\bibfnamefont {Toshihiro}\ \bibnamefont {Sato}}, \bibinfo
  {author} {\bibfnamefont {Martin}\ \bibnamefont {Hohenadler}}, \bibinfo
  {author} {\bibfnamefont {Chong}\ \bibnamefont {Wang}}, \bibinfo {author}
  {\bibfnamefont {Wenan}\ \bibnamefont {Guo}}, \ and\ \bibinfo {author}
  {\bibfnamefont {Fakher~F.}\ \bibnamefont {Assaad}},\ }\bibfield  {title}
  {\enquote {\bibinfo {title} {{Superconductivity from the Condensation of
  Topological Defects in a Quantum Spin-Hall Insulator}},}\ }\href {\doibase
  10.1038/s41467-019-10372-0} {\bibfield  {journal} {\bibinfo  {journal} {Nat.
  Comm.}\ }\textbf {\bibinfo {volume} {10}},\ \bibinfo {pages} {2658} (\bibinfo
  {year} {2019})},\ \Eprint {http://arxiv.org/abs/1811.02583} {arXiv:1811.02583
  [cond-mat.str-el]} \BibitemShut {NoStop}%
\bibitem [{\citenamefont {Sato}\ \emph {et~al.}(2017)\citenamefont {Sato},
  \citenamefont {Hohenadler},\ and\ \citenamefont
  {Assaad}}]{PhysRevLett.119.197203}%
  \BibitemOpen
  \bibfield  {author} {\bibinfo {author} {\bibfnamefont {Toshihiro}\
  \bibnamefont {Sato}}, \bibinfo {author} {\bibfnamefont {Martin}\ \bibnamefont
  {Hohenadler}}, \ and\ \bibinfo {author} {\bibfnamefont {Fakher~F.}\
  \bibnamefont {Assaad}},\ }\bibfield  {title} {\enquote {\bibinfo {title}
  {Dirac fermions with competing orders: Non-landau transition with emergent
  symmetry},}\ }\href {\doibase 10.1103/PhysRevLett.119.197203} {\bibfield
  {journal} {\bibinfo  {journal} {Phys. Rev. Lett.}\ }\textbf {\bibinfo
  {volume} {119}},\ \bibinfo {pages} {197203} (\bibinfo {year}
  {2017})}\BibitemShut {NoStop}%
\bibitem [{\citenamefont {Senthil}(2004)}]{Senthil:2004aza}%
  \BibitemOpen
  \bibfield  {author} {\bibinfo {author} {\bibfnamefont {T.}~\bibnamefont
  {Senthil}},\ }\bibfield  {title} {\enquote {\bibinfo {title} {{Deconfined
  Quantum Critical Points}},}\ }\href {\doibase 10.1126/science.1091806}
  {\bibfield  {journal} {\bibinfo  {journal} {Science}\ }\textbf {\bibinfo
  {volume} {303}},\ \bibinfo {pages} {1490--1494} (\bibinfo {year}
  {2004})}\BibitemShut {NoStop}%
\bibitem [{\citenamefont {Senthil}\ and\ \citenamefont
  {Fisher}(2006)}]{PhysRevB.74.064405}%
  \BibitemOpen
  \bibfield  {author} {\bibinfo {author} {\bibfnamefont {T.}~\bibnamefont
  {Senthil}}\ and\ \bibinfo {author} {\bibfnamefont {Matthew P.~A.}\
  \bibnamefont {Fisher}},\ }\bibfield  {title} {\enquote {\bibinfo {title}
  {Competing orders, nonlinear sigma models, and topological terms in quantum
  magnets},}\ }\href {\doibase 10.1103/PhysRevB.74.064405} {\bibfield
  {journal} {\bibinfo  {journal} {Phys. Rev. B}\ }\textbf {\bibinfo {volume}
  {74}},\ \bibinfo {pages} {064405} (\bibinfo {year} {2006})}\BibitemShut
  {NoStop}%
\bibitem [{\citenamefont {Li}\ \emph {et~al.}(2017)\citenamefont {Li},
  \citenamefont {Jiang}, \citenamefont {Jian},\ and\ \citenamefont
  {Yao}}]{Li2017}%
  \BibitemOpen
  \bibfield  {author} {\bibinfo {author} {\bibfnamefont {Zi-Xiang}\
  \bibnamefont {Li}}, \bibinfo {author} {\bibfnamefont {Yi-Fan}\ \bibnamefont
  {Jiang}}, \bibinfo {author} {\bibfnamefont {Shao-Kai}\ \bibnamefont {Jian}},
  \ and\ \bibinfo {author} {\bibfnamefont {Hong}\ \bibnamefont {Yao}},\
  }\bibfield  {title} {\enquote {\bibinfo {title} {Fermion-induced quantum
  critical points},}\ }\href {\doibase 10.1038/s41467-017-00167-6} {\bibfield
  {journal} {\bibinfo  {journal} {Nat. Comms.}\ }\textbf {\bibinfo {volume}
  {8}},\ \bibinfo {pages} {314} (\bibinfo {year} {2017})}\BibitemShut {NoStop}%
\bibitem [{\citenamefont {Li}\ \emph {et~al.}(2019{\natexlab{a}})\citenamefont
  {Li}, \citenamefont {Jian},\ and\ \citenamefont {Yao}}]{Li:2019acc}%
  \BibitemOpen
  \bibfield  {author} {\bibinfo {author} {\bibfnamefont {Zi-Xiang}\
  \bibnamefont {Li}}, \bibinfo {author} {\bibfnamefont {Shao-Kai}\ \bibnamefont
  {Jian}}, \ and\ \bibinfo {author} {\bibfnamefont {Hong}\ \bibnamefont
  {Yao}},\ }\bibfield  {title} {\enquote {\bibinfo {title} {{Deconfined quantum
  criticality and emergent SO(5) symmetry in fermionic systems}},}\ }\href@noop
  {} {\bibfield  {journal} {\bibinfo  {journal} {arXiv:1904.10975}\ } (\bibinfo
  {year} {2019}{\natexlab{a}})}\BibitemShut {NoStop}%
\bibitem [{\citenamefont {Torres}\ \emph {et~al.}(2020)\citenamefont {Torres},
  \citenamefont {Weber}, \citenamefont {Janssen}, \citenamefont {Wessel},\ and\
  \citenamefont {Scherer}}]{Torres:2019vcw}%
  \BibitemOpen
  \bibfield  {author} {\bibinfo {author} {\bibfnamefont {Emilio}\ \bibnamefont
  {Torres}}, \bibinfo {author} {\bibfnamefont {Lukas}\ \bibnamefont {Weber}},
  \bibinfo {author} {\bibfnamefont {Lukas}\ \bibnamefont {Janssen}}, \bibinfo
  {author} {\bibfnamefont {Stefan}\ \bibnamefont {Wessel}}, \ and\ \bibinfo
  {author} {\bibfnamefont {Michael~M.}\ \bibnamefont {Scherer}},\ }\bibfield
  {title} {\enquote {\bibinfo {title} {{Emergent symmetries and coexisting
  orders in Dirac fermion systems}},}\ }\href {\doibase
  10.1103/PhysRevResearch.2.022005} {\bibfield  {journal} {\bibinfo  {journal}
  {Phys. Rev. Research.}\ }\textbf {\bibinfo {volume} {2}},\ \bibinfo {pages}
  {022005} (\bibinfo {year} {2020})},\ \Eprint
  {http://arxiv.org/abs/1911.01244} {arXiv:1911.01244 [cond-mat.str-el]}
  \BibitemShut {NoStop}%
\bibitem [{\citenamefont {Chandrasekharan}(2010)}]{PhysRevD.82.025007}%
  \BibitemOpen
  \bibfield  {author} {\bibinfo {author} {\bibfnamefont {Shailesh}\
  \bibnamefont {Chandrasekharan}},\ }\bibfield  {title} {\enquote {\bibinfo
  {title} {Fermion bag approach to lattice field theories},}\ }\href {\doibase
  10.1103/PhysRevD.82.025007} {\bibfield  {journal} {\bibinfo  {journal} {Phys.
  Rev. D}\ }\textbf {\bibinfo {volume} {82}},\ \bibinfo {pages} {025007}
  (\bibinfo {year} {2010})}\BibitemShut {NoStop}%
\bibitem [{\citenamefont {Chandrasekharan}(2013)}]{Chandrasekharan:2013rpa}%
  \BibitemOpen
  \bibfield  {author} {\bibinfo {author} {\bibfnamefont {Shailesh}\
  \bibnamefont {Chandrasekharan}},\ }\bibfield  {title} {\enquote {\bibinfo
  {title} {{Fermion Bag Approach to Fermion Sign Problems}},}\ }\href {\doibase
  10.1140/epja/i2013-13090-y} {\bibfield  {journal} {\bibinfo  {journal} {Eur.
  Phys. J.}\ }\textbf {\bibinfo {volume} {A49}},\ \bibinfo {pages} {90}
  (\bibinfo {year} {2013})},\ \Eprint {http://arxiv.org/abs/1304.4900}
  {arXiv:1304.4900 [hep-lat]} \BibitemShut {NoStop}%
\bibitem [{\citenamefont {Huffman}\ and\ \citenamefont
  {Chandrasekharan}(2017)}]{Huffman:2017swn}%
  \BibitemOpen
  \bibfield  {author} {\bibinfo {author} {\bibfnamefont {Emilie}\ \bibnamefont
  {Huffman}}\ and\ \bibinfo {author} {\bibfnamefont {Shailesh}\ \bibnamefont
  {Chandrasekharan}},\ }\bibfield  {title} {\enquote {\bibinfo {title}
  {{Fermion bag approach to Hamiltonian lattice field theories in continuous
  time}},}\ }\href {\doibase 10.1103/PhysRevD.96.114502} {\bibfield  {journal}
  {\bibinfo  {journal} {Phys. Rev.}\ }\textbf {\bibinfo {volume} {D96}},\
  \bibinfo {pages} {114502} (\bibinfo {year} {2017})},\ \Eprint
  {http://arxiv.org/abs/1709.03578} {arXiv:1709.03578 [hep-lat]} \BibitemShut
  {NoStop}%
\bibitem [{\citenamefont {Huffman}\ and\ \citenamefont
  {Chandrasekharan}(2020)}]{PhysRevD.101.074501}%
  \BibitemOpen
  \bibfield  {author} {\bibinfo {author} {\bibfnamefont {Emilie}\ \bibnamefont
  {Huffman}}\ and\ \bibinfo {author} {\bibfnamefont {Shailesh}\ \bibnamefont
  {Chandrasekharan}},\ }\bibfield  {title} {\enquote {\bibinfo {title}
  {Fermion-bag inspired hamiltonian lattice field theory for fermionic quantum
  criticality},}\ }\href {\doibase 10.1103/PhysRevD.101.074501} {\bibfield
  {journal} {\bibinfo  {journal} {Phys. Rev. D}\ }\textbf {\bibinfo {volume}
  {101}},\ \bibinfo {pages} {074501} (\bibinfo {year} {2020})}\BibitemShut
  {NoStop}%
\bibitem [{\citenamefont {Chandrasekharan}\ and\ \citenamefont
  {Wiese}(1999)}]{Chandrasekharan:1999cm}%
  \BibitemOpen
  \bibfield  {author} {\bibinfo {author} {\bibfnamefont {Shailesh}\
  \bibnamefont {Chandrasekharan}}\ and\ \bibinfo {author} {\bibfnamefont
  {Uwe-Jens}\ \bibnamefont {Wiese}},\ }\bibfield  {title} {\enquote {\bibinfo
  {title} {{Meron cluster solution of a fermion sign problem}},}\ }\href
  {\doibase 10.1103/PhysRevLett.83.3116} {\bibfield  {journal} {\bibinfo
  {journal} {Phys. Rev. Lett.}\ }\textbf {\bibinfo {volume} {83}},\ \bibinfo
  {pages} {3116--3119} (\bibinfo {year} {1999})},\ \Eprint
  {http://arxiv.org/abs/cond-mat/9902128} {arXiv:cond-mat/9902128 [cond-mat]}
  \BibitemShut {NoStop}%
\bibitem [{\citenamefont {Chandrasekharan}\ \emph {et~al.}(2003)\citenamefont
  {Chandrasekharan}, \citenamefont {Cox}, \citenamefont {Osborn},\ and\
  \citenamefont {Wiese}}]{Chandrasekharan:2002vk}%
  \BibitemOpen
  \bibfield  {author} {\bibinfo {author} {\bibfnamefont {S.}~\bibnamefont
  {Chandrasekharan}}, \bibinfo {author} {\bibfnamefont {J.}~\bibnamefont
  {Cox}}, \bibinfo {author} {\bibfnamefont {J.~C.}\ \bibnamefont {Osborn}}, \
  and\ \bibinfo {author} {\bibfnamefont {U.~J.}\ \bibnamefont {Wiese}},\
  }\bibfield  {title} {\enquote {\bibinfo {title} {{Meron cluster approach to
  systems of strongly correlated electrons}},}\ }\href {\doibase
  10.1016/j.nuclphysb.2003.08.041} {\bibfield  {journal} {\bibinfo  {journal}
  {Nucl. Phys. B}\ }\textbf {\bibinfo {volume} {673}},\ \bibinfo {pages}
  {405--436} (\bibinfo {year} {2003})},\ \Eprint
  {http://arxiv.org/abs/cond-mat/0201360} {arXiv:cond-mat/0201360
  [cond-mat.str-el]} \BibitemShut {NoStop}%
\bibitem [{\citenamefont {Chandrasekharan}\ and\ \citenamefont
  {Osborn}(2002)}]{PhysRevB.66.045113}%
  \BibitemOpen
  \bibfield  {author} {\bibinfo {author} {\bibfnamefont {Shailesh}\
  \bibnamefont {Chandrasekharan}}\ and\ \bibinfo {author} {\bibfnamefont
  {James~C.}\ \bibnamefont {Osborn}},\ }\bibfield  {title} {\enquote {\bibinfo
  {title} {Kosterlitz-thouless universality in a fermionic system},}\ }\href
  {\doibase 10.1103/PhysRevB.66.045113} {\bibfield  {journal} {\bibinfo
  {journal} {Phys. Rev. B}\ }\textbf {\bibinfo {volume} {66}},\ \bibinfo
  {pages} {045113} (\bibinfo {year} {2002})}\BibitemShut {NoStop}%
\bibitem [{\citenamefont {Zhang}(1990)}]{zhang1990:su2}%
  \BibitemOpen
  \bibfield  {author} {\bibinfo {author} {\bibfnamefont {Shoucheng}\
  \bibnamefont {Zhang}},\ }\bibfield  {title} {\enquote {\bibinfo {title}
  {Pseudospin symmetry and new collective modes of the hubbard model},}\ }\href
  {\doibase 10.1103/PhysRevLett.65.120} {\bibfield  {journal} {\bibinfo
  {journal} {Phys. Rev. Lett.}\ }\textbf {\bibinfo {volume} {65}},\ \bibinfo
  {pages} {120--122} (\bibinfo {year} {1990})}\BibitemShut {NoStop}%
\bibitem [{\citenamefont {Beard}\ \emph {et~al.}(2005)\citenamefont {Beard},
  \citenamefont {Pepe}, \citenamefont {Riederer},\ and\ \citenamefont
  {Wiese}}]{Beard:2004jr}%
  \BibitemOpen
  \bibfield  {author} {\bibinfo {author} {\bibfnamefont {B.~B.}\ \bibnamefont
  {Beard}}, \bibinfo {author} {\bibfnamefont {M.}~\bibnamefont {Pepe}},
  \bibinfo {author} {\bibfnamefont {S.}~\bibnamefont {Riederer}}, \ and\
  \bibinfo {author} {\bibfnamefont {U.~J.}\ \bibnamefont {Wiese}},\ }\bibfield
  {title} {\enquote {\bibinfo {title} {{Study of CP(N-1) theta-vacua by
  cluster-simulation of SU(N) quantum spin ladders}},}\ }\href {\doibase
  10.1103/PhysRevLett.94.010603} {\bibfield  {journal} {\bibinfo  {journal}
  {Phys. Rev. Lett.}\ }\textbf {\bibinfo {volume} {94}},\ \bibinfo {pages}
  {010603} (\bibinfo {year} {2005})},\ \Eprint
  {http://arxiv.org/abs/hep-lat/0406040} {arXiv:hep-lat/0406040 [hep-lat]}
  \BibitemShut {NoStop}%
\bibitem [{\citenamefont {Affleck}\ \emph {et~al.}(1989)\citenamefont
  {Affleck}, \citenamefont {Gepner}, \citenamefont {Schulz},\ and\
  \citenamefont {Ziman}}]{Affleck:1988px}%
  \BibitemOpen
  \bibfield  {author} {\bibinfo {author} {\bibfnamefont {I.}~\bibnamefont
  {Affleck}}, \bibinfo {author} {\bibfnamefont {D.}~\bibnamefont {Gepner}},
  \bibinfo {author} {\bibfnamefont {H.~J.}\ \bibnamefont {Schulz}}, \ and\
  \bibinfo {author} {\bibfnamefont {T.}~\bibnamefont {Ziman}},\ }\bibfield
  {title} {\enquote {\bibinfo {title} {{Critical Behavior of Spin S Heisenberg
  Antiferromagnetic Chains: Analytic and Numerical Results}},}\ }\href
  {\doibase 10.1088/0305-4470/22/5/015} {\bibfield  {journal} {\bibinfo
  {journal} {J. Phys. A}\ }\textbf {\bibinfo {volume} {22}},\ \bibinfo {pages}
  {511} (\bibinfo {year} {1989})}\BibitemShut {NoStop}%
\bibitem [{\citenamefont {Okamoto}\ and\ \citenamefont
  {Nomura}(1992)}]{OKAMOTO1992433}%
  \BibitemOpen
  \bibfield  {author} {\bibinfo {author} {\bibfnamefont {Kiyomi}\ \bibnamefont
  {Okamoto}}\ and\ \bibinfo {author} {\bibfnamefont {Kiyohide}\ \bibnamefont
  {Nomura}},\ }\bibfield  {title} {\enquote {\bibinfo {title} {Fluid-dimer
  critical point in s = 12 antiferromagnetic heisenberg chain with next nearest
  neighbor interactions},}\ }\href {\doibase
  https://doi.org/10.1016/0375-9601(92)90823-5} {\bibfield  {journal} {\bibinfo
   {journal} {Phys. Letts. A}\ }\textbf {\bibinfo {volume} {169}},\ \bibinfo
  {pages} {433 -- 437} (\bibinfo {year} {1992})}\BibitemShut {NoStop}%
\bibitem [{\citenamefont {Eggert}(1996)}]{Eggert:1996er}%
  \BibitemOpen
  \bibfield  {author} {\bibinfo {author} {\bibfnamefont {Sebastian}\
  \bibnamefont {Eggert}},\ }\bibfield  {title} {\enquote {\bibinfo {title}
  {{Numerical evidence for multiplicative logarithmic corrections from marginal
  operators}},}\ }\href {\doibase 10.1103/PhysRevB.54.R9612} {\bibfield
  {journal} {\bibinfo  {journal} {Phys. Rev. B}\ }\textbf {\bibinfo {volume}
  {54}},\ \bibinfo {pages} {R9612} (\bibinfo {year} {1996})},\ \Eprint
  {http://arxiv.org/abs/cond-mat/9602026} {arXiv:cond-mat/9602026 [cond-mat]}
  \BibitemShut {NoStop}%
\bibitem [{\citenamefont {Evertz}\ \emph {et~al.}(1993)\citenamefont {Evertz},
  \citenamefont {Lana},\ and\ \citenamefont {Marcu}}]{Evertz:1992rb}%
  \BibitemOpen
  \bibfield  {author} {\bibinfo {author} {\bibfnamefont {Hans~Gerd}\
  \bibnamefont {Evertz}}, \bibinfo {author} {\bibfnamefont {Gideon}\
  \bibnamefont {Lana}}, \ and\ \bibinfo {author} {\bibfnamefont {Mihai}\
  \bibnamefont {Marcu}},\ }\bibfield  {title} {\enquote {\bibinfo {title}
  {{Cluster algorithm for vertex models}},}\ }\href {\doibase
  10.1103/PhysRevLett.70.875} {\bibfield  {journal} {\bibinfo  {journal} {Phys.
  Rev. Lett.}\ }\textbf {\bibinfo {volume} {70}},\ \bibinfo {pages} {875--879}
  (\bibinfo {year} {1993})},\ \Eprint {http://arxiv.org/abs/cond-mat/9211006}
  {arXiv:cond-mat/9211006} \BibitemShut {NoStop}%
\bibitem [{\citenamefont {Beard}\ and\ \citenamefont
  {Wiese}(1996)}]{Beard:1996wj}%
  \BibitemOpen
  \bibfield  {author} {\bibinfo {author} {\bibfnamefont {B.B.}\ \bibnamefont
  {Beard}}\ and\ \bibinfo {author} {\bibfnamefont {U.-J.}\ \bibnamefont
  {Wiese}},\ }\bibfield  {title} {\enquote {\bibinfo {title} {{Simulations of
  discrete quantum systems in continuous Euclidean time}},}\ }\href {\doibase
  10.1103/PhysRevLett.77.5130} {\bibfield  {journal} {\bibinfo  {journal}
  {Phys. Rev. Lett.}\ }\textbf {\bibinfo {volume} {77}},\ \bibinfo {pages}
  {5130--5133} (\bibinfo {year} {1996})},\ \Eprint
  {http://arxiv.org/abs/cond-mat/9602164} {arXiv:cond-mat/9602164} \BibitemShut
  {NoStop}%
\bibitem [{\citenamefont {Affleck}(1990)}]{Affleck1990}%
  \BibitemOpen
  \bibfield  {author} {\bibinfo {author} {\bibfnamefont {Ian}\ \bibnamefont
  {Affleck}},\ }\enquote {\bibinfo {title} {Field theory methods and strongly
  correlated electrons},}\ in\ \href {\doibase 10.1007/978-1-4615-3802-8_1}
  {\emph {\bibinfo {booktitle} {Physics, Geometry and Topology}}},\ \bibinfo
  {editor} {edited by\ \bibinfo {editor} {\bibfnamefont {H.~C.}\ \bibnamefont
  {Lee}}}\ (\bibinfo  {publisher} {Springer US},\ \bibinfo {address} {Boston,
  MA},\ \bibinfo {year} {1990})\ pp.\ \bibinfo {pages} {1--13}\BibitemShut
  {NoStop}%
\bibitem [{\citenamefont {Liu}(2019)}]{Liu:2019dvk}%
  \BibitemOpen
  \bibfield  {author} {\bibinfo {author} {\bibfnamefont {Hanqing}\ \bibnamefont
  {Liu}},\ }\bibfield  {title} {\enquote {\bibinfo {title} {{Quantum Critical
  Phenomena in an $O(4)$ Fermion Chain}},}\ }in\ \href@noop {} {\emph {\bibinfo
  {booktitle} {{37th International Symposium on Lattice Field Theory}}}}\
  (\bibinfo {year} {2019})\ \Eprint {http://arxiv.org/abs/1912.11237}
  {arXiv:1912.11237 [hep-lat]} \BibitemShut {NoStop}%
\bibitem [{\citenamefont {Affleck}(1988)}]{Affleck:1988zj}%
  \BibitemOpen
  \bibfield  {author} {\bibinfo {author} {\bibfnamefont {Ian}\ \bibnamefont
  {Affleck}},\ }\bibfield  {title} {\enquote {\bibinfo {title} {{FIELD THEORY
  METHODS AND QUANTUM CRITICAL PHENOMENA}},}\ }in\ \href@noop {} {\emph
  {\bibinfo {booktitle} {{Les Houches Summer School in Theoretical Physics:
  Fields, Strings, Critical Phenomena Les Houches, France, June 28-August 5,
  1988}}}}\ (\bibinfo {year} {1988})\ pp.\ \bibinfo {pages}
  {0563--640}\BibitemShut {NoStop}%
\bibitem [{\citenamefont {Witten}(1984)}]{Witten:1983ar}%
  \BibitemOpen
  \bibfield  {author} {\bibinfo {author} {\bibfnamefont {Edward}\ \bibnamefont
  {Witten}},\ }\bibfield  {title} {\enquote {\bibinfo {title} {{Nonabelian
  Bosonization in Two-Dimensions}},}\ }\href {\doibase 10.1007/BF01215276}
  {\bibfield  {journal} {\bibinfo  {journal} {Commun. Math. Phys.}\ }\textbf
  {\bibinfo {volume} {92}},\ \bibinfo {pages} {455--472} (\bibinfo {year}
  {1984})},\ \bibinfo {note} {[201(1983)]}\BibitemShut {NoStop}%
\bibitem [{\citenamefont {Lieb}\ \emph {et~al.}(1961)\citenamefont {Lieb},
  \citenamefont {Schultz},\ and\ \citenamefont {Mattis}}]{Lieb:1961fr}%
  \BibitemOpen
  \bibfield  {author} {\bibinfo {author} {\bibfnamefont {Elliott~H.}\
  \bibnamefont {Lieb}}, \bibinfo {author} {\bibfnamefont {Theodore}\
  \bibnamefont {Schultz}}, \ and\ \bibinfo {author} {\bibfnamefont {Daniel}\
  \bibnamefont {Mattis}},\ }\bibfield  {title} {\enquote {\bibinfo {title}
  {{Two soluble models of an antiferromagnetic chain}},}\ }\href {\doibase
  10.1016/0003-4916(61)90115-4} {\bibfield  {journal} {\bibinfo  {journal}
  {Annals of Phys.}\ }\textbf {\bibinfo {volume} {16}},\ \bibinfo {pages}
  {407--466} (\bibinfo {year} {1961})}\BibitemShut {NoStop}%
\bibitem [{\citenamefont {Majumdar}\ and\ \citenamefont
  {Ghosh}(1969)}]{Majumdar_1969}%
  \BibitemOpen
  \bibfield  {author} {\bibinfo {author} {\bibfnamefont {Chanchal~K.}\
  \bibnamefont {Majumdar}}\ and\ \bibinfo {author} {\bibfnamefont {Dipan~K.}\
  \bibnamefont {Ghosh}},\ }\bibfield  {title} {\enquote {\bibinfo {title} {On
  next‐nearest‐neighbor interaction in linear chain. i},}\ }\href {\doibase
  10.1063/1.1664978} {\bibfield  {journal} {\bibinfo  {journal} {Journal of
  Mathematical Physics}\ }\textbf {\bibinfo {volume} {10}},\ \bibinfo {pages}
  {1388--1398} (\bibinfo {year} {1969})},\ \Eprint
  {http://arxiv.org/abs/https://doi.org/10.1063/1.1664978}
  {https://doi.org/10.1063/1.1664978} \BibitemShut {NoStop}%
\bibitem [{\citenamefont {Majumdar}(1970)}]{Majumdar_1970}%
  \BibitemOpen
  \bibfield  {author} {\bibinfo {author} {\bibfnamefont {C~K}\ \bibnamefont
  {Majumdar}},\ }\bibfield  {title} {\enquote {\bibinfo {title}
  {Antiferromagnetic model with known ground state},}\ }\href {\doibase
  10.1088/0022-3719/3/4/019} {\bibfield  {journal} {\bibinfo  {journal} {J. of
  Phys. C: Solid State Physics}\ }\textbf {\bibinfo {volume} {3}},\ \bibinfo
  {pages} {911--915} (\bibinfo {year} {1970})}\BibitemShut {NoStop}%
\bibitem [{\citenamefont {Haldane}(1982)}]{PhysRevB.25.4925}%
  \BibitemOpen
  \bibfield  {author} {\bibinfo {author} {\bibfnamefont {F.~D.~M.}\
  \bibnamefont {Haldane}},\ }\bibfield  {title} {\enquote {\bibinfo {title}
  {Spontaneous dimerization in the $s=\frac{1}{2}$ heisenberg antiferromagnetic
  chain with competing interactions},}\ }\href {\doibase
  10.1103/PhysRevB.25.4925} {\bibfield  {journal} {\bibinfo  {journal} {Phys.
  Rev. B}\ }\textbf {\bibinfo {volume} {25}},\ \bibinfo {pages} {4925--4928}
  (\bibinfo {year} {1982})}\BibitemShut {NoStop}%
\bibitem [{\citenamefont {Sanyal}\ \emph {et~al.}(2011)\citenamefont {Sanyal},
  \citenamefont {Banerjee},\ and\ \citenamefont {Damle}}]{sanyal2011:1djq}%
  \BibitemOpen
  \bibfield  {author} {\bibinfo {author} {\bibfnamefont {Sambuddha}\
  \bibnamefont {Sanyal}}, \bibinfo {author} {\bibfnamefont {Argha}\
  \bibnamefont {Banerjee}}, \ and\ \bibinfo {author} {\bibfnamefont {Kedar}\
  \bibnamefont {Damle}},\ }\bibfield  {title} {\enquote {\bibinfo {title}
  {Vacancy-induced spin texture in a one-dimensional $s=\frac{1}{2}$ heisenberg
  antiferromagnet},}\ }\href {\doibase 10.1103/PhysRevB.84.235129} {\bibfield
  {journal} {\bibinfo  {journal} {Phys. Rev. B}\ }\textbf {\bibinfo {volume}
  {84}},\ \bibinfo {pages} {235129} (\bibinfo {year} {2011})}\BibitemShut
  {NoStop}%
\bibitem [{\citenamefont {Patil}\ \emph {et~al.}(2018)\citenamefont {Patil},
  \citenamefont {Katz},\ and\ \citenamefont {Sandvik}}]{Patil:2018wpt}%
  \BibitemOpen
  \bibfield  {author} {\bibinfo {author} {\bibfnamefont {Pranay}\ \bibnamefont
  {Patil}}, \bibinfo {author} {\bibfnamefont {Emanuel}\ \bibnamefont {Katz}}, \
  and\ \bibinfo {author} {\bibfnamefont {Anders~W.}\ \bibnamefont {Sandvik}},\
  }\bibfield  {title} {\enquote {\bibinfo {title} {{Numerical investigations of
  SO(4) emergent extended symmetry in spin- 12 Heisenberg antiferromagnetic
  chains}},}\ }\href {\doibase 10.1103/PhysRevB.98.014414} {\bibfield
  {journal} {\bibinfo  {journal} {Phys. Rev. B}\ }\textbf {\bibinfo {volume}
  {98}},\ \bibinfo {pages} {014414} (\bibinfo {year} {2018})},\ \Eprint
  {http://arxiv.org/abs/1803.02041} {arXiv:1803.02041 [cond-mat.str-el]}
  \BibitemShut {NoStop}%
\bibitem [{\citenamefont {Sandvik}\ \emph {et~al.}(2004)\citenamefont
  {Sandvik}, \citenamefont {Balents},\ and\ \citenamefont
  {Campbell}}]{sandvik2004:exthubb}%
  \BibitemOpen
  \bibfield  {author} {\bibinfo {author} {\bibfnamefont {Anders~W.}\
  \bibnamefont {Sandvik}}, \bibinfo {author} {\bibfnamefont {Leon}\
  \bibnamefont {Balents}}, \ and\ \bibinfo {author} {\bibfnamefont {David~K.}\
  \bibnamefont {Campbell}},\ }\bibfield  {title} {\enquote {\bibinfo {title}
  {Ground state phases of the half-filled one-dimensional extended hubbard
  model},}\ }\href {\doibase 10.1103/PhysRevLett.92.236401} {\bibfield
  {journal} {\bibinfo  {journal} {Phys. Rev. Lett.}\ }\textbf {\bibinfo
  {volume} {92}},\ \bibinfo {pages} {236401} (\bibinfo {year}
  {2004})}\BibitemShut {NoStop}%
\bibitem [{\citenamefont {Tsvelik}(2003)}]{tsvelik_2003}%
  \BibitemOpen
  \bibfield  {author} {\bibinfo {author} {\bibfnamefont {Alexei~M.}\
  \bibnamefont {Tsvelik}},\ }\href {\doibase 10.1017/CBO9780511615832} {\emph
  {\bibinfo {title} {Quantum Field Theory in Condensed Matter Physics}}},\
  \bibinfo {edition} {2nd}\ ed.\ (\bibinfo  {publisher} {Cambridge University
  Press},\ \bibinfo {year} {2003})\BibitemShut {NoStop}%
\bibitem [{\citenamefont {Knizhnik}\ and\ \citenamefont
  {Zamolodchikov}(1984)}]{Knizhnik:1984nr}%
  \BibitemOpen
  \bibfield  {author} {\bibinfo {author} {\bibfnamefont {V.~G.}\ \bibnamefont
  {Knizhnik}}\ and\ \bibinfo {author} {\bibfnamefont {A.~B.}\ \bibnamefont
  {Zamolodchikov}},\ }\bibfield  {title} {\enquote {\bibinfo {title} {{Current
  Algebra and Wess-Zumino Model in Two-Dimensions}},}\ }\href {\doibase
  10.1016/0550-3213(84)90374-2} {\bibfield  {journal} {\bibinfo  {journal}
  {Nucl. Phys. B}\ }\textbf {\bibinfo {volume} {247}},\ \bibinfo {pages}
  {83--103} (\bibinfo {year} {1984})},\ \bibinfo {note}
  {[,690(1984)]}\BibitemShut {NoStop}%
\bibitem [{\citenamefont {Zamolodchikov}\ and\ \citenamefont
  {Fateev}(1986)}]{Zamolodchikov:1986bd}%
  \BibitemOpen
  \bibfield  {author} {\bibinfo {author} {\bibfnamefont {A.~B.}\ \bibnamefont
  {Zamolodchikov}}\ and\ \bibinfo {author} {\bibfnamefont {V.~A.}\ \bibnamefont
  {Fateev}},\ }\bibfield  {title} {\enquote {\bibinfo {title} {{Operator
  Algebra and Correlation Functions in the Two-Dimensional Wess-Zumino SU(2) x
  SU(2) Chiral Model}},}\ }\href@noop {} {\bibfield  {journal} {\bibinfo
  {journal} {Sov. J. Nucl. Phys.}\ }\textbf {\bibinfo {volume} {43}},\ \bibinfo
  {pages} {657--664} (\bibinfo {year} {1986})},\ \bibinfo {note} {[Yad.
  Fiz.43,1031(1986)]}\BibitemShut {NoStop}%
\bibitem [{\citenamefont {Gepner}\ and\ \citenamefont
  {Witten}(1986)}]{Gepner:1986wi}%
  \BibitemOpen
  \bibfield  {author} {\bibinfo {author} {\bibfnamefont {Doron}\ \bibnamefont
  {Gepner}}\ and\ \bibinfo {author} {\bibfnamefont {Edward}\ \bibnamefont
  {Witten}},\ }\bibfield  {title} {\enquote {\bibinfo {title} {{String Theory
  on Group Manifolds}},}\ }\href {\doibase 10.1016/0550-3213(86)90051-9}
  {\bibfield  {journal} {\bibinfo  {journal} {Nucl. Phys. B}\ }\textbf
  {\bibinfo {volume} {278}},\ \bibinfo {pages} {493--549} (\bibinfo {year}
  {1986})}\BibitemShut {NoStop}%
\bibitem [{\citenamefont {Cardy}(1984)}]{Cardy:1984rp}%
  \BibitemOpen
  \bibfield  {author} {\bibinfo {author} {\bibfnamefont {John~L.}\ \bibnamefont
  {Cardy}},\ }\bibfield  {title} {\enquote {\bibinfo {title} {{Conformal
  invariance and universality in finite-size scaling}},}\ }\href@noop {}
  {\bibfield  {journal} {\bibinfo  {journal} {J. Phys. A}\ }\textbf {\bibinfo
  {volume} {17}},\ \bibinfo {pages} {L385--L387} (\bibinfo {year}
  {1984})}\BibitemShut {NoStop}%
\bibitem [{\citenamefont {Affleck}(1998)}]{Affleck_1998}%
  \BibitemOpen
  \bibfield  {author} {\bibinfo {author} {\bibfnamefont {Ian}\ \bibnamefont
  {Affleck}},\ }\bibfield  {title} {\enquote {\bibinfo {title} {Exact
  correlation amplitude for the heisenberg antiferromagnetic chain},}\ }\href
  {\doibase 10.1088/0305-4470/31/20/002} {\bibfield  {journal} {\bibinfo
  {journal} {Journal of Physics A: Mathematical and General}\ }\textbf
  {\bibinfo {volume} {31}},\ \bibinfo {pages} {4573--4581} (\bibinfo {year}
  {1998})}\BibitemShut {NoStop}%
\bibitem [{\citenamefont {Hikihara}\ \emph {et~al.}(2017)\citenamefont
  {Hikihara}, \citenamefont {Furusaki},\ and\ \citenamefont
  {Lukyanov}}]{PhysRevB.96.134429}%
  \BibitemOpen
  \bibfield  {author} {\bibinfo {author} {\bibfnamefont {Toshiya}\ \bibnamefont
  {Hikihara}}, \bibinfo {author} {\bibfnamefont {Akira}\ \bibnamefont
  {Furusaki}}, \ and\ \bibinfo {author} {\bibfnamefont {Sergei}\ \bibnamefont
  {Lukyanov}},\ }\bibfield  {title} {\enquote {\bibinfo {title} {Dimer
  correlation amplitudes and dimer excitation gap in spin-$\frac{1}{2}$ xxz and
  heisenberg chains},}\ }\href {\doibase 10.1103/PhysRevB.96.134429} {\bibfield
   {journal} {\bibinfo  {journal} {Phys. Rev. B}\ }\textbf {\bibinfo {volume}
  {96}},\ \bibinfo {pages} {134429} (\bibinfo {year} {2017})}\BibitemShut
  {NoStop}%
\bibitem [{\citenamefont {Vekua}\ and\ \citenamefont
  {Sun}(2016)}]{PhysRevB.94.014417}%
  \BibitemOpen
  \bibfield  {author} {\bibinfo {author} {\bibfnamefont {T.}~\bibnamefont
  {Vekua}}\ and\ \bibinfo {author} {\bibfnamefont {G.}~\bibnamefont {Sun}},\
  }\bibfield  {title} {\enquote {\bibinfo {title} {Exact asymptotic correlation
  functions of bilinear spin operators of the heisenberg antiferromagnetic
  spin-$\frac{1}{2}$ chain},}\ }\href {\doibase 10.1103/PhysRevB.94.014417}
  {\bibfield  {journal} {\bibinfo  {journal} {Phys. Rev. B}\ }\textbf {\bibinfo
  {volume} {94}},\ \bibinfo {pages} {014417} (\bibinfo {year}
  {2016})}\BibitemShut {NoStop}%
\bibitem [{\citenamefont {Li}\ \emph {et~al.}(2019{\natexlab{b}})\citenamefont
  {Li}, \citenamefont {Lu},\ and\ \citenamefont {Wang}}]{Wang:2019a}%
  \BibitemOpen
  \bibfield  {author} {\bibinfo {author} {\bibfnamefont {Yingzhou}\
  \bibnamefont {Li}}, \bibinfo {author} {\bibfnamefont {Jianfeng}\ \bibnamefont
  {Lu}}, \ and\ \bibinfo {author} {\bibfnamefont {Zhe}\ \bibnamefont {Wang}},\
  }\bibfield  {title} {\enquote {\bibinfo {title} {Coordinatewise descent
  methods for leading eigenvalue problem},}\ }\href {\doibase
  10.1137/18M1202505} {\bibfield  {journal} {\bibinfo  {journal} {SIAM J. on
  Sc. Comp.}\ }\textbf {\bibinfo {volume} {41}},\ \bibinfo {pages}
  {A2681--A2716} (\bibinfo {year} {2019}{\natexlab{b}})},\ \Eprint
  {http://arxiv.org/abs/https://doi.org/10.1137/18M1202505}
  {https://doi.org/10.1137/18M1202505} \BibitemShut {NoStop}%
\bibitem [{\citenamefont {Wang}\ \emph {et~al.}(2019)\citenamefont {Wang},
  \citenamefont {Li},\ and\ \citenamefont {Lu}}]{Wang:2019b}%
  \BibitemOpen
  \bibfield  {author} {\bibinfo {author} {\bibfnamefont {Zhe}\ \bibnamefont
  {Wang}}, \bibinfo {author} {\bibfnamefont {Yingzhou}\ \bibnamefont {Li}}, \
  and\ \bibinfo {author} {\bibfnamefont {Jianfeng}\ \bibnamefont {Lu}},\
  }\bibfield  {title} {\enquote {\bibinfo {title} {Coordinate descent full
  configuration interaction},}\ }\href {\doibase 10.1021/acs.jctc.9b00138}
  {\bibfield  {journal} {\bibinfo  {journal} {J. of Chem. Th. and Comp.}\
  }\textbf {\bibinfo {volume} {15}},\ \bibinfo {pages} {3558--3569} (\bibinfo
  {year} {2019})},\ \bibinfo {note} {pMID: 31042383},\ \Eprint
  {http://arxiv.org/abs/https://doi.org/10.1021/acs.jctc.9b00138}
  {https://doi.org/10.1021/acs.jctc.9b00138} \BibitemShut {NoStop}%
\bibitem [{\citenamefont {Karbach}\ \emph {et~al.}(1998)\citenamefont
  {Karbach}, \citenamefont {Hu},\ and\ \citenamefont
  {M\"{u}ller}}]{Karbach:1998}%
  \BibitemOpen
  \bibfield  {author} {\bibinfo {author} {\bibfnamefont {Michael}\ \bibnamefont
  {Karbach}}, \bibinfo {author} {\bibfnamefont {Kun}\ \bibnamefont {Hu}}, \
  and\ \bibinfo {author} {\bibfnamefont {Gerhard}\ \bibnamefont {M\"{u}ller}},\
  }\bibfield  {title} {\enquote {\bibinfo {title} {Introduction to the bethe
  ansatz ii},}\ }\href {\doibase 10.1063/1.168740} {\bibfield  {journal}
  {\bibinfo  {journal} {Computers in Physics}\ }\textbf {\bibinfo {volume}
  {12}},\ \bibinfo {pages} {565--573} (\bibinfo {year} {1998})},\ \Eprint
  {http://arxiv.org/abs/https://aip.scitation.org/doi/pdf/10.1063/1.168740}
  {https://aip.scitation.org/doi/pdf/10.1063/1.168740} \BibitemShut {NoStop}%
\bibitem [{\citenamefont {Nahum}\ \emph {et~al.}(2015)\citenamefont {Nahum},
  \citenamefont {Serna}, \citenamefont {Chalker}, \citenamefont {Ortu\~no},\
  and\ \citenamefont {Somoza}}]{PhysRevLett.115.267203}%
  \BibitemOpen
  \bibfield  {author} {\bibinfo {author} {\bibfnamefont {Adam}\ \bibnamefont
  {Nahum}}, \bibinfo {author} {\bibfnamefont {P.}~\bibnamefont {Serna}},
  \bibinfo {author} {\bibfnamefont {J.~T.}\ \bibnamefont {Chalker}}, \bibinfo
  {author} {\bibfnamefont {M.}~\bibnamefont {Ortu\~no}}, \ and\ \bibinfo
  {author} {\bibfnamefont {A.~M.}\ \bibnamefont {Somoza}},\ }\bibfield  {title}
  {\enquote {\bibinfo {title} {Emergent so(5) symmetry at the n\'eel to
  valence-bond-solid transition},}\ }\href {\doibase
  10.1103/PhysRevLett.115.267203} {\bibfield  {journal} {\bibinfo  {journal}
  {Phys. Rev. Lett.}\ }\textbf {\bibinfo {volume} {115}},\ \bibinfo {pages}
  {267203} (\bibinfo {year} {2015})}\BibitemShut {NoStop}%
\bibitem [{\citenamefont {Poland}\ \emph {et~al.}(2019)\citenamefont {Poland},
  \citenamefont {Rychkov},\ and\ \citenamefont {Vichi}}]{Poland:2018epd}%
  \BibitemOpen
  \bibfield  {author} {\bibinfo {author} {\bibfnamefont {David}\ \bibnamefont
  {Poland}}, \bibinfo {author} {\bibfnamefont {Slava}\ \bibnamefont {Rychkov}},
  \ and\ \bibinfo {author} {\bibfnamefont {Alessandro}\ \bibnamefont {Vichi}},\
  }\bibfield  {title} {\enquote {\bibinfo {title} {{The Conformal Bootstrap:
  Theory, Numerical Techniques, and Applications}},}\ }\href {\doibase
  10.1103/RevModPhys.91.015002} {\bibfield  {journal} {\bibinfo  {journal}
  {Rev. Mod. Phys.}\ }\textbf {\bibinfo {volume} {91}},\ \bibinfo {pages}
  {015002} (\bibinfo {year} {2019})},\ \Eprint
  {http://arxiv.org/abs/1805.04405} {arXiv:1805.04405 [hep-th]} \BibitemShut
  {NoStop}%
\bibitem [{\citenamefont {Towns}\ \emph {et~al.}(2014)\citenamefont {Towns},
  \citenamefont {Cockerill}, \citenamefont {Dahan}, \citenamefont {Foster},
  \citenamefont {Gaither}, \citenamefont {Grimshaw}, \citenamefont {Hazlewood},
  \citenamefont {Lathrop}, \citenamefont {Lifka}, \citenamefont {Peterson},
  \citenamefont {Roskies}, \citenamefont {Scott},\ and\ \citenamefont
  {Wilkins-Diehr}}]{xsede}%
  \BibitemOpen
  \bibfield  {author} {\bibinfo {author} {\bibfnamefont {J.}~\bibnamefont
  {Towns}}, \bibinfo {author} {\bibfnamefont {T.}~\bibnamefont {Cockerill}},
  \bibinfo {author} {\bibfnamefont {M.}~\bibnamefont {Dahan}}, \bibinfo
  {author} {\bibfnamefont {I.}~\bibnamefont {Foster}}, \bibinfo {author}
  {\bibfnamefont {K.}~\bibnamefont {Gaither}}, \bibinfo {author} {\bibfnamefont
  {A.}~\bibnamefont {Grimshaw}}, \bibinfo {author} {\bibfnamefont
  {V.}~\bibnamefont {Hazlewood}}, \bibinfo {author} {\bibfnamefont
  {S.}~\bibnamefont {Lathrop}}, \bibinfo {author} {\bibfnamefont
  {D.}~\bibnamefont {Lifka}}, \bibinfo {author} {\bibfnamefont {G.~D.}\
  \bibnamefont {Peterson}}, \bibinfo {author} {\bibfnamefont {R.}~\bibnamefont
  {Roskies}}, \bibinfo {author} {\bibfnamefont {J.~R.}\ \bibnamefont {Scott}},
  \ and\ \bibinfo {author} {\bibfnamefont {N.}~\bibnamefont {Wilkins-Diehr}},\
  }\bibfield  {title} {\enquote {\bibinfo {title} {Xsede: Accelerating
  scientific discovery},}\ }\href {\doibase 10.1109/MCSE.2014.80} {\bibfield
  {journal} {\bibinfo  {journal} {Computing in Science \& Engineering}\
  }\textbf {\bibinfo {volume} {16}},\ \bibinfo {pages} {62--74} (\bibinfo
  {year} {2014})}\BibitemShut {NoStop}%
\end{thebibliography}%

\appendix
\section{Degeneracy with lattice fermions}\label{app:degeneracy}

In this appendix we discuss a curious symmetry of a class of lattice fermion Hamiltonians in one dimension that is generalizable to higher dimensions, which makes the degeneracy of all energies, including that of the ground state, to be an even number when the lattice size is a multiple of four. We can formulate this result as the following theorem: 

\begin{thm*}
Degeneracy of all energy eigenvalues of a spin-half fermion system on a periodic lattice with translation symmetry, spin and charge symmetries, spin-charge flip symmetry and parity symmetry, will be an even number when the lattice size is a multiple of four.
\end{thm*}

\begin{proof}
From the assumptions in the theorem, the lattice Hamiltonian $H$ commutes with $T_a$
\begin{align}
  T_a = \exp\Big(-ia\sum_{k}k (c_{k\uparrow}^\dagger c_{k\uparrow} + c_{k\downarrow}^\dagger c_{k\downarrow})\Big),
\end{align}
where
\begin{align}
  c_{k\alpha} = \frac{1}{\sqrt{L}} \sum_j c_{j\alpha} \e^{ikaj},~ k = \frac{2\pi n}{aL},~ n = 1, \cdots, L,
\end{align}
and the spin and charge $SU(2)$ generators $S_i, Q_i, i=1,2,3$,
\begin{align}
  S_1 &= \frac{1}{2}\sum_{j} c_{j\downarrow}^\dagger c_{j\uparrow} + c_{j\uparrow}^\dagger c_{j\downarrow}, \nonumber\\
  S_2 &= \frac{1}{2}\sum_{j} i(c_{j\downarrow}^\dagger c_{j\uparrow} - c_{j\uparrow}^\dagger c_{j\downarrow}), \nonumber\\
  S_3 &= \frac{1}{2}\sum_{j} c_{j\uparrow}^\dagger c_{j\uparrow} - c_{j\downarrow}^\dagger c_{j\downarrow}, \nonumber\\
  Q_1 &= \frac{1}{2}\sum_{j} (-1)^j (c_{j\downarrow}^\dagger c_{j\uparrow}^\dagger + c_{j\uparrow} c_{j\downarrow}), \nonumber\\
  Q_2 &= \frac{1}{2}\sum_{j}(-1)^j i(c_{j\downarrow}^\dagger c_{j\uparrow}^\dagger - c_{j\uparrow} c_{j\downarrow}), \nonumber\\
  Q_3 &= \frac{1}{2}\sum_{j} c_{j\uparrow} c_{j\uparrow}^\dagger - c_{j\downarrow}^\dagger c_{j\downarrow}.
\end{align}
Further, $H$ is also assumed to commute with the spin-charge flip operator $C_\uparrow = C_\uparrow^\dagger$ and parity operator $P = P^\dagger$. These operators can be conveniently represented on fermion creation and annihilation operators in the following way,
\begin{align}
    C_\uparrow c_{i\uparrow}C_\uparrow = (-1)^ic_{i\uparrow}^\dagger, &\qquad
    C_\uparrow c_{k\uparrow} C_\uparrow = c_{\pi/a-k\uparrow}^\dagger, \nonumber \\
    Pc_{i\alpha} P = c_{(L+1-i)\alpha}, &\qquad 
    Pc_{k\alpha} P = c_{-k\alpha}. 
\end{align}
Note that the particle number $N$ is not an independent operator since $N = L - 2Q_3$. An example lattice Hamiltonian with the above symmetries is $H_J$ defined in \cref{eq:HJ}.

Since $H, T_a, S^2, S_3, Q^2$ and $Q_3$ commutes with each other, we can label the energy eigenstates by $|E, k, l_s, s_3, l_q, q_3, \alpha\>$, where $E, \e^{iak}, l_s(l_s+1), s_3, l_q(l_q+1), q_3$ are eigenvalues of $H, T_a, S^2, S_3, Q^2, Q_3$, and $\alpha$ denotes possible additional quantum numbers. Since $P$ satisfies the relationships $PT_a P = T_a^\dagger$, $PS_i P = S_i$, and $PQ_i P = Q_i$, we have
\begin{align}
    P |E, k, l_s, s_z, l_q, q_z, \alpha\> \propto |E, -k, l_s, s_z, l_q, q_z, \alpha\>.
\end{align}
This implies that the pair of states with $k=\pm k_0$ but all other quantum numbers being the same will be degenerate as long as 
$k_0\neq 0,\pi/a$. This is easily understandable since a state with a fixed lattice momentum will have a partner with a negative momentum and both will have the same energy as long as parity is a symmetry of the theory.

Interestingly, we will now show that there are additional pairs of degenerate states due to the $C_\uparrow$ operator, which has the following properties:
\begin{align}
C_\uparrow T_a C_\uparrow &= \e^{-ia\sum_{k}k (C_\uparrow c_{k\uparrow}^\dagger C_\uparrow C_\uparrow c_{k\uparrow}C_\uparrow + c_{k\downarrow}^\dagger c_{k\downarrow})} \nonumber\\
& = \e^{-ia\sum_{k}k (c_{\pi/a-k\uparrow} c_{\pi/a-k\uparrow}^\dagger + c_{k\downarrow}^\dagger c_{k\downarrow})} \nonumber\\
& = \e^{-i\sum_{k}(\pi - ak) c_{k\uparrow} c_{k\uparrow}^\dagger + ak c_{k\downarrow}^\dagger c_{k\downarrow}} \nonumber\\
& = \e^{-i\sum_{k}(ak - \pi) (c_{k\uparrow}^\dagger c_{k\uparrow}-1) + ak c_{k\downarrow}^\dagger c_{k\downarrow}} \nonumber\\
& = (-1)^{S_3 - Q_3 + L/2 +1}T_a,
\end{align}
and
\begin{align}
C_\uparrow S_iC_\uparrow = Q_i, \quad C_\uparrow Q_iC_\uparrow = S_i.
\end{align}
Using these relations it is easy to show that
\begin{align}
& \quad C_\uparrow|E, k, l_s, s_3, l_q, q_3, \alpha\> \nonumber\\
& \propto |E, k + (s_3 - q_3 + L/2 +1)\frac{\pi}{a}, l_q, q_3, l_s, s_3, \alpha\>.
\label{eq:Cdeg}
\end{align}
This relation shows that additional pairs of state can have the same energy. For example when $l_s \neq l_q$ or $s_3 \neq q_3$, the two states $|E, k, l_s, s_3, l_q, q_3, \alpha\>$ and $|E, k + (s_3 - q_3 + L/2 +1)\frac{\pi}{a}, l_q, q_3, l_s, s_3, \alpha\>$ are different but have the same energy irrespective of the value of $k$. Thus, the only situation where an energy eigenstate could remain non-degenerate is when $k=0$ or $\pi/a$, $l_s = l_q = j$ and $s_3 = q_3 = m$. In this case \cref{eq:Cdeg} mixes the states $|E, k, j, m, j, m, \alpha\>$ and 
$|E, (k + L/2 + 1)\pi/a, j, m, j, m, \alpha\>$. Indeed when $L=4n-2$ these two states are identical and $C_\uparrow$ cannot be used to pair degenerate states. However, when $L=4n$, $C_\uparrow$ mixes  $|E, 0, j, m, j, m, \alpha\>$ with $|E, \pi/a, j, m, j, m, \alpha\>$. This means the whole spectrum has an even degeneracy when $L=4n$. 
\end{proof}

\section{Lagrangian of the continuum model}
\label{app:lagrangian}

In this appendix we construct the Euclidean Lagrangian for the continuum Hamiltonian given in \cref{eq:contmodel}. For this purpose we define two flavors of two-component Grassmann variables $\chi_{\alpha}$ and $\bar{\chi}_{\alpha}$, using the following map to Fock space operators,
\begin{align}
  \chi_\alpha:=
  \begin{pmatrix}
    \psi_{\alpha,L} \\
    \psi_{\alpha,R}
  \end{pmatrix} ,\quad \bar{\chi}_\alpha:=
   (\psi^\dagger_{\alpha,L}, \ \psi^\dagger_{\alpha,R}) \gamma^1
\end{align}
and choose $\gamma^1 = \sigma^1$, $\gamma^2 = -\sigma^2$, $\gamma^3 = i\gamma^1\gamma^2 = \sigma^3$, where $\sigma^i$ are the Pauli matrices. The free Hamiltonian theory is then mapped to the theory with the Lagrangian
\begin{align}
{\cal L}_0 \ =\ \bar{\chi}_\alpha \gamma^\mu \partial_\mu \chi_{\alpha}.
\end{align}
In order to construct the interaction terms of the Lagrangian, we normal order the interaction terms in \cref{eq:contmodel}. A little bit of algebra gives
\thinmuskip=0mu
\medmuskip=0mu 
\thickmuskip=0mu 
\begin{align}
:\mathcal{J}_{sL}^i\mathcal{J}_{sR}^i:~ &= \frac{1}{2}(\bar{\chi}_{\alpha} \gamma^3\chi_\alpha)^2- \frac{1}{2}(\bar{\chi}_{\alpha} \chi_\alpha)^2 - \frac{1}{4}(\bar{\chi}_{\alpha} \gamma^\mu\chi_\alpha)^2, 
\label{eq:HLmap1}
\\
:\mathcal{J}_{cL}^i\mathcal{J}_{cR}^i:~ &= \frac{1}{2}(\bar{\chi}_{\alpha} \gamma^3\chi_\alpha)^2 + \frac{1}{2}(\bar{\chi}_{\alpha} \chi_\alpha)^2 + \frac{1}{4}(\bar{\chi}_{\alpha} \gamma^\mu\chi_\alpha)^2. 
\label{eq:HLmap2}
\end{align}
\thinmuskip=3mu
\medmuskip=4mu 
\thickmuskip=5mu 
Therefore the interaction Hamiltonian in \cref{eq:contmodel}
\begin{align}
H_{\rm int} \ =\ \int \d x\  \lambda_s\mathcal{J}_{sL}^i\mathcal{J}_{sR}^i + \lambda_c\mathcal{J}_{cL}^i\mathcal{J}_{cR}^i
\end{align}
is mapped to the interaction Lagrangian 
\thinmuskip=1mu
\medmuskip=2mu 
\thickmuskip=3mu 
\begin{align}
    \mathcal{L}_{\rm int} = \frac{g_1}{2} (\bar{\chi}_{\alpha} \gamma^3\chi_\alpha)^2 + \frac{g_2}{2} (\bar{\chi}_{\alpha} \chi^\alpha)^2 +  \frac{g_3}{2}(\bar{\chi}_{\alpha} \gamma^\mu\chi^\alpha)^2,
\label{eq:intL}
\end{align}
\thinmuskip=3mu
\medmuskip=4mu 
\thickmuskip=5mu 
where $g_1=\lambda_s+\lambda_c = 4\varepsilon$ and $g_2=2g_3=-\lambda_s+\lambda_c=U/J$. Therefore our lattice model contains all the three allowed four-fermion couplings in a relativistic two-flavor fermion model: the chiral-mass GN coupling $g_1$, the normal GN coupling $g_2$ and the Thirring coupling $g_3$. Our model is a lattice regularized chiral-mass GN model, while the Hubbard coupling $U$ introduces a linear combination of GN and Thirring terms.

A natural question now is whether we can find three \emph{independent} current couplings, and three \emph{independent} lattice couplings, which map to the three GN-Thirring couplings? Note that we know the massless Thirring model is conformal and invariant under $\big(SU(2)_{sL}\times SU(2)_{sR}\big) \times \big(U(1)_{cL}\times U(1)_{cR}\big)$. Therefore in order to introduce a Thirring coupling, we must break the charge $SU(2)_c$ symmetry down to $U(1)_c$, while preserving all the other chiral symmetries. Clearly the term $\mathcal{J}_{cL}^3\mathcal{J}_{cR}^3$ is what we need, and indeed maps to the Thirring coupling in the Lagrangian language,
\begin{align}
:\mathcal{J}_{cL}^3\mathcal{J}_{cR}^3: \ &=\  \frac{1}{4}(\bar{\chi}_{\alpha} \gamma^\mu\chi_\alpha)^2.
\label{eq:HLmap3}
\end{align}
When analyzing the RG flow it turns out to be more convenient if we introduce the three independent current couplings as follows,
\begin{align}
\mathcal{L}_{\rm int} &= \lambda_s:\mathcal{J}_{sL}^i\mathcal{J}_{sR}^i: + \lambda_c:\mathcal{J}_{cL}^i\mathcal{J}_{cR}^i: \nonumber\\
&\quad +\tilde{\lambda}_{c}:-\mathcal{J}_{cL}^1\mathcal{J}_{cR}^1 - \mathcal{J}_{cL}^2\mathcal{J}_{cR}^2 + \mathcal{J}_{cL}^3\mathcal{J}_{cR}^3:,
\end{align}
because those are the eigen directions of the RG flow. The beta functions of the three couplings are given by,
\begin{align}
\frac{\d\lambda_s}{\d\log\mu} &= -\frac{\lambda_s^2}{2\pi}, \nonumber \\
\frac{\d\lambda_{c}}{\d\log\mu} &= -\frac{ \lambda_{c} - \tilde{\lambda}_{c}}{2\pi}\lambda_{c},
\nonumber \\
\frac{\d\tilde{\lambda}_{c}}{\d\log\mu} &= \frac{ \lambda_{c} - \tilde{\lambda}_{c}}{2\pi}\tilde{\lambda}_{c}.
\end{align}
The RG flow diagram is shown in \cref{fig:flow-3d}.

\begin{figure}[htb]
\includegraphics[width=0.4\textwidth]{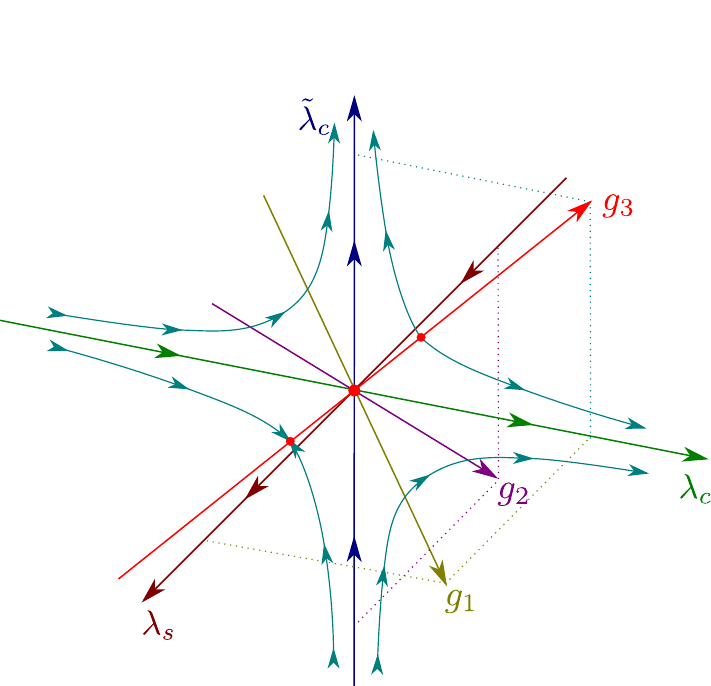}
  \caption{RG flow in the three dimensional four-fermion coupling space of two-flavor massless Dirac fermions in 2d invariant under $U(2) \times \mathbb{Z}_2^\chi$ flavor symmetries. All the axes are eigen directions of the RG flow. In addition, the red line of the Thirring coupling $g_3$ is a line of fixed points.}
  \label{fig:flow-3d}
\end{figure}

Using the relations given by \cref{eq:HLmap1,eq:HLmap2,eq:HLmap3} we see that these couplings are related to $g_1,g_2$ and $g_3$ via the linear transformation
\begin{align}
    \begin{pmatrix}
    g_1 \\
    g_2 \\
    2g_3
    \end{pmatrix}
    =
    \begin{pmatrix}
    1 & 1 & -1 \\
    -1 & 1 & -1 \\
    -1 & 1 & 1
    \end{pmatrix}
    \begin{pmatrix}
    \lambda_s \\
    \lambda_{c} \\
    \tilde{\lambda}_{c} \\
    \end{pmatrix}.
\end{align}
Notice that the pure Thirring coupling $g_3$ is obtained when $\lambda_s = \lambda_c - \tilde{\lambda}_c = 0$. This direction is special because the beta function vanishes for all the couplings and we obtain a line of fixed points. In \cref{tab:coupling-symmetries} we summarize the symmetries that are preserved by the various couplings. 
\begin{table}[htb]
  \centering
  \renewcommand{\arraystretch}{1.2}
  \setlength{\tabcolsep}{4pt}
  \begin{tabular}{c|c}
    \TopRule
     Coupling & Symmetry \\ \MidRule
    $\lambda_s$ & $\big(SU(2)_{s(L=R)}\times SU(2)_{cL}\times SU(2)_{cR}\big)/\mathbb{Z}_2$ \\
    $\lambda_c $ & $\big( SU(2)_{sL}\times SU(2)_{sR}\times SU(2)_{c(L=R)}\big) /\mathbb{Z}_2$ \\
    $\tilde\lambda_c$ & $\big( SU(2)_{sL}\times SU(2)_{sR}\times SU(2)_{c(L=R)}\big) /\mathbb{Z}_2$ \\
    $g_1$ & $\big(SU(2)_{s(L=R)}\times SU(2)_{c(L=R)}\big)/\mathbb{Z}_2 \times \mathbb{Z}_2^\chi$\\
    $g_2$ & $\big(SU(2)_{s(L=R)}\times SU(2)_{c(L=R)}\big)/\mathbb{Z}_2 \times 
    \mathbb{Z}_2^\chi$\\
    $g_3$ & $\big(SU(2)_{sL}\times SU(2)_{sR}\big) \times \big(U(1)_{cL}\times U(1)_{cR}\big)$\\ 
    \BotRule
  \end{tabular}
  \caption{Various interaction couplings and the subgroup of the full chiral symmetry group of the free theory that they preserve. The free theory is invariant under the chiral symmetry group $O_L(4) \times O_R(L)$ while a generic point in the three dimensional coupling constant space is invariant under $\big(SU_{s(L=R)}(2)\times U_{c(L=R)}(1)\big)/\mathbb{Z}_2$.}
  \label{tab:coupling-symmetries}
\end{table}

In our lattice model we introduced two independent couplings $H_J^\varepsilon$ and $H_U$. In order to explore the full three dimensional space discussed above we will need one more independent coupling. Using the same symmetry argument above, we need a lattice interaction which breaks $SU(2)_c$ down to $U(1)_c$ while preserving particle hole symmetry. The most straightforward way to do this is to include the interaction
\begin{align}
H_V =\ V \sum_{\langle i j\rangle} (n_i - 1)(n_j - 1)
\end{align}
in the lattice Hamiltonian, where $n_i = n_{i\uparrow} + n_{j\downarrow}$ is the total fermion number operator at the site $i$. At the tree level in the continuum limit we can show that 
\thinmuskip=1mu
\medmuskip=2mu 
\thickmuskip=3mu 
\begin{align}
    H_V^\text{cont} &= \frac{1}{2}aV \int \d x (\mathcal{J}_{sL}^i\mathcal{J}_{sR}^i - \mathcal{J}_{cL}^i\mathcal{J}_{cR}^i + 4\mathcal{J}_{cL}^3\mathcal{J}_{cR}^3),
\end{align}
\thinmuskip=3mu
\medmuskip=4mu 
\thickmuskip=5mu 
which maps to the Lagrangian
\begin{align}
{\cal L}_V = aV \ \bigg( \frac{1}{2}(\bar{\chi}_{\alpha} \gamma^\mu\chi^\alpha)^2 - (\bar{\chi}_{\alpha}\chi^\alpha)^2\bigg).
\end{align}
Therefore we see that $H_U+H_V$ gives the Thirring coupling, while $H_U-H_V$ gives the usual GN model. It's interesting to observe that when both $U$ and $V$ are positive, $H_U$ favors the spin sector while $H_V$ favors the charge sector, and the frustration between them gives the conformal Thirring model. Furthermore, $H_V$ can also be included in the meron-cluster algorithm for a range couplings.

\end{document}